\documentclass[11pt,a4paper]{article}
\usepackage[a4paper]{geometry}
\usepackage{amssymb,latexsym,amsmath,amsfonts,amsthm}
\usepackage{graphicx}
\usepackage{epstopdf}
\usepackage{float,amsmath,amssymb,mathrsfs,bm,multirow,graphics}
\usepackage{tikz}
\usepackage{pgflibraryshapes}
\usetikzlibrary{arrows,decorations.markings}
\usepackage{subfigure}
\usepackage{overpic}
\usepackage{fullpage}
\usepackage{comment,verbatim}
\usepackage{hyperref}
\usepackage{color}
\usepackage{mathrsfs}
\usepackage[title,titletoc]{appendix}
\usepackage{ulem}
\usepackage{enumitem}

\DeclareMathOperator*{\Res}{Res}

\renewcommand{\Re}{\mathrm{Re}\,}
\renewcommand{\Im}{\mathrm{Im}\,}
\newcommand{\Ai}{\mathrm{Ai}}

\renewcommand{\vec}{\mathbf}

\newcommand{\q}{{\mathsf{q}}}

\newcommand{\msf}{\mathsf}

\newcommand{\re}{\Re}

\newcommand{\Boh}{\mathcal{O}}
\newcommand{\sgn}{\mathrm{sgn}}

\def\blue{\color{blue}}

\newtheorem{theorem}{Theorem}[section]
\newtheorem{lemma}[theorem]{Lemma}

\newtheorem{corollary}[theorem]{Corollary}

\theoremstyle{definition}

\theoremstyle{remark}

\newtheorem{remark}[theorem]{Remark}

\numberwithin{equation}{section}

\begin{document}
\title{The multiplicative constant in  asymptotics of higher-order analogues of the Tracy-Widom distribution }

\author{Dan Dai\footnotemark[1],~ Wen-Gao Long\footnotemark[2],~ Shuai-Xia Xu\footnotemark[3],~ Lu-Ming Yao\footnotemark[4],~ Lun Zhang\footnotemark[5]}

\renewcommand{\thefootnote}{\fnsymbol{footnote}}
\footnotetext[1]{Department of Mathematics, City University of Hong Kong, Tat Chee
Avenue, Kowloon, Hong Kong.\\ E-mail: \texttt{dandai@cityu.edu.hk}}
\footnotetext[2]{School of Mathematics and Statistics, Hunan University of Science and Technology, Xiangtan, Hunan, China.
E-mail: \texttt{longwg@hnust.edu.cn}} 
\footnotetext[3]{Institut Franco-Chinois de l'Energie Nucl\'{e}aire, Sun Yat-sen University,
	Guangzhou 510275, China. E-mail: \texttt{xushx3@mail.sysu.edu.cn}}
	\footnotetext[4]{Institute for Advanced Study, Shenzhen University, Shenzhen 518060, China. E-mail: \texttt{lumingyao@szu.edu.cn}}
\footnotetext[5] {School of Mathematical Sciences, Fudan University, Shanghai 200433, China. E-mail: \texttt{lunzhang@fudan.edu.cn }}

\date{\today}

\maketitle

\begin{abstract}
In this paper, we are concerned with higher-order analogues of the Tracy-Widom distribution, which describe the eigenvalue distributions in unitary random matrix models near critical edge points. The associated kernels are constructed by functions related to the even members of the Painlev\'{e} I hierarchy $\mathrm{P_{I}^{2k}}, k\in\mathbb{N}^{+}$, and are regarded as higher-order analogues of the Airy kernel. We present a novel approach to establish the multiplicative constant in the large gap asymptotics of the distribution, resolving an open problem in the work of Clayes, Its and Krasovsky. An important new feature of the expression is the involvement of an integral of the Hamiltonian associated with a special, real, pole-free solution for $\mathrm{P_{I}^{2k}}$. In addition, we show that the total integral of the Hamiltonian vanishes for all $k$, and establish a transition from the higher-order Tracy-Widom distribution to the classical one in the asymptotic regime. Our approach can also be adapted to calculate similar critical constants in other problems arising from mathematical physics.



\end{abstract}

\tableofcontents

\section{Introduction and statement of results}

Let $\mathcal{M}_n$ be the space of $n \times n$ Hermitian matrices $M=(M_{ij})_{1\leq i,j \leq n}$, equipped with the probability measure
\begin{equation}\label{eq:probmeasure}
\frac{1}{\mathcal{Z}_n}e^{- n\textrm{tr}\, V(M)}dM.
\end{equation}
Here,
\begin{equation}
dM=\prod_{i=1}^{n}dM_{ii}\prod_{i=1}^{n-1}\prod_{j=i+1}^{n}d\Re M_{ij}d\Im M_{ij}
\end{equation}
denotes the Lebesgue measure for Hermitian matrices,
\begin{equation}
 \mathcal{Z}_n=\int_{\mathcal{M}_n} e^{- n\textrm{tr}V(M)}dM
\end{equation}
is the normalization constant, and the potential $V$ is a real analytic function over $\mathbb{R}$ satisfying
$$
\lim_{|x|\to\infty}\frac{V(x)}{\log(1+x^2)}=+\infty.
$$

Due to the unitary invariant feature of \eqref{eq:probmeasure}, it is well-known that (cf. \cite{Deift1999,Meh}) the eigenvalues of $M$ form a determinantal point process characterized by the correlation kernel
\begin{equation}\label{def:Kn}
K_n(x,y):=e^{-\frac{n}{2}(V(x)+V(y))}\sum_{i=0}^{n-1}p_i(x)p_i(y),
\end{equation}
where $p_j(x)=\kappa_jx^j+\cdots$, $\kappa_j>0$, are orthonormal polynomials associated with the weight function $e^{-nV(x)}$ over $\mathbb{R}$, {\it i.e.},
\begin{equation}
\int_\mathbb{R}p_i(x)p_j(x)e^{-nV(x)}dx=\delta_{i,j}.
\end{equation}

Various eigenvalue statistics are encoded in the correlation kernel $K_n$. The assumption on $V$ ensures that the limiting mean eigenvalue distribution $\mu_V$ admits a density function $\rho_V$, which can be retrieved from the relation
\begin{equation}
\lim_{n\to\infty}\frac 1n K_n(x,x)=\rho_V(x);
\end{equation}
see \cite{Deift1999,DKM}. Moreover, $\mu_V$ is the unique equilibrium measure that minimizes the logarithmic energy functional \cite{Saff}
\begin{equation}
I_{V}(\mu)=\iint \log \frac{1}{ |x-y|}d\mu(x)d\mu(y)+\int V(x)d\mu(x),
\end{equation}
among all the probability measure $\mu$ on $\mathbb{R}$. The remarkable result in \cite{DKM} shows that $\textrm{supp}(\mu_V)$ consists of a finite union of intervals, and
\begin{equation}
\rho_V(x)=\frac{1}{\pi}\sqrt{\mathcal{Q}_V^-(x)},
\end{equation}
where $\mathcal{Q}_V$ is a real analytic function depending on $V$ and $\mathcal{Q}_V^-$ denotes its negative part.

In contrast to the limiting mean distribution, the local statistics of $M$ are \textit{universal} as the matrix size tends to infinity in the sense that it depends only on the local behaviors of the density function $\rho_V$; see \cite{Bleher99,DG07,DKMVZ99,Pastur97}. This particularly implies that the microscopic limits of $K_n$ will converge to several definite limiting kernels. For example, since generically $\rho_V$ vanishes as a squared root at the endpoints of $\textrm{supp}(\mu_V)$ \cite{Kui00}, one encounters the so-called soft-edge universality. More precisely, let $b_V$ be the rightmost endpoint of $\textrm{supp}(\mu_V)$, there exists a constant $c_V$ such that
\begin{equation}
\lim_{n\to \infty}\frac{1}{c_V n^{2/3}}K_n\left(b_V+\frac{x}{c_Vn^{2/3}}, b_V+ \frac{y}{c_Vn^{2/3}}\right)
=K_{\Ai}(x,y),
\end{equation}
uniformly for $x$ and $y$ in any compact subset of $\mathbb{R}$, where
\begin{equation}\label{def:Aikernel}
K_{\Ai}(x,y)=\frac{\Ai(x)\Ai'(y)-\Ai'(x)\Ai(y)}{x-y}
\end{equation}
is the Airy kernel with $\Ai$ being the Airy function. This, together with the determinantal structure for the eigenvalues of $M$,  also implies that (cf. \cite{DKMVZ99})
\begin{equation}
\lim_{n\to\infty}\textrm{Prob}\left(c_Vn^{\frac23}(\lambda_n-b_V)<s\right)=\det (I-\mathbb{K}^{\Ai}_{s}),
\end{equation}
where $\lambda_n$ is the largest eigenvalue of $M$ and $\mathbb{K}^{\Ai}_{s}$ stands for the trace-class operator acting on $L^2(s,\infty)$ with the Airy kernel \eqref{def:Aikernel}.

The distribution $F_{\mathrm{TW}}(s):=\det (I-\mathbb{K}^{\Ai}_{s})$ is the celebrated Tracy-Widom distribution, whose appearance in a variety of physical, combinatorial and probabilistic models highlights the underlying connections among seemingly different areas; cf. \cite{CDI22}. An interesting result established in \cite{TW94} shows that the Tracy-Widom distribution admits the following integral representation:
\begin{equation}\label{eq:TWformula}
F_{\mathrm{TW}}(s)=\exp\left(-\int_s^{\infty}(x-s)q(x)^2 d x\right),
\end{equation}
where $q$ is the Hastings-McLeod solution \cite{HM} of the Painlev\'{e} II equation
\begin{equation}\label{def:PII}
q''(x)=xq(x)+2q(x)^3,
\end{equation}
fixed by the boundary condition
\begin{equation}\label{eq:HMasy}
q(x)\sim \left\{
           \begin{array}{ll}
             \Ai(x), & \hbox{$ \quad x\to +\infty$,} \\
             \sqrt{-x/2}, & \hbox{$\quad x\to -\infty$.}
           \end{array}
         \right.
\end{equation}
Combining \eqref{eq:TWformula} and \eqref{eq:HMasy}, the large gap asymptotics of the Airy-kernel determinant reads
\begin{equation} \label{eq:TW-large gap asy}
  \log F_{\mathrm{TW}}(s)=\log\det(I-\mathbb{K}^{\mathrm{Ai}}_{s}) =-\frac{|s|^3}{12}-\frac{1}{8}\log|s|+\chi^{(0)}+\mathcal{O}(|s|^{-\frac32}), \quad s\to -\infty,
\end{equation}
where the constant term $\chi^{(0)}$ is conjectured in \cite{TW94} to be
\begin{equation}\label{def:chi0}
\chi^{(0)}=\frac{1}{24}\log 2+\zeta'(-1)
\end{equation}
with $\zeta(z)$ being the Riemann zeta function. This conjecture was independently proved in \cite{BRD08,DIK2008}; see also \cite{KM24} for the Airy-kernel determinant on two large intervals. 

By tuning the potential $V$ in \eqref{eq:probmeasure}, it is possible that the limiting mean density $\rho_V$ vanishes faster than a square root at the rightmost endpoint $b_V$, which leads to singular cases. According to \cite{DKM,Kui00}, $\mathcal{Q}_V$ has a zero of order $4k+1$, $k \in \mathbb{N}=0,1,2,\ldots$, at $b_V$, which means that
\begin{equation}\label{eq:singrho}
\rho_V(x) \sim \textsf{c}_V |x-b_V|^{\frac{4k+1}{2}}, \qquad x \to b_V,
\end{equation}
for some positive constant $\textsf{c}_V$. At the singular edge points, the (multiple) scaling limit of the correlation kernel will be built from functions relevant to the Painlev\'{e} I hierarchy, as conjectured in the physical literature \cite{BB91,BMP90} and rigorously proved in \cite{CV07} for the case $k=1$. More precisely, to obtain the most general result, one could deform the singular potential to be
\begin{equation}
 \widehat{V}(x) = V(x)+\sum_{i=0}^{2k-1}T_iV_i(x)
\end{equation}
for some functions $V_i$, where $T_i$ are constants. In the multiple scaling regime, {\it i.e.}, $n\to \infty$ and simultaneously $T_i \to 0$ at an appropriate rate in $n$, it is expected that
\begin{equation}\label{eq:singularlim}
\lim_{n\to \infty}\frac{1}{c_V n^{2/(4k+3)}}K_n\left(b_V+\frac{u}{c_Vn^{2/(4k+3)}}, b_V + \frac{v}{c_Vn^{2/(4k+3)}}\right)=K^{(k)}(u,v),
\end{equation}
where the limiting kernel $K^{(k)}(u,v)=K^{(k)}(u,v;x,t_1,\ldots,t_{2k-1})$
depends on the real parameters $x,t_1, \ldots, t_{2k-1}$ related to the scaling of $T_0,\ldots,T_{2k-1}$. Since the function $K^{(k)}$ is given explicitly in terms of solutions of the Lax pair associated with a distinguished solution of $2k$-th member of the first Painlev\'{e} hierarchy (see \eqref{def:limKer} below for the precise definition), we call $K^{(k)}$ the $\mathrm{P_{I}^{2k}}$ kernel.

As a consequence of \eqref{eq:singularlim}, it is readily seen that, as $n\to \infty$ and after proper scaling, the fluctuations of the largest eigenvalue of $M$ around the edge $b_V$ in the singular cases is governed by the higher-order analogues of the Tracy-Widom distribution $\det(I-\mathbb{K}_{s}^{(k)})$, where $\mathbb{K}_{s}^{(k)}$ is the trace-class operator acting on $L^2(s,\infty)$ with the $\mathrm{P_{I}^{2k}}$ kernel \eqref{eq:singularlim}. A natural question then is to explore the integrable structure and large gap asymptotics of this determinant, following the classical results \eqref{eq:HMasy} and \eqref{eq:TW-large gap asy} of the Tracy-Widom distribution. The investigation of this direction has been initiated by Claeys, Its and Krasovsky in \cite{CIK10}. It comes out that $\frac{d}{ds}\det(I-\mathbb{K}_{s}^{(k)})$  can be explicitly expressed in terms of a special smooth solution to the Painlev\'{e} II hierarchy, which generalizes the Tracy-Widom formula \eqref{eq:TWformula}. 
In addition, as $s\to-\infty$, it follows from \cite[Theorem 1.5]{CIK10} that 
\begin{equation}\label{eq:largegapCIK}
\begin{split}
		\log \det(I-\mathbb{K}_{s}^{(k)}) &= \frac{1}{4}\alpha_k^2\frac{s^{4k+3}}{4k+3}+\frac{\alpha_k}{2(2k+2)}xs^{2k+2}+\sum_{m=2}^{4k+1}a_m |s|^{m} +\frac{x^2s}{4} \\
  & \quad -\frac{2k+1}{8}\log|s|+C^{(k)}+\mathcal{O}(|s|^{-2}),
   \end{split}
\end{equation}
where
\begin{equation} \label{eq-def-Ak} 
	\alpha_k:=
	\frac{2\Gamma\left(2k+\frac{3}{2}\right)}{\Gamma\left(2k+2\right)\Gamma\left(\frac{3}{2}\right)}
\end{equation}
with $\Gamma(z)$ being the gamma function, 
$a_m$ are functions of $x,t_1,\ldots,t_{2k-1}$ and vanish when $t_1=t_{2}=\ldots=t_{2k-1}=0$. The $s$-independent constant $C^{(k)}$ is unknown except for $k=0$. For our purpose, the variable and parameters used here are slightly different from those used in \cite{CIK10}, and they are related to each other through simple scalings 
\begin{equation}
\begin{split}
	&\det(I-\mathbb{K}_{s}^{(k)})(x,\ldots,t_j,\ldots,t_{2k-1})
	\\
	&\qquad =\det\left(I-\mathbb{K}_{2^{\frac{2}{4k+3}}s}^{(k)}\right)_{\mathrm{CIK}}\left(-2^{-\frac{4k+4}{4k+3}}x,\ldots,(-1)^{j-1}2^{1-\frac{2j+1}{4k+3}}t_j,\ldots,2^{\frac{2}{4k+3}}t_{2k-1}\right),
 \end{split}
\end{equation}
where $\det(I-\mathbb{K}_{s}^{(k)})_{\mathrm{CIK}}$ stands for the determinant in \cite{CIK10}. Thus, it follows that
\begin{equation}
   \det(I-\mathbb{K}_{s}^{(0)})=F_{\textrm{TW}}(2^{\frac23}s). 
\end{equation}

Clearly, there is no accurate description of the large gap asymptotics for the higher-order Tracy-Widom distribution without the explicit formula of $C^{(k)}$. Evaluation of the constant term in asymptotics of various distributions, the Painlev\'{e} equations or Hankel/Toeplitz determinants arising from mathematical physics, however, is a great challenge with a long history; cf. \cite{DIK13,ILP18,IP16,IP18,Kra09}. The constant term appearing in the asymptotics of sine determinant is also known as the Widom-Dyson constant, which was first obtained by Dyson \cite{Dyson76} based on an earlier work of Widom \cite{Widom}. A rigorous derivation of the Widom-Dyson constant was given independently in \cite{Ehr06,Kra04}; see also \cite{BK24,DIKZ07}. In addition to the Airy-kernel and sine-kernel determinants, the multiplicative constant problems were solved in \cite{DKV11,Ehr10} for the Bessel-kernel determinants, in \cite{CLM21a,CLM21b} for the determinants of some generalized Bessel kernels, in \cite{BI14,DXZ18,XD} for the determinants of Painlev\'{e}-type kernels and in \cite{BS24,Ch24} for two-dimensional point processes. 

It is the aim of this work to resolve the constant problem for the higher-order Tracy-Widom distribution by a novel approach. We present our main results in what follows.

%

\subsection{Main results}\label{sec:results}
We start with the introduction of the  Painlev\'{e} I hierarchy $\mathrm{P_{I}^m}$; cf. \cite{Gordoa,Kud97,Mugan,Shim04}. The $m$-th member
of $\mathrm{P_{I}^m}$ is a 
nonlinear ordinary differential equation of order $2m$
defined by 
\begin{equation}\label{def:PIm}
	x+\mathcal{L}_m(q)+\sum_{j=1}^{m-1}t_j\mathcal{L}_{j-1}(q)=0,\qquad t_1,\ldots,t_{m-1}\in\mathbb{R},
\end{equation}
where the operator $\mathcal{L}$ satisfies the Lenard-Magri recursion relation
\begin{align}\label{LOdef}
	\begin{cases}
		\frac{d}{dx}\mathcal{L}_{k+1}(q)=\bigg(\frac14 \frac{d^3}{dx^3}-2q\frac{d}{dx}-q_x \bigg)\mathcal{L}_{k}(q), \quad k=0,\ldots,m-1,\\
		\mathcal{L}_{0}(q)=-4q,\quad \mathcal{L}_{j}(0)=0, \quad j=1,\ldots,m.
	\end{cases}
\end{align}
If $m=1$ in \eqref{def:PIm}, one recovers the Painlev\'{e} I equation $q_{xx}=6q^2+x$, and the equation for $m=2$ reads
	\begin{equation}\label{eq-PI2-this-paper}
	q_{xxxx}=4x-40q^3+10q_{x}^2+20qq_{x}-16t_1q.
\footnote{After the rescalings 
	$U=-60^{2/7}q$, $X=60^{-1/7}x$, and $T=-4\times 60^{-3/7}t_1$,
	this equation reduces to 
	\begin{equation*}\label{eq-PI2-previous-paper}
		\frac{1}{240}U_{XXXX}+\frac{1}{24}(U_{X}^2+2UU_{XX})+
		\frac{1}{6}U^3+X-TU=0,
	\end{equation*}which is the $\mathrm{P_{I}^2}$ equation studied in several literature \cite{Claeys-2012,CV07, Grava-Kapaev-Klein-2015}.}
\end{equation}

The relevance to our work is the even member of the Painlev\'{e} I hierarchy. It follows from a series of works \cite{Claeys-2012,Claeys-Vanlessen-2007,Kap95} that there exists a real and pole-free solution $\mathsf{q}$ to each equation $\mathrm{P_{I}^{2k}}$ with the boundary condition
\begin{equation}\label{eq:asyq}
	\q(x)=\frac{1}{2}\alpha_k^{-\frac{1}{2k+1}}x^{\frac{1}{2k+1}}
 +\mathcal{O}\left(|x|^{-\frac{1}{2k+1}}\right),
 \qquad x\to\pm\infty,
\end{equation}
where $\alpha_k$ is given in \eqref{eq-def-Ak}. 
 For $\mathrm{P_{I}^{2}}$, $\q$ is called the tritronqu\'{e}e solution in \cite{Grava-Kapaev-Klein-2015}, and it is worth noting that this family of special functions is essential to describe the critical behaviors for the solutions of a large class of Hamiltonian PDEs \cite{Claeys-2012,Dub06,Dub08,Dub09}; see also \cite{CG09,CG12} for the studies in the context of Korteweg-de Vries equation and its hierarchy. It comes out the Hamiltonian associated with $\q$ will appear in the evaluation of the multiplicative constant $C^{(k)}$. 

\begin{theorem}
\label{Thm-asym-PI2k-determinant}
For $k=1,2,\ldots$, let $\mathbb{K}_{s}^{(k)}$ be the trace-class operator acting on $L^2(s,\infty)$ with the kernel $K^{(k)}(u,v)$  \eqref{eq:singularlim} and define
\begin{equation} \label{F-def}
F(s;x,t_1,\ldots,t_{2k-1}):=\log\det(I-\mathbb{K}_{s}^{(k)}).
\end{equation}
Assume that $t_{1}=t_{2}=\ldots=t_{2k-1}=0$, we have, as $s\to-\infty$,  
\begin{align}\label{eq-Fredholm-asym}
F(s;x)&=F(s;x,0,\ldots,0)
\nonumber \\
& =\frac{1}{4(4k+3)}\alpha_k^2s^{4k+3}+\frac{\alpha_k}{2(2k+2)}xs^{2k+2}+\frac{x^2s}{4}-\frac{1}{8}\log\left|\alpha_ks^{2k+1}+x\right|
\nonumber \\
		&~~~-I_{h}(x)+\frac{(2k+1)^2}{2(2k+2)(4k+3)}\alpha_k^{-\frac{1}{2k+1}}\cdot x^{\frac{4k+3}{2k+1}}+\frac{k\log(x^2+1)}{24(2k+1)} \nonumber\\
&~~~+\frac{\log(2k+1)}{24}+\frac{\log{\alpha_k}}{24(2k+1)}+\chi^{(0)}+\mathcal{O}(|s|^{-\epsilon_{0}}), 
\end{align}
uniformly valid for  
$
x\in\left[-c_{1}|s|^{2k+1},\alpha_k|s|^{2k+1}-c_{2}|s|^{\frac{2k}{3}+\epsilon}\right]
$ with $c_{1}$, $c_{2}$ being arbitrarily fixed real positive numbers and fixed $\epsilon\in(0,\frac{4k}{3}+1)$, where $\epsilon_{0}=\min\{\frac{1}{6},2\epsilon\}$. Here, the constants $\alpha_k$ and $\chi^{(0)}$ are defined in \eqref{eq-def-Ak} and \eqref{def:chi0}, respectively, 
\begin{equation}\label{def-J0-two-representation}
	I_{h}(x)
	:=-\int_{-\infty}^{x}\left[h(\mu)-h_{\mathrm{Asy}}(\mu)\right]d\mu=\int_{x}^{+\infty}\left[h(\mu)-h_{\mathrm{Asy}}(\mu)\right]d\mu,
\end{equation}
where $h$ is the Hamiltonian associated with the special solution $\mathsf{q}$ of the Painlev\'{e} I hierarchy $\mathrm{P_{I}^{2k}}$ \eqref{def:PIm} and 
\begin{equation}\label{eq-asym-h-infty}
	h_{\mathrm{Asy}}(x):=\frac{(2k+1)}{2(2k+2)}\alpha_k^{-\frac{1}{2k+1}}\cdot x^{\frac{2k+2}{2k+1}}+\frac{kx}{12(2k+1)(x^2+1)}.
\end{equation}
\end{theorem}

The Hamiltonian $h(x)=h(x,t_1,\ldots,t_{2k-1})$ is related to  $\mathsf{q}$  through the relation
$dh(x)/dx=\mathsf{q}(x)$, and following the analysis in \cite{Claeys-2012} (see also Appendix \ref{PIasy} below), we have, if $t_{1}=t_{2}=\ldots=t_{2k-1}=0$,
\begin{align}\label{eq-h(x)-asym-expansio}
	h(x)=h_{\mathrm{Asy}}(x)+\mathcal{O}\left(|x|^{-\frac{8k+5}{4k+2}}\right), \qquad x \to \pm \infty,
\end{align}
where $h_{\mathrm{Asy}}$ is defined in \eqref{eq-asym-h-infty}. Thus, the function $I_h$ in \eqref{def-J0-two-representation} is well-defined. We also mention that it is possible to improve the error estimate in \eqref{eq-Fredholm-asym}. Since the main focus of this work is on the constant term and uniformity of the expansion in $x$, we do not pursue an optimal error estimate here. 

We assume that all the parameters $t_j$, $j=1,\ldots,2k-1$, vanish in Theorem \ref{Thm-asym-PI2k-determinant} to simplify the asymptotic formula. For general non-zero $t_j$, it is expected that there will be extra terms for the powers of $s$ and $x$, whose coefficients are dependent on $t_j$. Our analysis is also applicable to handle this case. Below, we present the result for $k=1$; see Remark \ref{rem-tj-not-vanishe} below for some explanations.
\begin{corollary}\label{cor-asym-PI2k-determinant}
Let $h(\mu,t_{1})$ be the Hamiltonian corresponding to the $\mathrm{P_{I}^2}$ tritronqu\'{e}e solution $\mathsf{q}$. As $s\to-\infty$, we have
\begin{align}\label{eq-Fredholm-asym-k=1}
	\log\det(I-K_{s}^{(1)})&=\frac{25}{448}s^7+\frac{t_1}{4}s^5+\frac{5x}{32}s^4+\frac{t_1^2}{3}s^3+\frac{t_1x}{2}s^2+\frac{x^2}{4}s-\frac{\log\left|\frac{5}{4}s^3+x+2t_1s\right|}{8} 
 \nonumber \\
		& ~~~ -J_{0}(x,t_1)+J_{1}(x,t_1)+\frac{\log(x^2+1)}{72}
  \nonumber \\
   & ~~~ +\frac{\log{3}}{24}+\frac{\log{(5/4)}}{72}+\chi^{(0)}+o(1),
\end{align}
uniformly valid for $t_1$ in any compact subset of $\mathbb{R}$ and $x\in\left[-c_{1}|s|^3,\frac{5}{4}|s|^3-c_{2}|s|^{\frac{2}{3}+\epsilon}\right]$ with $c_{1}, c_{2}>0$ and $\epsilon\in(0,\frac{7}{3})$ being fixed, where $\chi^{(0)}$ is defined in \eqref{def:chi0},
\begin{equation}
    J_{0}(x,t_{1})=\int_{x}^{+\infty}\left[h(\mu,t_{1})-h_{\mathrm{Asy}}(\mu,t_{1})\right]d\mu
\end{equation}
with 
\begin{equation}
    h_{\mathrm{Asy}}(x,t_{1})=\frac{3}{8}\left(\frac{4}{5}\right)^{\frac{1}{3}}x^{\frac{4}{3}}-\frac{t_1}{2}\left(\frac{4}{5}\right)^{\frac{2}{3}}x^{\frac{2}{3}}+\frac{4t_1^2}{15}-\frac{8}{135}\left(\frac{4}{5}\right)^{\frac{1}{3}}t_1^3x^{-\frac{1}{3}}+\frac{x}{36(x^2+1)}
\end{equation}
and
\begin{equation}
	J_{1}(x,t_1)=\frac{9}{56}\left(\frac{4}{5}\right)^{\frac{1}{3}}x^{\frac{7}{3}}-\frac{3t_1}{10}\left(\frac{4}{5}\right)^{\frac{2}{3}}x^{\frac{5}{3}}+\frac{4t_1^2}{15}x-\frac{8}{45}\left(\frac{4}{5}\right)^{\frac{1}{3}}t_1^3x^{\frac{1}{3}}.
\end{equation}

\end{corollary}

If $x$ belongs to a compact subset of $\mathbb{R}$, one can further simplify \eqref{eq-Fredholm-asym} by noting that $\log\left|\alpha_ks^{2k+1}+x\right|=\log\left|\alpha_ks^{2k+1}\right|+\mathcal{O}(|s|^{-2k-1})$ as $s\to -\infty$. This particularly leads to \eqref{eq:largegapCIK} with $a_m=0$, $m=2,\ldots,4k+1$, and 
\begin{align}
C^{(k)}&=\int_{-\infty}^{x}\left[h(\mu)-h_{\textrm{Asy}}(\mu)\right]d\mu+\frac{(2k+1)^2}{2(2k+2)(4k+3)}\alpha_k^{-\frac{1}{2k+1}}\cdot x^{\frac{4k+3}{2k+1}}\nonumber
\\
&~~~+\frac{k\log(x^2+1)}{24(2k+1)}+\frac{\log{(2k+1)}}{24}-\frac{1+3k}{{12} (2k+1)}\log{\alpha_k}+\chi^{(0)}.
\end{align}

Two applications of our main theorem are given in sequel. In view of \eqref{def-J0-two-representation}, it is immediate that the total integral of the Hamiltonian corresponding to the special solution $\mathsf{q}$ vanish.
\begin{corollary}
With the same Hamiltonian $h$ as in Theorem \ref{Thm-asym-PI2k-determinant},  
we have
\begin{equation} \label{eq:total-int-H}
	\int_{-\infty}^{+\infty}\left[h(\mu)-h_{\mathrm{Asy}}(\mu)\right]d\mu=0.
\end{equation} 	
\end{corollary}
This formula generalizes the result in \cite{Dai-Long}, where the above equality is proved for the $\mathrm{P_{I}^{2}}$ case.

The other application of Theorem \ref{Thm-asym-PI2k-determinant} follows from the uniformity of \eqref{eq-Fredholm-asym} in $x$. By a proper scaling of $x$ in terms of $s$, one could recover the asymptotics of Tracy-Widom distribution, as stated in the following corollary. 
\begin{corollary}
With $F(s;x)$ defined in \eqref{eq-Fredholm-asym}, we have, as $s\to -\infty$, 
\begin{equation}\label{eq:Ftansition}
		F\left(s;\alpha_k|s|^{2k+1}+(2k+1)^{\frac{1}{3}}\alpha_k^{\frac{1}{3}}\tilde{s}|s|^{\frac{2k}{3}}\right)=-\frac{1}{12}|\tilde{s}|^3-\frac{1}{8}\log|\tilde{s}|+\chi^{(0)}+o(1),
	\end{equation}
	uniformly for $|s|^{\epsilon}\ll-\tilde{s}\ll |s|^{\frac{k}{3}+\frac{1}{4}}$, where $\epsilon$ is any real number belonging to $\left(0,\frac{k}{3}+\frac{1}{4}\right)$, the constants $\alpha_k$ and $\chi^{(0)}$ are defined in \eqref{eq-def-Ak} and \eqref{def:chi0}, respectively. 
\end{corollary}

The above corollary is straightforward when replacing $x$ by $\alpha_k|s|^{2k+1}+(2k+1)^{\frac{1}{3}}\alpha_k^{\frac{1}{3}}\tilde{s}|s|^{\frac{2k}{3}}$ in \eqref{eq-Fredholm-asym}. It particularly reveals the transition from higher-order analogues of Tracy-Widom distribution to the classical Tracy-Widom distribution in the large gap asymptotic regime. It is also worth noting that a transition from the $\mathrm{P_{I}^{2}}$ kernel to the Airy kernel was established in \cite{Claeys-2006-thesis}, which should compare with \eqref{eq:Ftansition} for $k=1$. 


\subsection{Strategy of the proof}
With $F(s;x)$ defined in  \eqref{eq-Fredholm-asym}, it is readily seen that $F(+\infty; x) = 0$, and thus
\begin{equation}\label{eq: F-original-int}
	F(s; x) = -\int_s^{+\infty} \frac{\partial F}{\partial \tau} (\tau; x) d \tau.  
\end{equation} 
The large gap asymptotics of $F$ then follows from the asymptotics of $\frac{\partial F}{\partial s} (s;x)$ as $s \to -\infty$. Particularly, due to the integrable structure of the $\mathrm{P_{I}^{2k}}$ kernel (see \eqref{def:limKer} below), $\frac{\partial F}{\partial s} (s; x) $ is related to a Riemann-Hilbert (RH) problem under a general framework established in \cite{DIZ97}. By performing the powerful Deift-Zhou nonlinear steepest descent analysis \cite{DX93} for the associated RH problem, one finally achieves the goal. This is exactly the idea adapted in \cite{CIK10} to establish \eqref{eq:largegapCIK}. 

An insurmountable obstacle occurs if one wants to further evaluate the constant of integration $C^{(k)}$ by the same strategy. The challenge lies in the fact that one needs to understand detailed information of $\frac{\partial F}{\partial s} (s;x)$ across an infinite interval $(s, +\infty)$ for that purpose. This seems impracticable since the function $\frac{\partial F}{\partial s}(s; x)$ is highly transcendental (see \cite[Theorem 1.12]{CIK10}). This limitation also explains why most of the similar constant problems arising from mathematical physics remain open.

To overcome this difficulty, the key idea in our approach is to investigate uniform asymptotics of the partial derivatives of $F(s; x)$ with respect to both $s$ and $x$. The motivation follows from the observation that 
\begin{equation} \label{eq: F-new-int}
	F(s; x) = - \int_{x}^{x_0} \frac{\partial F}{\partial \mu}(s; \mu) d\mu -\int_{s}^{+\infty} \frac{\partial F}{\partial\tau}(\tau; x_0) d\tau,
\end{equation}
which is valid for any real $x_0$. In particular, the arbitrariness of $x_0$ provides the flexibility of choosing the variable $|x_0|$ to be large alongside the variable $|s|$. This is essential in the sense that the behaviors of $\frac{\partial F}{\partial x}(s; x)$ and $\frac{\partial F}{\partial s}(s; x)$ degenerate in the asymptotic regime, which can be readily established through their connections with RH problems. It comes out that both $\frac{\partial F}{\partial s}(s; x)$ and $ \frac{\partial F}{\partial x}(s; x)$ exhibit qualitatively different asymptotic behaviors in three different regions of the $(x,s)$-plane as illustrated in Figure \ref{Fig:asy-region} and explained next. 

\begin{figure}[h]
	\begin{center}
	\includegraphics[width=14cm]{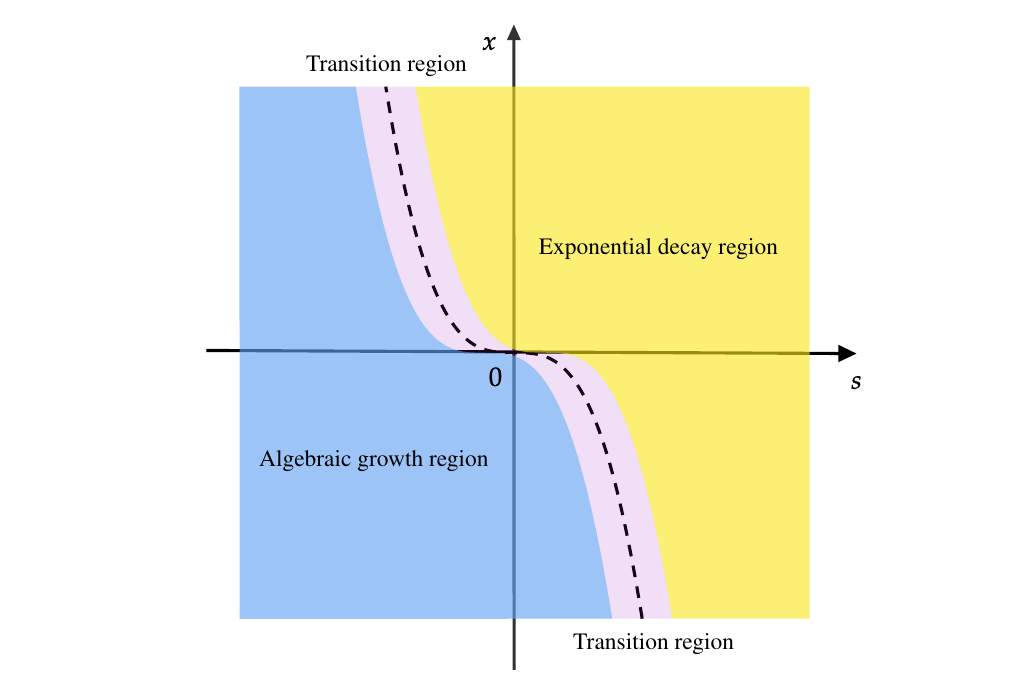}
	\caption{The three regions in asymptotic studies of $\frac{\partial}{\partial s} F(s; x)$ and $ \frac{\partial}{\partial x} F(s; x)$. The exponentially decay and algebraic growth regions are separated by the transitional region in the middle. The dashed curve stands for the critical curve $x = -\alpha_k s^{2k+1}$, where $\alpha_k$ is given in \eqref{eq-def-Ak}.} \label{Fig:asy-region}
	\end{center}
\end{figure}

\begin{itemize}
    \item (Algebraic growth region) This region lies on the left hand side of the critical curve $x = -\alpha_{k}s^{2k+1}$; see the blue part of Figure \ref{Fig:asy-region}. The partial derivatives of $F$ grow like powers of $s$ in this region; see Lemmas \ref{lem-asym-dF/dmu} and \ref{lem-dF/dtau-case-I} for detailed descriptions.
    \item (Exponential decay region) This region lies on the right hand side of the critical curve $x = -\alpha_{k}s^{2k+1}$; see the yellow part of Figure \ref{Fig:asy-region}. The partial derivatives of $F$ decay exponentially in this region; see Lemma \ref{lem-dF/dtau-case-III} 
     for a detailed description.
    \item (Transition region) This region bridges the aforementioned two regions, which surrounds the critical curve $x = -\alpha_{k}s^{2k+1}$. Transition asymptotics in this region involves the Hamiltonian associated with the Hastings-McLeod solution of the Painlev\'{e} II equation; see Lemma \ref{lem-dF/dtau-case-II} for a detailed description.
\end{itemize}

As long as uniform asymptotics of the partial derivatives of $F(s;x)$ are available, we are able to determine the multiplicative constant in the asymptotics of $F(s;x)$ after substituting the integrand in \eqref{eq: F-new-int} by their asymptotic approximations. More precisely, we will set $x_0 = \pm |s|^{2k+1}$ in \eqref{eq: F-new-int}, replace the contour of integration therein by the lines depicted in Figure \ref{Fig:int-contour} and conduct the following derivation.

\begin{figure}[h]
	\begin{center}
	\includegraphics[width=16cm]{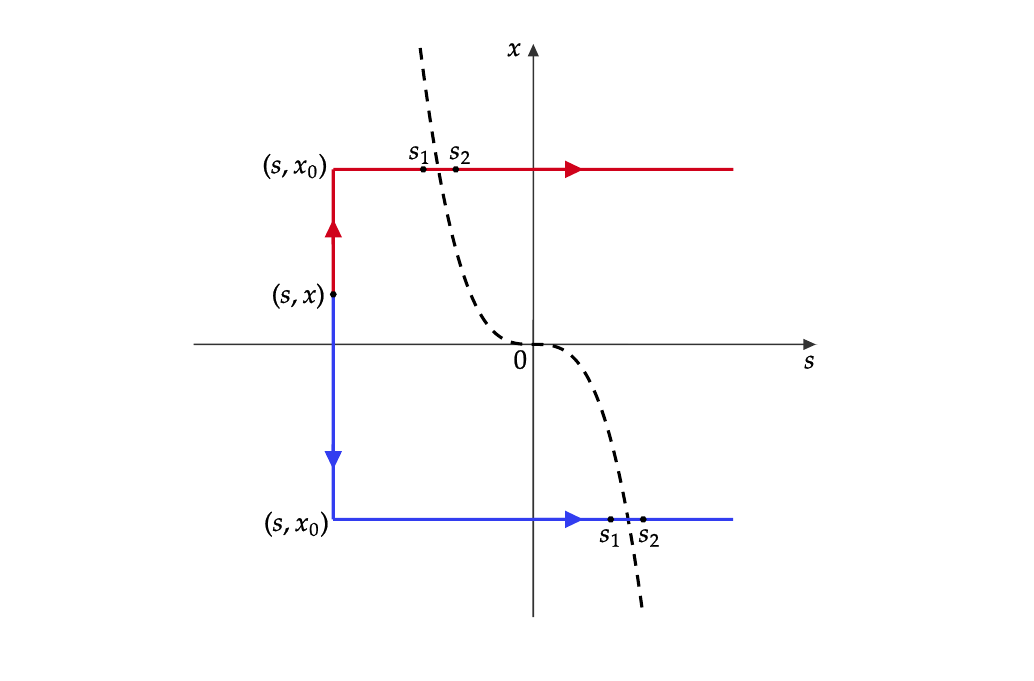}
			%
		\caption{The contours of integration in the $(x,s)$-plane. We choose the red lines if $x_0>0$ and the blue lines if $x_0<0$. The dashed curve is the critical curve $x = -\alpha_k s^{2k+1}$, where $\alpha_k$ is given in \eqref{eq-def-Ak}.} \label{Fig:int-contour}
	\end{center}
\end{figure}

\begin{itemize}
   
\item 
To evaluate the integral $\int_{x}^{x_0}\frac{\partial F}{\partial \mu}(s;\mu)d\mu$, we note that $(s,\mu)$  with $\mu\in(x,x_0)$ is always contained in the algebraic growth region of Figure \ref{Fig:asy-region}. The asymptotic approximation then follows from the result in Lemma \ref{lem-asym-dF/dmu}.

\item 
To evaluate the integral $\int_{s}^{+\infty}\frac{\partial F}{\partial\tau}(\tau;x_0)d\tau$, we split the interval of integration  into three parts: $(s, s_1)$, $(s_1, s_2)$ and $(s_2, +\infty)$, as shown in Figure \ref{Fig:int-contour}. These three sub-intervals lie in the algebraic growth region, the transition region and the exponential decay region, respectively. Uniform asymptotics of the $\frac{\partial}{\partial s} F(s; x)$ for each region
are available in Lemmas \ref{lem-dF/dtau-case-I}--\ref{lem-dF/dtau-case-III}. It is worth noting that the integral over $(s_1,s_2)$ in the transition region will contribute to the nontrivial constant term $\chi^{(0)}$ in \eqref{def:chi0} in the asymptotics of $F(s;x)$.


\end{itemize}

In the proof of Theorem \ref{Thm-asym-PI2k-determinant}, it suffices to choose either the red or the blue lines of integration depicted in Figure \ref{Fig:int-contour}. Given that the final asymptotic result remains the same regardless of the chosen contour, we can compare the asymptotic results obtained with $x_0 = -|s|^{2k+1}$ and $x_0 = |s|^{2k+1}$. This comparison leads to the total integral for the Hamiltonian in \eqref{eq:total-int-H}, which is a by-product of our major findings.

Finally, we believe that our approach has the advantage of being applicable to evaluate the similar multiplicative constant encountered in the studies of large gap asymptotics. Some examples include the (hard-edge) Pearcey determinant \cite{DXZ21,YZ24a}, the higher order Airy determinant \cite{CTG21} and the determinants associated with transcendental kernels constructed from the Painlev\'{e} II hierarchy. We plan to report the relevant results in future publications.


\paragraph{Organization of the paper}
The rest of the paper is organized as follows. In Section \ref{Sec: main-proof}, we state four key lemmas about the uniform asymptotics of partial derivatives of $F$ in different asymptotic regions. A combination of these lemmas will lead to the proof of Theorem \ref{Thm-asym-PI2k-determinant}. In Section \ref{Sec:Pre}, a few preliminary work is made to prepare for the asymptotic analysis. This includes formulation of RH problems and derivations of two differential identities. The Deift-Zhou steepest analysis of the associated RH problems is conducted in Sections \ref{Sec:Alg}--\ref{Sec:Exp}, which serves to validate the lemmas presented in Section \ref{Sec: main-proof}. Additional generalizations of the existing asymptotic results are provided in Appendices \ref{sec-appendix-p34-parametrix}--\ref{PIasy} to complete the analysis presented in this paper. 

\paragraph{Notations} Throughout this paper, the following notations are frequently used.
	\begin{itemize}
		\item If $A$ is a matrix, then $(A)_{ij}$ stands for its $(i,j)$-th entry and $A^{\msf T}$ stands for its transpose. An unimportant entry of $A$ is denoted by $\star$. We use $I$ to denote the identity matrix.

		\item We denote by $U(z_0; r)$ the open disc centered at $z_0$ with radius $r > 0$, {\it i.e.},
		\begin{equation}\label{def:dz0r}
			U(z_0; r) := \{ z\in \mathbb{C} \mid |z-z_0|<r \},
		\end{equation}
		and by $\partial U(z_0; r)$ its boundary. The orientation of $\partial U(z_0; r)$ is taken in a clockwise manner.
		\item As usual, the Pauli matrices $\{\sigma_j\}_{j=1}^3$ are 
		\begin{equation}\label{def:Pauli}
			\sigma_1=\begin{pmatrix}
				0 & 1 \\
				1 & 0
			\end{pmatrix},
			\qquad
			\sigma_2=\begin{pmatrix}
				0 & -i \\
				i & 0
			\end{pmatrix},
			\qquad
			\sigma_3=
			\begin{pmatrix}
				1 & 0 \\
				0 & -1
			\end{pmatrix}.
		\end{equation}
\item If a function $f(z)$ admits a pole at $z_0$, we use $\Res(f(z),z_0)$ to denote its residue at $z_0$. 
  \item We will perform Deift-Zhou nonlinear steepest descent analysis for a RH problem several times. We adopt the same notations (such as $R$, $v_{R}$, $\ldots$) during each analysis. They should be understood in different contexts, and we trust that this will not lead to any confusion.
	\end{itemize}

\section{Key lemmas and proof of Theorem \ref{Thm-asym-PI2k-determinant}} \label{Sec: main-proof}


Recall \eqref{eq: F-new-int} and denote the integrals therein by 
\begin{equation} \label{eq-integral-representation-F}
I_{1}(s;x,x_{0}):= \int_{x}^{x_{0}}\frac{\partial F}{\partial \mu}(s;\mu)d\mu, \qquad I_{2}(s;x_{0}) = \int_{s}^{+\infty}\frac{\partial F}{\partial\tau}(\tau;x_{0})d\tau.
\end{equation}
As mentioned before, we take $x_{0}=\pm |s|^{2k+1}$ and derive the asymptotics of $I_{1}(s;x,x_{0})$ and $I_{2}(s;x_{0})$ as $s\to-\infty$. To achieve this goal, we need the following key lemmas, whose proofs will be postponed in Sections \ref{sec:case-I-proof}, \ref{sec:case-II-proof} and \ref{sec:case-III-proof}, respectively. These lemmas provide essential asymptotics for the partial derivatives of $F(s;x)$ in the $(x,s)$-plane, under which Theorem \ref{Thm-asym-PI2k-determinant} will be proved.

\begin{lemma}[Large $s$ asymptotics of $\frac{\partial F}{\partial x}$ in the algebraic growth region] \label{lem-asym-dF/dmu}
Let $\alpha_{k}$ be defined in \eqref{eq-def-Ak} and $h(x)$ be the Hamiltonian associated with the special solution $\mathsf{q}$ of the Painlev\'{e} I hierarchy $\mathrm{P_{I}^{2k}}$ \eqref{def:PIm}. As $s\to-\infty$, we have 
\begin{equation}\label{eq-asym-dF/dmu}
\begin{split}
\frac{\partial F}{\partial x}(s;x)=h(x)+\frac{\alpha_{k}}{4k+4}s^{2k+2}+\frac{x s}{2}-\frac{1}{8\left(\alpha_{k}s^{2k+1}+x\right)}\left(1+\mathcal{O}(|s|^{-3\epsilon})\right),
\end{split}
\end{equation}
uniformly for all $x\in\left[-c_{1}|s|^{2k+1},\alpha_{k}|s|^{2k+1}-c_{2}|s|^{\frac{2k}{3}+\epsilon}\right]$, where $c_{1},c_{2}$ and $\epsilon$ are three arbitrarily fixed positive constants with $\epsilon\in(0,\frac{4k}{3}+1)$.
\end{lemma} 

\begin{lemma}[Large $x$ asymptotics of $\frac{\partial F}{\partial s}$  in the algebraic growth region] \label{lem-dF/dtau-case-I}
With $\alpha_{k}$ given in \eqref{eq-def-Ak}, we have
\begin{equation} \label{eq-dF/dtau-case-I}
\begin{split}
\frac{\partial F}{\partial s}(s;x)= \frac{1}{4}\left(\alpha_{k}s^{2k+1}+x\right)^2-\frac{(2k+1)\alpha_{k}s^{2k}}{8\left(\alpha_{k}s^{2k+1}+x\right)}\left(1+\mathcal{O}(|x|^{-\frac{k+\frac{1}{2}}{2k+1}})\right)
\end{split}
\end{equation}
as $|x|\to+\infty$, uniformly for $s$ such that $\alpha_{k}s^{2k+1}+x\in\left[-M|x|,-C|x|^{\frac{k+1/6}{2k+1}}\right]$. Here, $M$ and $C$ are two arbitrarily positive constants.
\end{lemma}

\begin{lemma}[Large $x$ asymptotics of $\frac{\partial F}{\partial s}$  in the transition region] \label{lem-dF/dtau-case-II}
As $|x|\to+\infty$, we have
\begin{equation}\label{eq-dF/dtau-case-II}
\frac{\partial F}{\partial s}(s;x)=\frac{\partial\chi(s;x)}{\partial s}H_{\mathrm{P_{II}}}\left(\chi(s;x)\right)\left(1+\mathcal{O}(|x|^{-\frac{k+5/6}{2k+1}})\right)+\mathcal{O}(|x|^{\frac{k-1/3}{2k+1}}),
\end{equation}
uniformly for $s$ such that 
$\alpha_{k}s^{2k+1}+x\in\left[-C|x|^{\frac{k+1/6}{2k+1}},C|x|^{\frac{k+1/6}{2k+1}}\right]$. Here, $C$ is an arbitrarily fixed positive constant, $\alpha_{k}$ is defined in \eqref{eq-def-Ak}, $H_{\mathrm{P_{II}}}(\cdot)$ is the Hamiltonian corresponding to the Hastings-McLeod solution of the  Painlev\'{e} II equation and 
\begin{equation}\label{eq-def-chi(tau)}
\chi(s;x):=\begin{cases}
\frac{\alpha_{k}s^{2k+1}+x}
{(2k+1)^{\frac{1}{3}}\alpha_{k}^{\frac{1}{3(2k+1)}}x^{\frac{2k}{3(2k+1)}}}, & \alpha_{k}s^{2k+1}+x\in\left[-C|x|^{\frac{k+1/6}{2k+1}},0\right],\\
f_{3}\left(s;s_{0}\right), & \alpha_{k}s^{2k+1}+x\in\left[0,C|x|^{\frac{k+1/6}{2k+1}}\right],
\end{cases}
\end{equation} 
is a continuous and piecewise differentiable function. In the above formula, the function $f_{3}$ is defined in  \eqref{eq-def-tilde-f2} and 
\begin{equation}\label{eq-def-s0}
s_{0}=-\left(x/\alpha_{k}\right)^{\frac{1}{2k+1}}.
\end{equation}
\end{lemma}

\begin{lemma}[Large $x$ asymptotics of $\frac{\partial F}{\partial s}$  in the exponential decay region] \label{lem-dF/dtau-case-III}
There exists a positive constant $M$ such that as $|x|\to+\infty$,
\begin{equation}\label{eq-dF/ds-case-III}
\frac{\partial F}{\partial s}(s;x)=\Boh\left(e^{-M|x|^{\frac{4k}{4k+2}}(s-s_{0})^{\frac{3}{2}}}\right),
\end{equation}
uniformly for all $s$ such that 
$\alpha_{k}s^{2k+1}+x\in\left[C|x|^{\frac{k+1/6}{2k+1}},+\infty\right)$. Here,  $C$ is an arbitrarily fixed positive constant, $\alpha_{k}$ is defined in \eqref{eq-def-Ak} and $s_{0}$ is given in \eqref{eq-def-s0}.
\end{lemma}

\begin{remark}
    Note that in the transition region, both $|s|$ and $|x|$ are large.  Although the error term in \eqref{eq-dF/dtau-case-II} is large as $|x|\to+\infty$, its contribution to the asymptotics of $F(s;x)$ as $s\to-\infty$ is small since the integration interval with respect to $s$ in the transition region is small enough and we refer to the proof of Theorem \ref{Thm-asym-PI2k-determinant} below for a detailed discussion.  
\end{remark}

With the aid of the above lemmas, we are ready to prove our main theorem.

\medskip



\paragraph{Proof of Theorem \ref{Thm-asym-PI2k-determinant}}
Let us first consider the integral $I_{1}(s;x,x_{0})$ in \eqref{eq-integral-representation-F}. According to Figures \ref{Fig:asy-region} and \ref{Fig:int-contour}, the contour of integration is always contained in the algebraic region. It then follows from a combination of \eqref{eq-asym-h-infty} and \eqref{eq-asym-dF/dmu}  that, as $s\to-\infty $,
\begin{align}
I_{1}(s;x,x_{0})&=\int_{x}^{x_{0}}\left[h(\mu)-h_{\mathrm{Asy}}(\mu)\right]d\mu+\int_{x}^{x_{0}}h_{\mathrm{Asy}}(\mu)d\mu+\int_{x}^{x_{0}}\left[\frac{\alpha_{k}s^{2k+2}}{4k+4}+\frac{\mu s}{2}\right]d\mu \nonumber \\
&\quad -\frac{1}{8}\int_{x}^{x_{0}}\frac{1}{\alpha_{k}s^{2k+1}+\mu}d\mu+\mathcal{O}(|s|^{-2\epsilon}) \nonumber \\
&=I_{1,1}(x,x_{0})+I_{1,2}(s;x_{0})-I_{1,2}(s;x)+\mathcal{O}(|s|^{-2\epsilon}), \label{eq-I1-representation}
\end{align}
uniformly for all $x\in\left[-c_{1}|s|^{2k+1},\alpha_{k}|s|^{2k+1}-c_{2}|s|^{\frac{2k}{3}+\epsilon}\right]$, where $c_{1},c_{2},\epsilon$ are three arbitrary fixed positive constants with $\epsilon\in(0,\frac{4k}{3}+1)$,
where  
\begin{align}
I_{1,1}(x,x_{0})&= \int_{x}^{x_{0}}\left[h(\mu)-h_{\mathrm{Asy}}(\mu)\right]d\mu \nonumber \\
&=\begin{cases}
\int_{x}^{+\infty}\left[h(\mu)-h_{\mathrm{Asy}}(\mu)\right]d\mu+\mathcal{O}\left(|s|^{-(4k+3)/2}\right),& \text{when } x_{0}=|s|^{2k+1},\\
-\int_{-\infty}^{x}\left[h(\mu)-h_{\mathrm{Asy}}(\mu)\right]d\mu+\mathcal{O}\left(|s|^{-(4k+3)/2}\right),& \text{when }  x_{0}=-|s|^{2k+1},
\end{cases}
\end{align}
and 
\begin{equation}\label{eq-expression-I12}
\begin{split}
I_{1,2}(s;x) &  = \frac{(2k+1)^2}{2(2k+2)(4k+3)} \alpha_{k}^{-\frac{1}{2k+1}}\cdot x^{\frac{4k+3}{2k+1}}+\frac{k}{24(2k+1)}\log{(x^2+1)}\\
&\quad+\frac{\alpha_{k}}{4k+4}s^{2k+2}x+\frac{1}{4}x^2s-\frac{1}{8}\log|\alpha_{k}s^{2k+1}+x|.
\end{split}
\end{equation}

Next, we proceed to compute the second integral $I_{2}(s;x_{0})$ in \eqref{eq-integral-representation-F}. In view of Figures \ref{Fig:asy-region} and \ref{Fig:int-contour} again, the contour of integration  has to cross all the algebraic growth, transition and exponential decay regions. Based on this observation, we further divide the contour of integration  into three parts: $(s, s_1) \cup (s_1, s_2) \cup (s_2, +\infty)$, where
\begin{equation} \label{eq:s12-def}
\begin{split}
s_{1}:=&-\left(\frac{x_{0}}{\alpha_{k}}\right)^{\frac{1}{2k+1}}(1+\sgn(x_{0})|s|^{-k-\frac{5}{6}})^{\frac{1}{2k+1}}, \\
s_{2}:=&-\left(\frac{x_{0}}{\alpha_{k}}\right)^{\frac{1}{2k+1}}(1-\sgn(x_{0})|s|^{-k-\frac{5}{6}})^{\frac{1}{2k+1}}.
\end{split}
\end{equation}
With the choice of $x_{0}=\pm |s|^{2k+1}$, we have from the above formula that 
\begin{equation} \label{eq:s1-x0-relation}
\alpha_{k}s_{1}^{2k+1}+x_{0}=-|s|^{k+\frac{1}{6}}=-|x_{0}|^{\frac{k+1/6}{2k+1}} \quad \text{and}\quad \alpha_{k}s_{2}^{2k+1}+x_{0}=|s|^{k+\frac 16}=|x_{0}|^{\frac{k+1/6}{2k+1}}.
\end{equation}
Then, we rewrite $I_{2}(s;x_{0})$ as
\begin{align}\label{eq-separate-into-3-integrals}
I_{2}(s;x_{0})&=\int_{s}^{s_{1}}\frac{\partial F}{\partial\tau}(\tau;x_{0})d\tau+\int_{s_{1}}^{s_{2}}\frac{\partial F}{\partial\tau}(\tau;x_{0})d\tau+\int_{s_{2}}^{+\infty}\frac{\partial F}{\partial\tau}(\tau;x_{0})d\tau
\nonumber \\
&=I_{2,1}(s;x_{0})+I_{2,2}(s;x_{0})+I_{2,3}(s;x_{0}),
\end{align}
and calculate the three terms $I_{2,j}(s;x_{0})$, $j=1,2,3$, one by one.
According to Lemma \ref{lem-dF/dtau-case-I}, we have
\begin{equation}\label{eq-I21-representation}
I_{2,1}(s;x_{0})=J_{2}(s_{1}; x_{0})-J_{2}(s;x_{0})+\frac{1}{8}\log\left|\alpha_{k}s^{2k+1}+x_{0}\right|-\frac{1}{8}\log\left|\alpha_{k}s^{2k+1}+x_{0}\right|+\mathcal{O}(|s|^{-k}),
\end{equation}
as $s\to-\infty$, where
\begin{equation}\label{eq-J2(-tildex,s)}
J_{2}(s;x_{0})=\frac{\alpha_{k}^2}{4(4k+3)}s^{4k+3}+\frac{\alpha_{k}}{4k+4}x_{0}s^{2k+2}+\frac{1}{4}x_{0}^2s.
\end{equation}
Moreover, Lemma \ref{lem-dF/dtau-case-III} implies that $I_{2,3}(s;x_{0})$ is exponentially small as $s\to-\infty$.  

The remaining task is to derive the asymptotic behavior of $I_{2,2}(s;x_{0})$ as $s\to-\infty$.
It is readily seen from the definitions of $s_{1}$ and $s_{2}$ in \eqref{eq:s12-def} that $\chi(s_{1};x_{0})$ and $\chi(s_{2};x_{0})$ tend to negative and positive infinity, respectively and $|s_{1}-s_{2}|=\mathcal{O}(|s|^{-k+\frac{1}{6}})$ as $s\to-\infty$. Hence, by Lemma \ref{lem-dF/dtau-case-II} and the fact that the Hamiltonian $H_{\mathrm{P_{II}}}(\chi)$ decays exponentially as $\chi\to+\infty$, we have
\begin{align}
I_{2,2}(s;x_{0})&
=\left[ \int_{s_{1}}^{s_{2}}\frac{\partial \chi(\tau;x_{0})}{\partial \tau}H_{\mathrm{P_{II}}}(\chi(\tau;x_{0}))d\tau \right]
\left(1+\mathcal{O}(|s|^{-k-\frac{5}{6}})\right)+\mathcal{O}(|s|^{-\frac{1}{6}}) 
\nonumber \\
& =\left[\int_{\chi(s_{1};x_{0})}^{+\infty}H_{\mathrm{P_{II}}}(\xi)d\xi\right]\left(1+\mathcal{O}(|s|^{-k-\frac{5}{6}})\right)+\mathcal{O}(|s|^{-\frac{1}{6}}). \nonumber
\end{align}
A combination of \eqref{eq:TW-large gap asy} and \eqref{eq-tracy-widom-formula} yields, as $s\to-\infty$,
\begin{align}
&I_{2,2}(s;x_{0}) \nonumber \\
&=\left[-\chi^{(0)}-\frac{\chi(s_{1};x_{0})^3}{12}+\frac{\log|\chi(s_{1};x_{0})|}{8}+\mathcal{O}(|\chi(s_{1};x_{0})|^{-\frac{3}{2}})\right]\left(1+\mathcal{O}(|s|^{-k-\frac{5}{6}})\right)+\mathcal{O}(|s|^{-\frac{1}{6}}),\nonumber\\
&=\left[-\chi^{(0)}-\frac{\chi(s_{1};x_{0})^3}{12}+\frac{1}{8}\log|\chi(s_{1};x_{0})|\right]\left(1+\mathcal{O}(|s|^{-k-\frac{5}{6}})\right)+\mathcal{O}(|s|^{-\frac{1}{6}}),
\end{align}
where $\chi(\tau;x_{0})$ is defined in \eqref{eq-def-chi(tau)}. Here we have used the facts that, as $s\to-\infty$, $\chi(s_{1};x_{0})$ tends to negative infinity and $|\chi(s_{1};x_{0})|^{-\frac{3}{2}}$ is controlled by $\mathcal{O}(|s|^{-\frac{3}{4}})$. Moreover, it follows from the definition of $s_{1}$ and $\chi(\tau;x_{0})$ that $\chi(s_{1};x_{0})=\mathcal{O}(|s|^{\frac{k}{3}+\frac{1}{6}})$ and 
$\frac{|s|^{-k-5/6}\chi(s_{1};x_{0})^3}{12}=\mathcal{O}(|s|^{-\frac{1}{3}})$ as $s\to-\infty$. Hence, as $s\to-\infty$, we have
\begin{align}
\label{eq-asym-I22}
I_{2,2}(s;x_{0})&
=-\frac{\chi(s_{1};x_0)^3}{12}+\frac{1}{8}\log\left|\left(\alpha_{k}s_{1}^{2k+1}+x_{0}\right)\right|-\frac{k}{12(2k+1)}\log|x_{0}|\nonumber 
\\
&\quad -\chi^{(0)}-\frac{\log(2k+1)}{24}-\frac{\log{\alpha_{k}}}{24(2k+1)}
+\mathcal{O}(|s|^{-\frac{1}{3}}).
\end{align}

Now, we have obtained asymptotics of the integrals $I_{1}(s;x,x_{0})$ and $I_{2}(s;x_{0})$ in \eqref{eq-integral-representation-F}. A combination of the formulas \eqref{eq-I1-representation}--\eqref{eq-expression-I12}, \eqref{eq-separate-into-3-integrals}--\eqref{eq-J2(-tildex,s)} and \eqref{eq-asym-I22} yields 
\begin{align}\label{eq-asym-F-presentation}
F(s;x)&=\frac{1}{4(4k+3)}\alpha_{k}^2s^{4k+3}+\frac{\alpha_{k}}{4k+4}xs^{2k+2}+\frac{x^2s}{4}-\frac{1}{8}\log\left|\alpha_{k}s^{2k+1}+x\right| \nonumber \\
&\quad -I_{1,1}(x)+\frac{(2k+1)^2}{2(2k+2)(4k+3)} \alpha_{k}^{-\frac{1}{2k+1}}\cdot x^{\frac{4k+3}{2k+1}}+\frac{k\log(x^2+1)}{24(2k+1)} \nonumber \\
&\quad +\chi^{(0)}+\frac{\log(2k+1)}{24}+\frac{\log{\alpha_{k}}}{24(2k+1)} \nonumber \\
&\quad -\frac{(2k+1)^2}{2(2k+2)(4k+3)} \alpha_{k}^{-\frac{1}{2k+1}}\cdot x_{0}^{\frac{4k+3}{2k+1}}-J_2(s_1; x_0)+\frac{\chi(s_{1};x_{0})^3}{12}+\mathcal{O}(|s|^{-\frac{1}{6}})
\end{align}
as $s\to-\infty$. To complete the proof of Theorem \ref{Thm-asym-PI2k-determinant}, we only need to show the last line in the above formula does not contribute to the constant term, that is,
\begin{equation}\label{J1-J2+yhi^3/12}
\frac{(2k+1)^2}{2(2k+2)(4k+3)} \alpha_{k}^{-\frac{1}{2k+1}}\cdot x_{0}^{\frac{4k+3}{2k+1}}+J_2(s_1; x_0)-\frac{\chi(s_{1};x_{0})^3}{12}=\mathcal{O}(|s|^{-\frac{1}{3}}), \qquad s\to-\infty.
\end{equation}
Let us consider the case when $x_0 = |s|^{2k+1}$; the case $x_0 = -|s|^{2k+1}$ can be proved in a similar manner. First, it follows from the definition of $\chi(s;x)$ in \eqref{eq-def-chi(tau)} and the relation \eqref{eq:s1-x0-relation} that
\begin{multline}
     \frac{(2k+1)^2}{2(2k+2)(4k+3)} \alpha_{k}^{-\frac{1}{2k+1}}\cdot x_{0}^{\frac{4k+3}{2k+1}}-\frac{\chi(s_{1};x_{0})^3}{12} 
   \\  
   =\frac{(2k+1)^2\alpha_{k}^{-\frac{1}{2k+1}}}{2(2k+2)(4k+3)}  |s|^{4k+3} + \frac{\alpha_{k}^{-\frac{1}{2k+1}}}{12 (2k+1)} |s|^{k+\frac{1}{2}}.
\end{multline}
Next, for the remaining factor $J_2(s_1; x_0)$ in \eqref{J1-J2+yhi^3/12}, we use its definition from \eqref{eq-J2(-tildex,s)} and the relationship \eqref{eq:s1-x0-relation} once more to obtain
\begin{align}
    & J_2(s_1; x_0) = s_1 \left( \frac{\alpha_{k}^2 s_1^{4k+2}}{4(4k+3)}+\frac{\alpha_{k} s_1^{2k+1}}{4k+4}x_{0}+\frac{1}{4}x_{0}^2 \right) \nonumber \\
    &  = s_1 \left( \frac{(2k+1)^2}{4(k+1)(4k+3)} |s|^{4k+2}-\frac{2k+1}{4(k+1)(4k+3)} |s|^{3k+\frac{7}{6}} +\frac{1}{4(4k+3)} |s|^{2k+\frac{1}{3}} \right).
\end{align}
Recalling the definition of $s_1$ in \eqref{eq:s12-def}, when $x_0 = |s|^{2k+1}>0$, it is straightforward to see that
\begin{equation}\label{eq-expand-s1-by-s}
    \alpha_{k}^{\frac{1}{2k+1}}s_{1} =- |s|-\frac{|s|^{-k+\frac{1}{6}}}{2k+1}  + \frac{k \, |s|^{-2k-\frac{2}{3}}}{(2k+1)^2} -\frac{k(4k+1)}{3(2k+1)^3}|s|^{-3k-\frac{3}{2}}+\mathcal{O}\left(|s|^{-4k-\frac{7}{3}}\right).
\end{equation}
as $s\to-\infty$. Finally, by combining the above three formulas, we arrive at the desired approximation in \eqref{J1-J2+yhi^3/12}. This finishes the proof of Theorem \ref{Thm-asym-PI2k-determinant}. \qed

\section{Preliminaries} \label{Sec:Pre}

\subsection{A Riemann-Hilbert characterization of the $\mathrm{P_{I}^{2k}}$ kernel}
Recall the $\mathrm{P_{I}^{2k}}$ kernel $K^{(k)}$ in \eqref{eq:singularlim}. By \cite{CIK10}, we have 
\begin{equation}\label{def:limKer}
 K^{(k)}(u,v) = K^{(k)}(u,v;x,t_1,\ldots,t_{2k-1}) = \frac{\Psi_1^{(2k)}(u)\Psi_{2}^{(2k)}(v)-\Psi_1^{(2k)}(v)\Psi_{2}^{(2k)}(u)}{-2\pi i (u-v)},
\end{equation}
where $(\Psi_{1}^{(2k)}(\zeta),\Psi_{2}^{(2k)}(\zeta))^\textsf{T}$ is related to the following RH problem.

\paragraph{RH problem for $\Psi$}
\begin{enumerate}
\item [\rm (a)]
 $\Psi(\zeta)=\Psi^{(2k)}(\zeta;x,t_1,\ldots,t_{2k-1})$ is a $2\times 2$  matrix-valued function, which is analytic for $\zeta$ in
  $\mathbb{C}\backslash \{\cup_{j=1}^4\Gamma_j\cup\{0\}\}$, where
\begin{align*}
\Gamma_1=(0,+\infty), ~~ \Gamma_2=e^{\frac{4k+2 }{4k+3}\pi i}(0,+\infty), ~~\Gamma_3=(-\infty,0),~~\Gamma_4=e^{-\frac{4k+2 }{4k+3}\pi i}(0,+\infty),
\end{align*}
with the orientation shown in Figure \ref{fig:jumpsPsi}.

\begin{figure}[t]
\begin{center}

\tikzset{every picture/.style={line width=0.75pt}} 

\begin{tikzpicture}[x=0.75pt,y=0.75pt,yscale=-1,xscale=1]

\draw    (52,120) -- (203,120) -- (333,120) ;
\draw [shift={(131.3,120)}, rotate = 180] [fill={rgb, 255:red, 0; green, 0; blue, 0 }  ][line width=0.08]  [draw opacity=0] (7.14,-3.43) -- (0,0) -- (7.14,3.43) -- cycle    ;
\draw [shift={(271.8,120)}, rotate = 180] [fill={rgb, 255:red, 0; green, 0; blue, 0 }  ][line width=0.08]  [draw opacity=0] (7.14,-3.43) -- (0,0) -- (7.14,3.43) -- cycle    ;
\draw    (116.5,20) -- (216.5,120) ;
\draw [shift={(169.19,72.69)}, rotate = 225] [fill={rgb, 255:red, 0; green, 0; blue, 0 }  ][line width=0.08]  [draw opacity=0] (7.14,-3.43) -- (0,0) -- (7.14,3.43) -- cycle    ;
\draw    (117.5,208) -- (216.5,120) ;
\draw [shift={(169.84,161.48)}, rotate = 138.37] [fill={rgb, 255:red, 0; green, 0; blue, 0 }  ][line width=0.08]  [draw opacity=0] (7.14,-3.43) -- (0,0) -- (7.14,3.43) -- cycle    ;
\draw  [fill={rgb, 255:red, 0; green, 0; blue, 0 }  ,fill opacity=1 ] (216.67,119.83) .. controls (216.67,119.19) and (216.15,118.66) .. (215.5,118.66) .. controls (214.85,118.66) and (214.33,119.19) .. (214.33,119.83) .. controls (214.33,120.48) and (214.85,121) .. (215.5,121) .. controls (216.15,121) and (216.67,120.48) .. (216.67,119.83) -- cycle ;

\draw (218.5,123) node [anchor=north west][inner sep=0.75pt]   [align=left] {0};
\draw (343,117) node [anchor=north west][inner sep=0.75pt]   [align=left] {$\Gamma_1$};
\draw (98,14) node [anchor=north west][inner sep=0.75pt]   [align=left] {$\Gamma_2$};
\draw (34,113) node [anchor=north west][inner sep=0.75pt]   [align=left] {$\Gamma_3$};
\draw (95,205) node [anchor=north west][inner sep=0.75pt]   [align=left] {$\Gamma_4$};

\end{tikzpicture}
   \caption{The jump contours $\Gamma_j$, $j=1,2,3,4$, of the RH problem for $\Psi$.}
   \label{fig:jumpsPsi}
\end{center}
\end{figure}
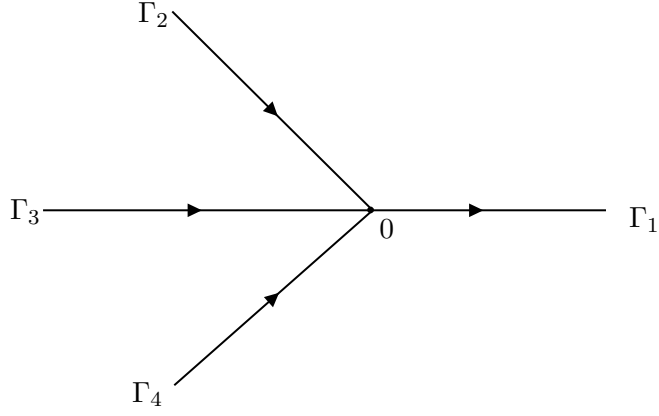

  \item [\rm (b)] On the contour $\cup_{j=1}^4\Gamma_j$, the limiting values of $\Psi_\pm(\zeta)$ exist and satisfy the jump condition
  \begin{gather}\label{eq:Psi-jump}
  \Psi_+ (\zeta)=\Psi_- (\zeta)
  \left\{ \begin{array}{ll}
           \begin{pmatrix}
                                 1 & 1 \\
                                 0 & 1
                                 \end{pmatrix}, &\quad \zeta \in {\Gamma}_1,
\\
           \begin{pmatrix}
                                 1 & 0 \\
                                1 & 1
                                 \end{pmatrix}, &\quad \zeta \in {\Gamma}_2 \cup \Gamma_4,
\\
           \begin{pmatrix}
                                 0 & 1 \\
                                 -1 & 0
                                 \end{pmatrix},& \quad
                                           \zeta \in {\Gamma}_3.
          \end{array}
    \right .
  \end{gather}
\item [\rm (c)]  As $\zeta\rightarrow \infty$, we have
\begin{equation} \label{eq:Psi-infinity}
 \Psi(\zeta) =
       \left (I+\frac {\Psi_{-1}}{\zeta}
    +\mathcal{O}\left( \zeta^{-2}\right) \right)\zeta^{-\frac{1}{4}\sigma_3} N e^{-\theta(\zeta) \sigma_3},
\end{equation}
where $\Psi_{-1}$ is a $\zeta$-independent matrix,
$\sigma_3$ is the Pauli matrix defined in \eqref{def:Pauli}, 
\begin{align}\label{eq-def-N}
N &= \frac{1}{\sqrt{2}}
 \begin{pmatrix}
 1 & 1
 \\
 -1 & 1
 \end{pmatrix}e^{-\frac{1}{4}\pi i \sigma_3}
\end{align}
is a constant matrix,
and
\begin{align}\label{eq-def-theta}
\theta(\zeta)&=\theta(\zeta;x,t_1,\ldots,t_{2k-1})=\frac{4}{4k+3}\zeta^{\frac{4k+3}{2}}+4\sum_{j=1}^{2k-1}\frac{t_j}{2j+1}\zeta^{\frac{2j+1}{2}}+x\zeta^{\frac{1}{2}}.
\end{align}
In the above formulas, the principal branch of the fractional powers is taken.

\item [\rm (d)] $\Psi(\zeta)$ is bounded at $\zeta=0$.

\end{enumerate}

In \cite{Claeys-Vanlessen-2007}, it has been proved that the above RH problem for $\Psi$ is uniquely solvable for $k=1$. The proof can be easily generalized to any integer larger than $1$; see also \cite{CIK10}. The RH problem for $\Psi$ is related to the Painlev\'{e} I hierarchy through the fact that
\begin{equation}
\label{eq-def-Psi--1}
\Psi_{-1}=\left(\begin{matrix}
\frac{h^2-\mathsf{q}}{2} & h\\
\star &  \frac{h^2+\mathsf{q}}{2}
\end{matrix}\right),
\end{equation} 
where $\mathsf{q}$ is the special solution of $\mathrm{P_{I}^{2k}}$ described in Section \ref{sec:results} and $h$ is the associated Hamiltonian.

The functions $(\Psi_{1}^{(2k)},\Psi_{2}^{(2k)})^\textsf{T}$ appearing in the $\mathrm{P_{I}^{2k}}$ kernel \eqref{def:limKer} are the analytic continuation of the first column of $\Psi$ in the region bounded by $\Gamma_1$ and $\Gamma_2$ to the whole complex plane. Moreover, 
$\Psi$ satisfies a Lax pair as follows (cf. \cite{Claeys-2012}):
\begin{equation}\label{Lax-pair}
\frac{\partial \Psi}{\partial \zeta}(\zeta,x)=A(\zeta,x)\Psi(\zeta,x),\qquad \frac{\partial \Psi}{\partial x}(\zeta,x)=L(\zeta,x)\Psi(\zeta,x),
\end{equation}
where $A$ is a polynomial in $\zeta$ of degree $2k+1$ and 
\begin{equation}\label{eq-lax-pair-L}
L(\zeta,x)=\left(\begin{matrix}
0&1\\\zeta+2\q(x)&0
\end{matrix}\right).
\end{equation}

\begin{remark}
If $k=0$ ({\it i.e.}, all the parameters $x, t_{1},\cdots,t_{2k-1}$, vanish), the RH problem for $\Psi$ can reduce to the one for the Airy kernel. More precisely, we have $\Psi^{(0)}(\zeta)=2^{\frac{1}{6}\sigma_{3}}\Phi^{(\mathrm{Ai})}(2^{\frac{2}{3}}\zeta)$, where $\Phi^{(\mathrm{Ai})}$ is the classical Airy parametrix; cf. \cite{Deift1999,DKMVZ99}.  
\end{remark}


\subsection{Differential identities}
\label{sec:diffident}
In view of the integrable structure of the $\mathrm{P_{I}^{2k}}$ kernel given in \eqref{def:limKer}, it is well-known that we are able to establish connections between $\partial F$ and specific RH problems.

\paragraph{RH problem for $Y$}

\begin{enumerate}
\item [(a)] $Y(\zeta)$ is analytic in $\mathbb{C}\setminus[s,+\infty)$.
\item [(b)] For $\zeta\in(s,+\infty)$, the limiting values $Y_{\pm}(\zeta)$ exist and satisfy the jump condition $Y_{+}(\zeta)=Y_{-}(\zeta)J_{Y}(\zeta)$  with 
\begin{equation}
J_{Y}(\zeta)=I-2\pi i \vec{f}(\zeta)\vec{h}^{\top}(\zeta),
\end{equation}
where
\begin{equation}
\vec{f}(\zeta) = \begin{pmatrix} \Psi_{1}^{(2k)}(\zeta) \\ \Psi_{2}^{(2k)}(\zeta) \end{pmatrix}, \qquad \vec{h}(\zeta) = \frac{1}{2\pi i} \begin{pmatrix} -\Psi_{2}^{(2k)}(\zeta) \\ \Psi_{1}^{(2k)}(\zeta) \end{pmatrix}.
\end{equation}
\item [(c)] As $\zeta\to\infty$, we have
\begin{equation} \label{eq: Y-inf}
Y(\zeta)=I+\frac{Y_{-1}}{\zeta}+\mathcal{O}(\zeta^{-2}).
\end{equation}

\item [(d)] As $\zeta \to s$, we have $Y(\zeta)=\Boh(\log(\zeta-s))$. 
\end{enumerate}

By \cite[Lemma 2.12]{DIZ97}, we have 
\begin{equation}\label{eq-def-Y}
Y(\zeta)=I-\int_{s}^{\infty}\frac{\vec{F}(u)\vec{h}^{\top}(u)}{u-\zeta}du,
\end{equation}  
where 
\begin{equation}
\vec{F}:=(I-\mathbb{K}_{s}^{(k)})^{-1}\vec{f},\qquad \vec{H}:=(I-\mathbb{K}_{s}^{(k)})^{-1}\vec{h}.
\end{equation}
Then, $\partial F / \partial x$ is related to the above RH problem in the following way. 

\begin{lemma}\label{lem-differential-indentity-1}
For any $x,s\in\mathbb{R}$, we have
\begin{equation}\label{eq-diff-identity-1}
\frac{\partial F}{\partial x}(s;x) = \frac{\partial}{\partial x}\log\det(I-\mathbb{K}_{s}^{(k)})=-(Y_{-1})_{12},
\end{equation}
where $Y_{-1}$ is the coefficient of $1/\zeta$ term in \eqref{eq: Y-inf}.
\end{lemma}
\begin{proof}
A combination of \eqref{def:limKer} and the Lax pair of $\Psi(\zeta)$ in \eqref{Lax-pair} yields
\begin{equation}
\frac{\partial}{\partial x}K_{s}^{(k)}(u,v;x,t)=\frac{1}{2\pi i}\Psi^{(2k)}_{1}(u)\Psi^{(2k)}_{1}(v).
\end{equation}
With the well-known property of trace-class operators, we have
\begin{equation}
\begin{split}
\frac{\partial}{\partial x}\log\det(I-\mathbb{K}_{s}^{(k)})&=-\text{tr}\left((I-\mathbb{K}_{s}^{(k)})^{-1}\frac{\partial \mathbb{K}_{s}^{(k)}}{\partial x}\right)\\
&=-\frac{1}{2\pi i}\int_{s}^{\infty}\left((I-\mathbb{K}_{s}^{(k)})^{-1}\Psi^{(2k)}_{1}\right)(u)\Psi^{(2k)}_{1}(u)du.
\end{split}
\end{equation}
From the explicit expression of $Y(\zeta)$ in \eqref{eq-def-Y}, we find $Y_{-1} = \int_s^\infty \vec{F}(u)\vec{h}^{\top}(u) du$. Recall that $\vec{F}(u)=(I-\mathbb{K}_{s}^{(k)})^{-1}\vec{f}(u)$, we obtain \eqref{eq-diff-identity-1} from the above formula.
\end{proof}

By combining the RH problems for $Y$ and $X$, we can get a new one with constant jump matrices. More precisely, as in \cite{CIK10}, we define
\begin{equation}
\begin{split}
X(\zeta)&=\begin{cases}
Y(\zeta)\Psi (\zeta),&\quad \text{$ \zeta \in \mathrm{I} \cup \mathrm{II} \cup \mathrm{III}$,
} \\
Y(\zeta)\Psi (\zeta)\left(\begin{matrix}
1&0\\1&1
\end{matrix}\right), &\quad \text{$\zeta \in \mathrm{IV}$,}\\
Y(\zeta)\Psi (\zeta)\left(\begin{matrix}
1&0\\-1&1
\end{matrix}\right), &\quad \text{$\zeta\in \mathrm{V}$,}
\end{cases} \quad \text{if } s<0,\\
X(\zeta)& =\begin{cases}
Y(\zeta)\Psi(\zeta),&\quad \text{$\zeta \in \mathrm{I} \cup \mathrm{II} \cup \mathrm{III}$,}\\
Y(\zeta)\Psi(\zeta)\left(\begin{matrix}
1&0\\-1&1
\end{matrix}\right), &\quad \text{$\zeta \in \mathrm{IV}$,}\\
Y(\zeta)\Psi(\zeta)\left(\begin{matrix}
1&0\\1&1
\end{matrix}\right), &\quad \text{$\zeta \in \mathrm{V}$,}
\end{cases} \quad \text{if } s>0,
\end{split}
\end{equation}
where the regions I--V are illustrated in Figure \ref{fig-tilde-X}. It is then readily seen that $X$ satisfies the following RH problem.

\paragraph{RH problem for $X$}

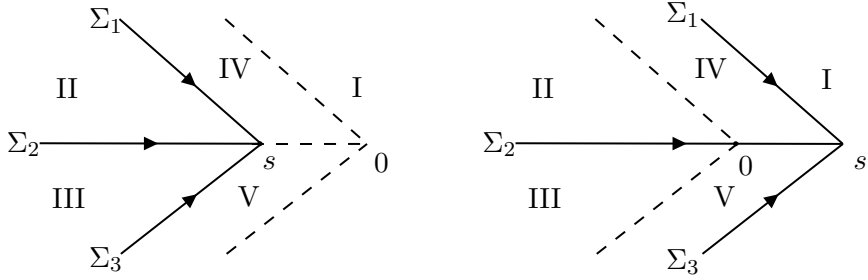
\begin{figure}[h]
\centering
\tikzset{every picture/.style={line width=0.75pt}} 

\begin{tikzpicture}[x=0.75pt,y=0.75pt,yscale=-1,xscale=1]

\draw    (141.67,44.59) -- (212.47,107.42) ;
\draw [shift={(179.91,78.53)}, rotate = 221.59] [fill={rgb, 255:red, 0; green, 0; blue, 0 }  ][line width=0.08]  [draw opacity=0] (7.14,-3.43) -- (0,0) -- (7.14,3.43) -- cycle    ;
\draw    (142.38,162.72) -- (212.47,107.42) ;
\draw [shift={(180.41,132.72)}, rotate = 141.73] [fill={rgb, 255:red, 0; green, 0; blue, 0 }  ][line width=0.08]  [draw opacity=0] (7.14,-3.43) -- (0,0) -- (7.14,3.43) -- cycle    ;
\draw  [fill={rgb, 255:red, 0; green, 0; blue, 0 }  ,fill opacity=1 ] (212.59,107.32) .. controls (212.59,106.91) and (212.22,106.58) .. (211.76,106.58) .. controls (211.31,106.58) and (210.94,106.91) .. (210.94,107.32) .. controls (210.94,107.72) and (211.31,108.05) .. (211.76,108.05) .. controls (212.22,108.05) and (212.59,107.72) .. (212.59,107.32) -- cycle ;
\draw  [dash pattern={on 4.5pt off 4.5pt}]  (194.46,44.59) -- (265.26,107.42) ;
\draw  [dash pattern={on 4.5pt off 4.5pt}]  (207.79,107.48) -- (265.26,107.42) ;
\draw  [dash pattern={on 4.5pt off 4.5pt}]  (195.17,162.72) -- (265.26,107.42) ;
\draw    (101.84,107.08) -- (212.59,107.32) ;
\draw [shift={(161.02,107.21)}, rotate = 180.12] [fill={rgb, 255:red, 0; green, 0; blue, 0 }  ][line width=0.08]  [draw opacity=0] (7.14,-3.43) -- (0,0) -- (7.14,3.43) -- cycle    ;
\draw  [dash pattern={on 4.5pt off 4.5pt}]  (379.21,44.59) -- (450.02,107.42) ;
\draw  [dash pattern={on 4.5pt off 4.5pt}]  (379.92,162.72) -- (450.02,107.42) ;
\draw  [fill={rgb, 255:red, 0; green, 0; blue, 0 }  ,fill opacity=1 ] (450.14,107.32) .. controls (450.14,106.91) and (449.77,106.58) .. (449.31,106.58) .. controls (448.85,106.58) and (448.48,106.91) .. (448.48,107.32) .. controls (448.48,107.72) and (448.85,108.05) .. (449.31,108.05) .. controls (449.77,108.05) and (450.14,107.72) .. (450.14,107.32) -- cycle ;
\draw    (432,44.59) -- (502.8,107.42) ;
\draw [shift={(470.24,78.53)}, rotate = 221.59] [fill={rgb, 255:red, 0; green, 0; blue, 0 }  ][line width=0.08]  [draw opacity=0] (7.14,-3.43) -- (0,0) -- (7.14,3.43) -- cycle    ;
\draw    (432.71,162.72) -- (502.8,107.42) ;
\draw [shift={(470.74,132.72)}, rotate = 141.73] [fill={rgb, 255:red, 0; green, 0; blue, 0 }  ][line width=0.08]  [draw opacity=0] (7.14,-3.43) -- (0,0) -- (7.14,3.43) -- cycle    ;
\draw    (339.38,107.08) -- (502.8,107.42) ;
\draw [shift={(424.89,107.26)}, rotate = 180.12] [fill={rgb, 255:red, 0; green, 0; blue, 0 }  ][line width=0.08]  [draw opacity=0] (7.14,-3.43) -- (0,0) -- (7.14,3.43) -- cycle    ;

\draw (267.26,110.42) node [anchor=north west][inner sep=0.75pt]   [align=left] {0};
\draw (125,37.85) node [anchor=north west][inner sep=0.75pt]   [align=left] {$\Sigma_1$};
\draw (83.92,100.05) node [anchor=north west][inner sep=0.75pt]   [align=left] {$\Sigma_2$};
\draw (125,157.86) node [anchor=north west][inner sep=0.75pt]   [align=left] {$\Sigma_3$};
\draw (211.59,112.17) node [anchor=north west][inner sep=0.75pt]   [align=left] {$s$};
\draw (256.14,72.26) node [anchor=north west][inner sep=0.75pt]   [align=left] {I};
\draw (108.58,72.95) node [anchor=north west][inner sep=0.75pt]   [align=left] {II};
\draw (106.58,128.36) node [anchor=north west][inner sep=0.75pt]   [align=left] {III};
\draw (189.79,62.56) node [anchor=north west][inner sep=0.75pt]   [align=left] {IV};
\draw (199.2,127.2) node [anchor=north west][inner sep=0.75pt]   [align=left] {V};
\draw (449.13,112.17) node [anchor=north west][inner sep=0.75pt]   [align=left] {0};
\draw (413,34.69) node [anchor=north west][inner sep=0.75pt]   [align=left] {$\Sigma_1$};
\draw (321.46,100.05) node [anchor=north west][inner sep=0.75pt]   [align=left] {$\Sigma_2$};
\draw (413,159.01) node [anchor=north west][inner sep=0.75pt]   [align=left] {$\Sigma_3$};
\draw (506.6,112.12) node [anchor=north west][inner sep=0.75pt]   [align=left] {$s$};
\draw (490.69,68.26) node [anchor=north west][inner sep=0.75pt]   [align=left] {I};
\draw (346.12,72.95) node [anchor=north west][inner sep=0.75pt]   [align=left] {II};
\draw (344.12,128.36) node [anchor=north west][inner sep=0.75pt]   [align=left] {III};
\draw (427.33,62.56) node [anchor=north west][inner sep=0.75pt]   [align=left] {IV};
\draw (436.75,127.2) node [anchor=north west][inner sep=0.75pt]   [align=left] {V};

\end{tikzpicture}

\caption{Regions I--V and jump contours of the RH problem for $X$. The case $s<0$ is presented on the left and the case $s>0$ is presented on the right.} \label{fig-tilde-X}
\end{figure}

\begin{enumerate}
\item [(a)] $X(\zeta)$ is analytic in $\mathbb{C}\setminus\Sigma$, where $\Sigma:=\Sigma_{1}\cup\Sigma_{2}\cup\Sigma_{3}$ is shown in Figure \ref{fig-tilde-X}.
\item [(b)] On $\Sigma$, the limiting values $X_{\pm}(\zeta)$ exist and satisfy the following jump conditions
\begin{equation}
X_{+}(\zeta)=X_{-}(\zeta)\begin{cases}
    \left(\begin{matrix}
1&0\\1&1
\end{matrix}\right), \qquad &\zeta\in\Sigma_{1}\cup\Sigma_{3},
\\
\left(\begin{matrix}
0&1\\-1&0
\end{matrix}\right), \qquad & \zeta\in\Sigma_{2}.\end{cases}
\end{equation} 
\item [(c)] As $\zeta\to\infty$, we have  
\begin{equation}\label{eq-tilde-X-by-X}
X(\zeta)=\left(I+\frac{X_{-1}}{\zeta}+\mathcal{O}(\zeta^{-2})\right)\zeta^{-\frac{1}{4}\sigma_{3}}Ne^{-\theta(\zeta;x)\sigma_{3}},
\end{equation}
where $N$ and $\theta$ are defined in \eqref{eq-def-N} and \eqref{eq-def-theta}, 
\begin{equation}\label{eq-relation-X-Y-Psi}
X_{-1}=\Psi_{-1}+Y_{-1}
\end{equation}
with $\Psi_{-1}$ and $Y_{-1}$ given in \eqref{eq-def-Psi--1} and \eqref{eq: Y-inf}, respectively.
\item [(d)] As $\zeta\to s$, we have $X(\zeta)=\mathcal{O}(\log(\zeta-s))$.
\end{enumerate}

As shown in \cite[Equation (2.17)]{CIK10}, the second differential identity reads as follows. 

\begin{lemma}\label{lem-differential-indentity-2}
For any $x,s\in\mathbb{R}$, we have
\begin{equation}\label{eq-differential-indentity-2}
\frac{\partial F}{\partial s}(s;x) = \frac{\partial}{\partial s}\log\det(I-\mathbb{K}_{s}^{(k)})=\lim\limits_{\zeta\to s}\frac{1}{2\pi i}(X(\zeta)^{-1}X'(\zeta))_{21},
\end{equation}
where the limit is taken as $\zeta\to s$ from sector $\mathrm{I}$ in Fig \ref{fig-tilde-X}. 
\end{lemma}

The differential identities in Lemmas \ref{lem-differential-indentity-1} and \ref{lem-differential-indentity-2} establish connections between $\frac{\partial F}{\partial s}$ and $\frac{\partial F}{\partial x}$ with specific entries of a RH problem. To prove Lemmas \ref{lem-asym-dF/dmu}--\ref{lem-dF/dtau-case-III} presented in Section \ref{Sec: main-proof}, we will employ the Deift-Zhou nonlinear steepest descent method to study asymptotics of the above RH problem for $X$. It is important to note that the asymptotic results in Lemmas \ref{lem-asym-dF/dmu}--\ref{lem-dF/dtau-case-III} hold uniformly as $|s|$ and/or $|x|$ tend to $+\infty$. While one could conduct separate asymptotic analysis for cases like $s \to -\infty$ in Lemma \ref{lem-asym-dF/dmu} or $x \to +\infty$ and $x \to -\infty$ in Lemmas \ref{lem-dF/dtau-case-I}--\ref{lem-dF/dtau-case-III}, such an approach may lead to redundant analyses due to similarities. Therefore, we adopt a unified approach by introducing a new large variable $\lambda >0$ and setting
\begin{equation} \label{s-r-and-x-y}
s= r \lambda \quad \textrm{and} \quad x=y\lambda^{2k+1},
\end{equation}
with $r\geq -1$ and $y\in\mathbb{R}$ being bounded parameters. In this way, the three regions in Lemmas \ref{lem-asym-dF/dmu}--\ref{lem-dF/dtau-case-III} can be characterized as follows (see Remark \ref{rem-three-regions} below for a more detailed explanation on why these three regions are divided):
\begin{equation}\label{eq-three-cases-c-r-plane}
\begin{split}
\text{algebraic growth region: }& \alpha_{k}r^{2k+1}+y\in\left[-M,-C_{1}\lambda^{-k-1+\delta}\right],\\
\text{transition region: }& \alpha_{k}r^{2k+1}+y\in\left[-C_{1}\lambda^{-k-1+\delta},C_{1}\lambda^{-k-1+\delta}\right],\\
\text{exponential decay region: }& \alpha_{k}r^{2k+1}+y\in\left[C_{1}\lambda^{-k-1+\delta},+\infty\right),\\ 
\end{split}
\end{equation} 
where $\alpha_k$ is defined in \eqref{eq-def-Ak}, $M>0,$ $C_{1}>0$, and $\delta \in \mathbb{R}$ are arbitrary fixed constants with $
\delta\in\left(-\frac{k}{3},k+1\right)$ for the algebraic growth case and $\delta\in\left(0,\frac{1}{4}\right)$ for the transition and exponential decay cases.

Before proceeding with the asymptotic analysis, we note that our analysis requires the presence of at least one large variable, either $s$ or $x$. This condition implies that \begin{equation}\label{eq-r-y-not-both-vanish}
    |r|+|y|\geq\delta_{0}
\end{equation}
for some fixed constant $\delta_{0}>0$.

%

%

\section{Asymptotic analysis in the algebraic growth region} \label{Sec:Alg}
In this section, we assume that $\alpha_{k}r^{2k+1}+y\in\left[-M,-C_{1}\lambda^{-k-1+\delta}\right]$, where $M$ and $C$ are two fixed positive constants with $\delta\in\left(-\frac{k}{3},k+1\right)$. By choosing $r$ and $y$ appropriately, we will prove Lemmas \ref{lem-asym-dF/dmu} and \ref{lem-dF/dtau-case-I} by carrying out asymptotic analysis of the RH problem for $X$. The relevant asymptotic analysis bears some similarity to that performed in \cite[Section 3]{CIK10}, corresponding to $r=-1$ and $y=0$ in the notation introduced in  \eqref{s-r-and-x-y}. The main difference is that, instead of $\alpha_{k}r^{2k+1}+y=-\alpha_k$ as in \cite[Section 3]{CIK10}, our asymptotic analysis is uniformly valid for $\alpha_{k}r^{2k+1}+y$ within a large interval $\left[-M,-C_{1}\lambda^{-k-1+\delta}\right]$, where the right ending point can approach 0 as $\lambda \to +\infty$. Thus, some careful uniform treatments have to be conducted in this context. We start with the introduction of the so-called $g$-function. 


\subsection{Construction of the $g$-function}
By introducing the scaling $\zeta=\lambda\eta$ and the change of variables in \eqref{s-r-and-x-y}, the intersection point $\zeta = s$ in Figure \ref{fig-tilde-X} is mapped to $\eta = r$. Moreover, we have from \eqref{eq-def-theta} that 
\begin{equation}\label{def:thetaeta}
\theta(\zeta;x):=\theta(\zeta;x,0,\ldots,0)=\theta(\lambda\eta; y\lambda^{2k+1})=\lambda^{\frac{4k+3}{2}}\left[\frac{4}{4k+3}\eta^{\frac{4k+3}{2}}+y\eta^{\frac{1}{2}}\right]=\lambda^{\frac{4k+3}{2}}\theta(\eta;y).
\end{equation}
In order to normalize the RH problem for $X$, we introduce the $g$-function as follows:
\begin{equation}\label{eq-def-g1-asym}
g_{1}(\eta)=\sum\limits_{j=0}^{2k+1}(-1)^{j+1}b_{j}(\eta-r)^{\frac{1}{2}+j},\qquad \eta \in \mathbb{C}\setminus (-\infty,r],
\end{equation}  
where 
\begin{align}
\label{eq-def-b0}
b_{0}&=-\alpha_{k}r^{2k+1}-y,\\
b_{j}&=(-1)^{j+1}\frac{\Gamma\left(2k+2\right)\Gamma\left(\frac{3}{2}\right)}{\Gamma\left(2k+2-j\right)\Gamma\left(j+\frac{3}{2}\right)}\alpha_{k}r^{2k+1-j}. \label{eq-def-b0-bi}
\end{align}
Then, it is straightforward to verify that
\begin{equation}\label{eq-matching-g1-theta}
\theta(\eta;y)-g_{1}(\eta)=d_{1}\eta^{-\frac{1}{2}}+\mathcal{O}(\eta^{-\frac{3}{2}}),\qquad \eta \to \infty, 
\end{equation}
where
\begin{equation}
\label{eq-def-d1}
d_{1}=\frac{\alpha_{k}}{4(k+1)}r^{2k+2}+\frac{ry}{2}.
\end{equation}
\begin{remark}
It is also natural to rewrite $g_{1}(\eta)$ and $g_{1}'(\eta)$ in the form
\begin{equation}\label{eq-def-g1(eta)-g1'(eta)}
g_{1}(\eta)=(\eta-r)^{\frac{1}{2}}p_{1}(\eta),\qquad g_{1}'(\eta)=(\eta-r)^{-\frac{1}{2}}\tilde{p}_{1}(\eta),
\end{equation}
where both $p_{1}(\eta)$ and $\tilde{p}_{1}(\eta)$ are polynomials of degree $2k+1$. From the matching condition \eqref{eq-matching-g1-theta}, we conclude that 
\begin{align}
p_{1}(\eta)&=\sum\limits_{j=0}^{2k+1}\frac{4}{4k+3}\frac{\Gamma\left(j+\frac{1}{2}\right)}{\Gamma\left(j+1\right)\Gamma\left(\frac{1}{2}\right)}r^{j}\eta^{2k+1-j}+y,  \label{eq-def-p1}\\
\tilde{p}_{1}(\eta)&= 2\eta^{2k+1}\sum\limits_{j=0}^{2k+1}\frac{\Gamma\left(j-\frac{1}{2}\right)}{\Gamma\left(j+1\right)\Gamma\left(-\frac{1}{2}\right)}\left(\frac{r}{\eta}\right)^{j}+\frac{y}{2}.  \label{eq-def-tilde-p1}
\end{align}
\end{remark}
%

%

\begin{remark}\label{rem-tj-not-vanishe}
In Theorem \ref{Thm-asym-PI2k-determinant}, it is assumed that $t_{1}=t_{2}=\cdots=t_{2k-1}=0$. When considering general parameters $t_{j}$, the function $\theta(\zeta)$ in \eqref{eq-def-theta} involves additional terms dependent on $t_{j}$. Consequently, when we define the $g$-function in a manner similar to \eqref{eq-def-g1-asym}, the coefficients $b_j$ in \eqref{eq-def-b0-bi} and $d_1$ in \eqref{eq-def-d1} also depend on the parameters $t_j$. In particular, when $k=1$, we have
\begin{equation}
    d_{1}=\frac{5}{32}r^4+\frac{ry}{2}+\frac{t_{1}r^2}{2|s|^2}, \qquad b_{0}=-\frac{5}{4}r^{3}-y-\frac{2t_{1}r}{|s|^2},\qquad b_{1}=\frac{5}{2}r^2+\frac{4t_{1}}{3|s|^2}.
\end{equation}
As $s \to -\infty$, the additional factors dependent on $t_1$ are relatively small compared to the leading terms in \eqref{eq-def-b0-bi}. This suggests that we can essentially adopt the same approach as in the present paper to study the general case when $t_{j}\neq 0$ and obtain results similar to that in Corollary \ref{cor-asym-PI2k-determinant}. 
\end{remark}

\begin{remark} \label{rem-three-regions}
From \eqref{eq-def-b0-bi}, it is clear that the function $g_1(\eta)$ has different properties (particularly the sign of $\re g_1(\eta)$) near $\eta = r$, when the leading coefficient $b_0$ varies with respect to $r$ and $y$. This difference significantly influences our asymptotic analysis, leading to very different results. This is the reason why we characterize three distinct regions based on $\alpha_{k}r^{2k+1}+y$ in \eqref{eq-three-cases-c-r-plane}.
\end{remark}

\begin{lemma}\label{lem-g1(eta)}
There exists $\theta_{0}>0$ such that, for any bounded $M>0$ and $r$ with $\alpha_{k}r^{2k+1}+y\in[-M,0]$, we have
\begin{equation} \label{eq:lemma-g1}
\Re g_{1}(\eta)<0,
\end{equation}
where $\arg{(\eta-r)}\in(-\pi,-\pi+\theta_{0}]\cup[\pi-\theta_{0},\pi)$.
\end{lemma}

\begin{proof}
To prove \eqref{eq:lemma-g1}, it is sufficient to show that, for $\alpha_{k}r^{2k+1}+y\in[-M,0]$,  
\begin{equation}\label{eq-Im-g1+b0+}
\Im g'_{1,\pm}(\eta)>0, \qquad  \eta \in (-\infty, r).
\end{equation}
It follows from \eqref{eq-def-g1-asym} and \eqref{eq-def-g1(eta)-g1'(eta)} that
\begin{equation}\label{eq-g1-derivative}
g_{1}'(\eta)=(\eta-r)^{-\frac{1}{2}}\left[\sum\limits_{j=0}^{2k+1}(-1)^{j+1}\left(j+\frac{1}{2}\right)b_{j}(\eta-r)^{j}\right]:=(\eta-r)^{-\frac{1}{2}}\tilde{p}_{1}(\eta),
\end{equation}
where $\tilde{p}_{1}(\eta)$ is a polynomial in $\eta$ with degree $2k+1$. Obviously, we have  
\begin{equation}
\begin{split}
\label{eq-q(s)}
\tilde{p}_{1}(r)=-\frac{b_{0}}{2}=\frac{\alpha_{k}r^{2k+1}}{2}+\frac{y}{2}\leq 0.
\end{split}
\end{equation}
Next, to establish \eqref{eq-Im-g1+b0+}, we will show that $\tilde{p}_{1}(\eta)<0$ holds for all  $\eta\in (-\infty, r)$.

When $r\leq 0$, it follows from the definition of $b_{j},j=0,1,\cdots,2k+1,$ in \eqref{eq-def-b0-bi} that $b_{j}\geq 0$ for all $j$. Given that $|r|+|y|\geq\delta_{0} > 0$ (see the explanation at the end of Section \ref{Sec:Pre}), we immediately have $\tilde{p}_{1}(\eta)<0$ for $\eta\in (-\infty, r)$. 

When $r>0$, we note the explicit expression of  $\tilde{p}_{1}(\eta)$ in \eqref{eq-def-tilde-p1} and first consider the function
\begin{equation}
\varphi(z)=\sum\limits_{j=0}^{2k+1}\frac{\Gamma\left(j-\frac{1}{2}\right)}{\Gamma\left(j+1\right)\Gamma\left(-\frac{1}{2}\right)}z^{j}.
\end{equation}
It is easy to verify that $\varphi(z)$ is a monotonically decreasing function for $z\in(1,+\infty)$ and $\varphi(1)>0$. Although the above series is alternating when $z<0$, we can still show that $\varphi(z) > 0$ for $z\in(-\infty,0)$. To see this, when $z \in (-1,0)$, as the absolute values of the coefficients decrease with respect to $j$, we have $\varphi(z)>0$. When $z \in (-\infty, -1)$, we consider the second order derivative
\begin{equation}
\varphi''(z)=\sum\limits_{j=0}^{2k-1}\frac{\Gamma\left(j+\frac{3}{2}\right)}{\Gamma\left(j+1\right)\Gamma\left(-\frac{1}{2}\right)}z^{j},
\end{equation}
which is also an alternating series. In this case, the absolute values of the coefficients increase with respect to $j$. This gives us $\varphi''(z) > 0$ for $z \in (-\infty, -1)$ and 
\begin{equation}
\varphi'(z)<\varphi'(-1)=\sum\limits_{j=0}^{2k}\frac{\Gamma\left(j+\frac{1}{2}\right)}{\Gamma\left(j+1\right)\Gamma\left(-\frac{1}{2}\right)}(-1)^{j}<0,\quad  \forall z\in(-\infty,-1).
\end{equation}
Given that $\varphi(-1) > 0$, we conclude that  $\varphi(z) > 0$ for $z\in(-\infty,-1].$

Based on the above properties for $\varphi(z)$, we can have that when $\eta\in(0,r)$ ({\it i.e.}, $ z= \frac{r}{\eta} \in (1,+\infty)$),
\begin{equation}
\begin{split}
\tilde{p}_{1}(\eta)<2\eta^{2k+1} \varphi(1)+\frac{y}{2} <  \tilde{p}_{1}(r)\leq 0, \qquad \forall \eta\in(0,r).
\end{split}
\end{equation}
When $\eta\in(-\infty,0)$ ({\it i.e.}, $z=\frac{r}{\eta}\in(-\infty,0)$),  we get
\begin{equation}
\tilde{p}_{1}(\eta)= 2\eta^{2k+1} \varphi\left(\frac{r}{\eta}\right)+\frac{y}{2}<\frac{y}{2},\quad \forall \eta\in(-\infty,0). 
\end{equation}
Since $\alpha_{k}r^{2k+1}+y\leq 0$,  it follows that $y<0$. Therefore, we obtain $\tilde{p}_{1}(\eta)<0$ for $\eta\in(-\infty,0)$. This finishes the proof of the lemma.
\end{proof} 

\begin{remark}\label{rem-g1-wide-applied}
The above lemma holds for $\alpha_{k}r^{2k+1}+y\in[-M,0]$, a condition that is stronger than $\alpha_{k}r^{2k+1}+y\in[-M,-C_{1}\lambda^{-k-1+\delta}]$ in the present section. Therefore, $g_{1}(\eta)$ can serve as a suitable $g$-function for the subsequent analysis in Section \ref{sec:RH-analysis-case-II-} where $\alpha_{k}r^{2k+1}+y\in[-C_{1}\lambda^{-k-1+\delta},0]$. 
\end{remark}

\subsection{Riemann-Hilbert analysis when $\alpha_{k}r^{2k+1}+y\in[-M,-C_{1}\lambda^{-k-1+\delta}]$}
\label{sec:RH-analysis-case-I}
\paragraph{First transformation}
The first transformation is a scaling and partial normalization of the RH problem for $X$. It is defined by 
\begin{equation}\label{eq-transform-tilde-X-T}
T(\eta)=(\lambda b_{0})^{\frac{1}{4}\sigma_{3}}\left(\begin{matrix}
1&0\\-d_{1}\lambda^{2k+2}&1
\end{matrix}\right)X(\lambda (b_{0}\eta+r))\exp\{\lambda^{\frac{4k+3}{2}}g_{1}(b_{0}\eta+r)\sigma_{3}\},
\end{equation}
where the function $g_1$ is defined in \eqref{eq-def-g1-asym}, $b_0$ and $d_1$ are two constants given in \eqref{eq-def-b0} and \eqref{eq-def-d1}, respectively. Here, we rescale $\eta$ to $\lambda b_{0}
\eta+r$ rather than $
\lambda\eta+r$, which is different from the counterpart used in \cite[Equation (3.1)]{CIK10}. This modification enables us to  perform uniform asymptotic analysis on $X$ for $\alpha_{k}r^{2k+1}+y=-b_{0}
\in [-M,-C_{1}\lambda^{-k-1+\delta}]$ with $
\delta\in(-\frac{k}{3},k+1)$. In view of the RH problem for $X$, it is readily seen that $T$ satisfies the following RH problem.

\paragraph{RH problem for $T$}
\begin{enumerate}
\item [(a)] $T(\eta)$ is analytic in $\mathbb{C}\setminus\Sigma^{T}$, where $\Sigma^{T}:=\Sigma_{1}^{T}\cup\Sigma_{2}^{T}\cup \Sigma_{3}^{T}$ is illustrated in Figure \ref{fig-T}.
\item [(b)] On $\Sigma^T$, the limiting values $T_{\pm}(\eta)$ exist and satisfy the jump condition
\begin{equation}
T_{+}(\eta)=T_{-}(\eta)\begin{cases}\left(\begin{matrix}
1&0\\ \exp\{2\lambda^{\frac{4k+3}{2}}g_{1}(b_{0}\eta+r)\}&1
\end{matrix}\right),\quad &\eta\in\Sigma_{1}^{T}\cup\Sigma_{3}^{T},\\
\left(\begin{matrix}
0&1\\ -1&0
\end{matrix}\right),\quad &\eta\in\Sigma_{2}^{T}=(-\infty,0).
\end{cases}
\end{equation}
\item [(c)] As $\eta\to\infty$, we have 
\begin{equation}\label{eq-asym-S-infty}
T(\eta)=\left(I+\frac{T_{-1}}{\eta}+\mathcal{O}(\eta^{-2})\right)\eta^{-\frac{1}{4}\sigma_{3}}N,
\end{equation}
where $N$ is defined in \eqref{eq-def-N} and 
\begin{equation}\label{eq-relation-T-X}
(T_{-1})_{12}=\left((X_{-1})_{12}+ d_{1}\lambda^{2k+2}\right)\lambda^{-\frac{1}{2}}b_{0}^{-\frac{1}{2}}.
\end{equation}
\item [(d)] As $\eta\to 0$, we have $T(\eta)=\mathcal{O}(\log{\eta})$.
\end{enumerate}

\begin{figure}
    \centering
    
\begin{tikzpicture}[x=0.75pt,y=0.75pt,yscale=-1,xscale=1]

\draw    (271.5,32) -- (371.5,132) ;
\draw [shift={(325.04,85.54)}, rotate = 225] [fill={rgb, 255:red, 0; green, 0; blue, 0 }  ][line width=0.08]  [draw opacity=0] (8.93,-4.29) -- (0,0) -- (8.93,4.29) -- cycle    ;
\draw    (272.5,220) -- (371.5,132) ;
\draw [shift={(325.74,172.68)}, rotate = 138.37] [fill={rgb, 255:red, 0; green, 0; blue, 0 }  ][line width=0.08]  [draw opacity=0] (8.93,-4.29) -- (0,0) -- (8.93,4.29) -- cycle    ;
\draw  [fill={rgb, 255:red, 0; green, 0; blue, 0 }  ,fill opacity=1 ] (371.67,131.83) .. controls (371.67,131.19) and (371.15,130.66) .. (370.5,130.66) .. controls (369.85,130.66) and (369.33,131.19) .. (369.33,131.83) .. controls (369.33,132.48) and (369.85,133) .. (370.5,133) .. controls (371.15,133) and (371.67,132.48) .. (371.67,131.83) -- cycle ;
\draw    (226.33,130.83) -- (369.33,131.83) ;
\draw [shift={(302.83,131.37)}, rotate = 180.4] [fill={rgb, 255:red, 0; green, 0; blue, 0 }  ][line width=0.08]  [draw opacity=0] (8.93,-4.29) -- (0,0) -- (8.93,4.29) -- cycle    ;

\draw (373.5,135) node [anchor=north west][inner sep=0.75pt]   [align=left] {0};
\draw (244,19) node [anchor=north west][inner sep=0.75pt]   [align=left] {$\displaystyle {\textstyle \Sigma_{1}^{T}}$};
\draw (199,125) node [anchor=north west][inner sep=0.75pt]   [align=left] {$\displaystyle {\textstyle \Sigma_{2}^{T}}$};
\draw (248,218) node [anchor=north west][inner sep=0.75pt]   [align=left] {$\displaystyle {\textstyle \Sigma_{3}^{T}}$};

\end{tikzpicture}
    \caption{The jump contour $\Sigma^{T}$ of the RH problem for $T$.}
    \label{fig-T}
\end{figure}
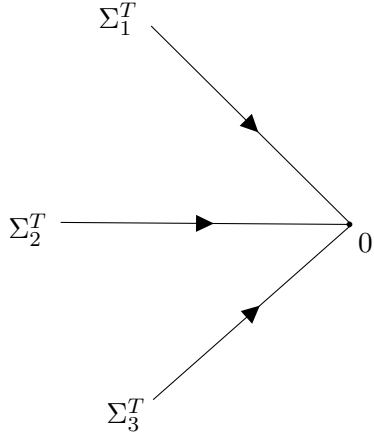

By Lemma \ref{lem-g1(eta)}, it follows that there exists $\theta_{0}>0$ such that $\re g_{1}(\eta)<0$ when $\arg(\eta-r)\in(-\pi,-\pi+\theta_{0}]\cup[\pi-\theta_{0},\pi)$. Thus, there exists a positive
constant $M_1$ such that for $\eta\in\Sigma_{1}^{T}\cup\Sigma_{3}^{T}$,
\begin{equation}\label{eq-equality-g1}
\lambda^{\frac{4k+3}{2}}\Re g_{1}(b_{0}\eta+r)
\leq -M_{1}\lambda^{\frac{k}{2}+\frac{3}{2}\delta}<0,
\end{equation}
provided that $\eta$ is bounded away from the origin. This particularly implies that the jump matrix of $T$ tends to the identity matrix exponentially fast on $\Sigma_{1}^{T}\cup \Sigma_{3}^{T}$ as $\lambda \to +\infty$. We are then in the stage of construction of global and local parametrices.


\paragraph{Global parametrix}
The global parametrix $T^{(\infty)}$ reads as follows.
\paragraph{RH problem for $T^{(\infty)}$}
\begin{enumerate}

\item [(a)] $T^{(\infty)}(\eta)$ is analytic in $\mathbb{C}\setminus (-\infty,0]$. 

\item [(b)] $T^{(\infty)}$ satisfies the jump condition
\begin{equation}
T^{(\infty)}_{+}(\eta)=T^{(\infty)}_{-}(\eta) \left(\begin{matrix}
0&1\\ -1&0
\end{matrix}\right),
\qquad \eta\in(-\infty,0).
\end{equation}
\item [(c)] As $\eta\to\infty$, we have 
\begin{equation}
T^{(\infty)}(\eta)=\left(I+\mathcal{O}(\eta^{-1})\right)\eta^{-\frac{1}{4}\sigma_{3}}N.
\end{equation}
\end{enumerate}
It is easily seen that 
\begin{equation}\label{eq-def-S-infty}
T^{(\infty)}(\eta)=\eta^{-\frac{1}{4}\sigma_{3}}N,\qquad
\eta\in \mathbb{C}\setminus (-\infty,0],
\end{equation}
solves the above RH problem.


\paragraph{Local parametrix}
The non-uniform convergence of the jump matrix of $T$ on $\Sigma_{1}^{T}\cup \Sigma_{3}^{T}$ to the identity matrix
suggests the following local parametrix in a small neighborhood of the origin.

\paragraph{RH problem for $T^{(0)}$}
\begin{enumerate}
\item [(a)] $T^{(0)}(\eta)$ is analytic in $\Sigma^T \setminus U(0;\rho_1)$, where $U(0;\rho_1)$ defined in \eqref{def:dz0r}
is an open disk centered at $0$ with fixed and small radius $\rho_1>0$.

\item [(b)] On $\Sigma^T \cap U(0;\rho_1) $, the limiting values $T^{(0)}_{\pm}(\eta)$ exist and satisfy the jump condition
\begin{equation}
T^{(0)}_{+}(\eta)=T^{(0)}_{-}(\eta)\begin{cases}\left(\begin{matrix}
1&0\\ \exp\{2\lambda^{\frac{4k+3}{2}}g_{1}(b_{0}\eta+r)\}&1
\end{matrix}\right),\quad &\eta\in U(0;\rho_1)\cap(\Sigma_{1}^{T}\cup\Sigma_{3}^{T}),\\
\left(\begin{matrix}
0&1\\ -1&0
\end{matrix}\right),\quad &\eta\in U(0;\rho_1)\cap \Sigma_{2}^{T}.
\end{cases}
\end{equation}
\item [(c)] As $\eta\to 0$, we have $T^{(0)}(\eta)=\mathcal{O}(\log{\eta})$.
\item [(d)] As $\lambda \to\infty$, $T^{(0)}$ satisfies the matching condition 
\begin{equation}
    T^{(0)}(\eta)= (I+o(1))T^{(\infty)}(\eta), \qquad \eta \in \partial U(0;\rho_1). 
\end{equation}

\end{enumerate}

As in \cite{CIK10}, one can solve the above RH problem by using a modified Bessel parametrix. More precisely, we introduce
\begin{equation}\label{eq-def-f1(eta)}
f_{1}(\eta):=b_{0}^{-3}\left(-g_{1}(b_{0}\eta+r)\right)^2=\eta\left[1+\sum\limits_{j=1}^{2k+1}(-1)^{j}b_{j}b_{0}^{j-1}\eta^{j}\right]^2,
\end{equation}
which is a conformal mapping in $U(0;\rho_{1})$, and then define
\begin{equation}\label{def:T0caseI}
T^{(0)}(\eta)=E_{1}(\eta)\Phi^{(\mathrm{Bes})}(\lambda^{4k+3}b_{0}^3f_{1}(\eta))\exp\{\lambda^{\frac{4k+3}{2}}g_{1}(b_{0}\eta+r)\sigma_{3}\},
\end{equation}
where 
\begin{equation}\label{eq-def-E1}
E_{1}(\eta):=\eta^{-\frac{1}{4}\sigma_{3}}(\lambda^{4k+3}b_{0}^{3}f_{1}(\eta))^{\frac{1}{4}\sigma_{3}}
\end{equation}
is a prefactor analytic in $U(0;\rho_1)$
and $\Phi^{(\mathrm{Bes})}$ is a modified Bessel parametrix used in  \cite[Equations (3.22)--(3.24)] {CIK10}
inspired by the one in \cite{KMVV,KV02}. Since $\Phi^{(\mathrm{Bes})}$ satisfies the jump condition
\begin{equation}\label{eq:PhiBesjump}
\Phi^{(\mathrm{Bes})}_{+}(z)=\Phi^{(\mathrm{Bes})}_{-}(z)\begin{cases}\left(\begin{matrix}
1&0
\\ 
1&1
\end{matrix}\right),\quad &z \in\Sigma_{1}^{T}\cup\Sigma_{3}^{T},\\
\left(\begin{matrix}
0&1\\ -1&0
\end{matrix}\right),\quad &z \in\Sigma_{2}^{T},
\end{cases}
\end{equation} 
the asymptotic behaviors
\begin{equation}\label{eq-asym-bessel}
\begin{split}
\Phi^{(\mathrm{Bes})}(z)&=z^{-\frac{1}{4}\sigma_{3}}N\left[I+\frac{1}{8\sqrt{z}}\left(\begin{matrix}
-1&-2i\\-2i&1
\end{matrix}\right)+\mathcal{O}(z^{-1})\right]e^{z^{\frac{1}{2}}\sigma_{3}}, \quad z\to\infty,
\end{split}
\end{equation}
and
\begin{equation}\label{eq:asym-Bessel-parametrix-near-0}
\Phi^{(\mathrm{Bes})}(z) = \left(\frac{\sqrt{\pi}}{e^{\pi i/4}}\right)^{\sigma_{3}}\left[\left(\begin{matrix}
1&\frac{\gamma_{\mathrm{E}}-\log{2}}{\pi i}\\0&1\end{matrix}\right)+\mathcal{O}\left(z\right)\right]\left(\begin{matrix}
1&\frac{\log{z}}{2\pi i}\\0&1\end{matrix}\right),\quad z\to 0, \quad \arg{z}\in\left(-\frac{2\pi}{3},\frac{2\pi}{3}\right),
\end{equation}
with $\gamma_{\mathrm{E}}$ being Euler's constant, one can check that $T^{(0)}$ defined in \eqref{def:T0caseI} indeed solves the local parametrix near the origin. 

\paragraph{Final transformation}
With the aid of the global and local parametrices built in \eqref{eq-def-S-infty} and \eqref{def:T0caseI}, the final transformation is defined by 
\begin{equation}\label{eq-def-R-by-S}
R(\eta)=\begin{cases}
T(\eta)T^{(0)}(\eta)^{-1},\quad &\eta\in U(0;\rho_{1}),\\
T(\eta)T^{(\infty)}(\eta)^{-1},\quad &\eta\in \mathbb{C}\setminus U(0;\rho_{1}).
\end{cases}
\end{equation}
It is straightforward to check that $R$ satisfies the following RH problem.

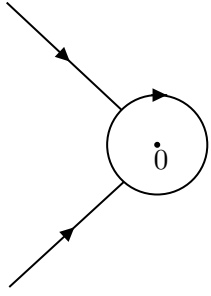
\begin{figure}

\centering

\tikzset{every picture/.style={line width=0.75pt}} 

\begin{tikzpicture}[x=0.75pt,y=0.75pt,yscale=-1,xscale=1]

\draw    (178,82) -- (235.53,136.16) ;
\draw [shift={(209.53,111.69)}, rotate = 223.27] [fill={rgb, 255:red, 0; green, 0; blue, 0 }  ][line width=0.08]  [draw opacity=0] (7.14,-3.43) -- (0,0) -- (7.14,3.43) -- cycle    ;
\draw    (236.53,172.16) -- (179.28,225.09) ;
\draw [shift={(211.8,195.03)}, rotate = 137.24] [fill={rgb, 255:red, 0; green, 0; blue, 0 }  ][line width=0.08]  [draw opacity=0] (7.14,-3.43) -- (0,0) -- (7.14,3.43) -- cycle    ;
\draw   (228.14,153.55) .. controls (228.14,139.74) and (239.34,128.55) .. (253.14,128.55) .. controls (266.95,128.55) and (278.14,139.74) .. (278.14,153.55) .. controls (278.14,167.35) and (266.95,178.55) .. (253.14,178.55) .. controls (239.34,178.55) and (228.14,167.35) .. (228.14,153.55) -- cycle ;
\draw  [fill={rgb, 255:red, 0; green, 0; blue, 0 }  ,fill opacity=1 ] (252.14,153.55) .. controls (252.14,152.99) and (252.59,152.55) .. (253.14,152.55) .. controls (253.69,152.55) and (254.14,152.99) .. (254.14,153.55) .. controls (254.14,154.1) and (253.69,154.55) .. (253.14,154.55) .. controls (252.59,154.55) and (252.14,154.1) .. (252.14,153.55) -- cycle ;
\draw    (250.14,128.55) -- (258.14,128.55) ;
\draw [shift={(257.94,128.55)}, rotate = 180] [fill={rgb, 255:red, 0; green, 0; blue, 0 }  ][line width=0.08]  [draw opacity=0] (7.14,-3.43) -- (0,0) -- (7.14,3.43) -- cycle    ;

\draw (250,155) node [anchor=north west][inner sep=0.75pt]   [align=left] {0};

\end{tikzpicture}
\caption{The jump contour $\Sigma^{R}$ of the RH problem for $R$ in Section \ref{sec:RH-analysis-case-I}.} 
\label{fig-R-bessel}
\end{figure}

\paragraph{RH problem for $R$}
\begin{enumerate}
\item [(a)] $R(\eta)$ is analytic in $\mathbb{C}\setminus \Sigma^{R}$, where $\Sigma^{R}:=\Sigma_{1}^{T}\cup\Sigma_{3}^{T}\cup \partial U(0;\rho_{1})\setminus U(0;\rho_{1})$ is depicted in Figure \ref{fig-R-bessel}.
\item [(b)] On $\Sigma^{R}$, the limiting values $R_{\pm}(\eta)$ exist and satisfy the jump condition $$R_{+}(\eta)=R_{-}(\eta)v_{R}(\eta),$$ 
where
\begin{eqnarray}
v_{R}(\eta)=
\begin{cases}
T^{(\infty)}(\eta)\left(\begin{matrix}
1&0\\e^{2|s|^{\frac{4k+3}{2}}g_{1}(\eta)}&1
\end{matrix}\right)T^{(\infty)}(\eta)^{-1},  &\eta\in\Sigma^{R}\setminus \partial U(0;\rho_{1}),\\
T^{(0)}(\eta)T^{(\infty)}(\eta)^{-1}, &\eta\in\partial U(0;\rho_{1}).
\end{cases}
\end{eqnarray}
\item [(c)] As $\eta\to\infty$, we have $R(\eta)=I+\frac{T_{-1}}{\eta} +\mathcal{O}(\eta^{-2})$, where $T_{-1}$ is defined in \eqref{eq-asym-S-infty}. 
\end{enumerate}

From \eqref{eq-equality-g1} and \eqref{eq-def-S-infty}, it is readily seen that if $\eta\in \Sigma^{R}\setminus \partial U(0;\rho_{1})$,
\begin{equation}\label{eq-vR-exp-small-I}
v_{R}(\eta)=I+\mathcal{O}(\lambda^{\frac{2k}{3}}e^{-M_{2}\lambda^{\frac{k}{2}+\frac{3}{2}\delta}}), \qquad \lambda \to +\infty.
\end{equation}
If $\eta\in\partial U(0;\rho_{1})$, making use of the full asymptotic expansion of the modified Bessel parametrix (see \cite[Equation (3.79)]{Vanlessen-2007}), we have 
\begin{equation}
v_{R}(\eta) \sim I+\sum\limits_{j=1}^{+\infty}\frac{v_{j}(\eta)}{\left(\lambda^{(4k+3)/2}b_{0}^{3/2}\right)^{j}}, \qquad \lambda \to+\infty,
\end{equation}
where $b_0$ is defined in \eqref{eq-def-b0} and
\begin{equation}\label{def:vj}
v_{j}(\eta)=\eta^{-\frac{1}{4}\sigma_{3}}J_{j}\eta^{\frac{1}{4}\sigma_{3}}f_{1}(\eta)^{-\frac{j}{2}}, \qquad j=1,2,\ldots.
\end{equation}
In \eqref{def:vj}, the function $f_{1}$ is defined in \eqref{eq-def-f1(eta)} and $J_{j}$ are constant matrices arising from the asymptotic expansion of $\Phi^{(\mathrm{Bes})}$. In particular, we have  
$J_{2j}$ and $J_{2j-1}$ are diagonal and anti-diagonal matrices, respectively, and
\begin{equation}
J_{1}=\frac{1}{8}\left(\begin{matrix}
0&-1\\3&0
\end{matrix}\right).
\end{equation}
Recall that $b_0 \in\left[ C_{1}\lambda^{-k-1+\delta}, M\right]$, it then follows that
for all fixed $\delta\in(-\frac{k}{3},k+1)$,  $v_{R}(\eta)=I+\mathcal{O}(\lambda^{-\frac{k}{2}-\frac{3}{2}\delta})$ uniformly for all $\eta\in \partial U(0;\rho_{1})$. 

By a standard argument of the small norm RH problem (cf. \cite{Deift1999}), we conclude that
\begin{equation}\label{eq-R-full-expand}
R(\eta)\sim I+\sum\limits_{j=1}^{\infty}\frac{R_{j}(\eta)}{\left(\lambda^{(4k+3)/2)}b_{0}^{3/2}\right)^{j}}, \qquad \lambda \to +\infty,
\end{equation}
where $R_{j}$ is analytic in $\mathbb{C} \setminus \partial U(0;\rho_1)$, behaves like $\mathcal{O}(1/\eta)$ as $\eta
\to \infty$, and satisfies the jump condition
\begin{equation}\label{eq-recurrence-Rj}
    \begin{aligned}
        R_{1,+}(\eta)&=R_{1,-}(\eta)+v_{1}(\eta),\\
        R_{j,+}(\eta)&=R_{j,-}(\eta)+v_{j}(\eta)+\sum\limits_{m=1}^{j-1}R_{m,-}(\eta)v_{j-m}(\eta), \qquad j>1,
    \end{aligned} 
\end{equation}
for $\eta\in\partial U(0;\rho_{1})$. It is easily seen that
\begin{equation}
    R_{1}(\eta)=\begin{cases}
\frac{\Res(v_1,0)}{\eta}-v_1(\eta), &\eta\in U(0;\rho_{1}),\\
\frac{\Res(v_1,0)}{\eta},&\eta \in\mathbb{C}\setminus U(0;\rho_{1}),
    \end{cases}
\end{equation}
where $\Res(v_1,0)$ stands for the residue of $v_1$ at the origin. From the definition of $v_1$ in
\eqref{def:vj}, it follows that
\begin{equation}\label{eq-R1-by-H(v)}
R_{1}(0)=\left(\begin{matrix}
0&\frac{b_{1}}{8}\\-\frac{3}{8}&0
\end{matrix}\right), \qquad
R_{1}'(0)=\left(\begin{matrix}
0&\star\\-\frac{3b_{1}}{8}&0
\end{matrix}\right). 
\end{equation} 
Since $J_{2j}$ and $J_{2j-1}$ are diagonal and anti-diagonal matrices, respectively, we see from \eqref{def:vj} that $v_{2j}$ and $v_{2j-1}$ bear the same structure. By induction, we then conclude from \eqref{eq-recurrence-Rj} that $R_{2j}(\eta)$ are all diagonal matrices and $R_{2j-1}(\eta)$ are all anti-diagonal matrices. These, together with \eqref{eq-R-full-expand}, \eqref{eq-R1-by-H(v)} and the fact that $b_0 \in\left[ C_{1}\lambda^{-k-1+\delta}, M\right]$, implies that, as $\lambda\to+\infty$,  
\begin{equation}
    \label{eq-R(0)}
    R(0)=\left(\begin{matrix}
        1+\mathcal{O}(\lambda^{-k-3\delta}) &\frac{b_{1}}{8}b_{0}^{-\frac{3}{2}}\lambda^{-\frac{4k+3}{2}}\left(1+\mathcal{O}(\lambda^{-k-3\delta}) \right)\\
        -\frac{3}{8}b_{0}^{-\frac{3}{2}}\lambda^{-\frac{4k+3}{2}}\left(1+\mathcal{O}(\lambda^{-k-3\delta}) \right) & 1+\mathcal{O}(\lambda^{-k-3\delta})
    \end{matrix}\right)
\end{equation}
and 
\begin{equation}
    \label{eq-R'(0)}
    R'(0)=\left(\begin{matrix}
        \mathcal{O}(\lambda^{-k-3\delta}) &\star\\
        -\frac{3b_{1}}{8}b_{0}^{-\frac{3}{2}}\lambda^{-\frac{4k+3}{2}}\left(1+\mathcal{O}(\lambda^{-k-3\delta}) \right) & \mathcal{O}(\lambda^{-k-3\delta})
    \end{matrix}\right).
\end{equation}

\subsection{Proofs of Lemmas \ref{lem-asym-dF/dmu} and \ref{lem-dF/dtau-case-I}}
\label{sec:case-I-proof}

\paragraph{Proof of Lemma \ref{lem-asym-dF/dmu}}
Note that $R(\eta)=I+\frac{T_{-1}}{\eta}+\mathcal{O}(\eta^{-2})$ as $\eta\to\infty$, we have
\begin{equation}
 T_{-1}=\lim\limits_{\eta\to \infty}\eta\left(R(\eta)-I\right).
\end{equation}
According to \eqref{eq-R-full-expand}, \eqref{eq-recurrence-Rj} and the fact that $R_{1}(\eta)$ is anti-diagonal and $R_{2}(\eta)$ is diagonal, we further obtain 
\begin{align}
(T_{-1})_{12}&=b_{0}^{-\frac{3}{2}}\lambda^{-\frac{4k+3}{2}}\lim\limits_{\eta\to\infty}\eta(R_{1}(\eta))_{12}\left(1+\mathcal{O}(\lambda^{-k-3\delta})\right) \nonumber \\
&=-\frac{1}{8}b_{0}^{-\frac{3}{2}}\lambda^{-\frac{4k+3}{2}}\Res\left(\eta^{-\frac{1}{2}}f_{1}(\eta)^{-\frac{1}{2}},0\right)\left(1+\mathcal{O}(\lambda^{-k-3\delta})\right)
\nonumber 
\\
&= -\frac{1}{8}b_{0}^{-\frac{3}{2}}\lambda^{-\frac{4k+3}{2}}\left(1+\mathcal{O}(\lambda^{-k-3\delta})\right).
\end{align} 
It then follows from \eqref{eq-def-b0}, \eqref{eq-relation-T-X} and the above estimate that 
\begin{multline}\label{eq-asym-(X-1)12}
(X_{-1})_{12}=-d_{1}\lambda^{2k+2} + (\lambda b_{0})^{\frac{1}{2}}(T_{-1})_{12}
\\
={-d_{1}\lambda^{2k+2}}+\frac{1}{8}\frac{1}{\left(\alpha_{k}s^{2k+1}+x\right)}\left(1+\mathcal{O}\left(\lambda^{-k-3\delta}\right)\right), \qquad \lambda\to+\infty.
\end{multline}
The error bound holds uniformly for all $(y,r)$ satisfying $\alpha_{k}r^{2k+1}+y\in[-M,-C_{1}\lambda^{-k-1+\delta}]$ with $\delta\in\left(-\frac{k}{3},k+1\right)$. Recalling the differential identity in \eqref{eq-diff-identity-1}, we have from \eqref{eq-def-Psi--1} and \eqref{eq-relation-X-Y-Psi} that
\begin{equation}
    \frac{\partial F}{\partial x}(s;x) = -(Y_{-1})_{12} = (\Psi_{-1})_{12} - (X_{-1})_{12} = h(x) - (X_{-1})_{12}.
\end{equation}
Taking $\lambda=|s|$ ({\it i.e.} $r=-1$ in \eqref{s-r-and-x-y}), $\delta=-\frac{k}{3}+\epsilon$, and substituting \eqref{eq-asym-(X-1)12} into the above formula, we obtain \eqref{eq-asym-dF/dmu} with the definition of $d_1$ in \eqref{eq-def-d1}. Recall that $x=y\lambda^{2k+1}=y|s|^{2k+1}$ (see \eqref{s-r-and-x-y}), we find that the asymptotic formula \eqref{eq-asym-dF/dmu} holds uniformly for $x\in[-c_{1}|s|^{2k+1},\alpha_{k}|s|^{2k+1}-c_{2}|s|^{\frac{2k}{3}+\epsilon}]$ with $c_{1},c_{2}>0$ and $\epsilon\in(0,\frac{4k}{3}+1)$ being arbitrary fixed. This finishes the proof of Lemma \ref{lem-asym-dF/dmu}.

\paragraph{Proof of Lemma \ref{lem-dF/dtau-case-I}} 
By inverting the transformations \eqref{eq-transform-tilde-X-T} and \eqref{eq-def-R-by-S}, it follows from \eqref{eq-differential-indentity-2} that 
\begin{align}\label{eq-dF/dtau-three-parts}
\frac{\partial F}{\partial s}(s;x)& =\left.\frac{b_{0}^2\lambda^{4k+2}}{2\pi i}\left(\Phi^{(\mathrm{Bes})}(z)^{-1}\frac{d\Phi^{(\mathrm{Bes})}(z)}{dz}\right)_{21}\right|_{z\to 0} \nonumber \\
&~~~+\left.\frac{1}{2\pi i\lambda b_{0}}\left(\Phi^{(\mathrm{Bes})}(z)^{-1}E_{1}(\eta)^{-1}E_{1}'(\eta)\Phi^{(\mathrm{Bes})}(z)\right)_{21}\right|_{\eta\to 0} \nonumber \\
&~~~+\left.\frac{1}{2\pi i\lambda b_{0}}\left(\Phi^{(\mathrm{Bes})}(z)^{-1}E_{1}(\eta)^{-1} R(\eta)^{-1} R'(\eta)E_{1}(\eta)\Phi^{(\mathrm{Bes})}(z)\right)_{21}\right|_{\eta\to 0},
\end{align}
where $z=z(\eta)=\lambda^{4k+3}b_{0}^{3}f_{1}(\eta)$ with $f_1$ given in \eqref{eq-def-f1(eta)} and $E_1$ is defined in \eqref{eq-def-E1}. In view of the local behavior of  $\Phi^{(\mathrm{Bes})}$ near the origin given in \eqref{eq:asym-Bessel-parametrix-near-0}, one can show that
\begin{equation}\label{eq-dF/dtau-part-1}
\left.\frac{b_{0}^{2}\lambda^{4k+2}}{2\pi i}\left(\Phi^{(\mathrm{Bes})}(z)^{-1}\frac{d\Phi^{(\mathrm{Bes})}(z)}{dz}\right)_{21}\right|_{z\to 0}=\frac{1}{4}b_{0}^2\lambda^{4k+2}
\end{equation} 
and
\begin{equation}\label{eq-dF/dtau-part-2}
\left.\frac{1}{2\pi i\lambda b_{0}}\left(\Phi^{(\mathrm{Bes})}(z)^{-1}E_{1}(\eta)^{-1}E_{1}'(\eta)\Phi^{(\mathrm{Bes})}(z)\right)_{21}\right|_{\eta\to 0}=0.
\end{equation}
To estimate the third term on the right hand side of \eqref{eq-dF/dtau-three-parts}, we first observe from \eqref{eq-R(0)} and \eqref{eq-R'(0)} that
\begin{equation}
    \label{eq-approx-tilde-R(r)}
   R(0)^{-1}R'(0)=\left(\begin{matrix}
      \star &\star\\
      -\frac{3}{8}b_{1}b_{0}^{-\frac{3}{2}}\lambda^{-\frac{4k+3}{2}}\left(1+\mathcal{O}(\lambda^{-k-3\delta})\right) &\star
    \end{matrix}\right),\qquad \lambda \to +\infty.
\end{equation}
Moreover, from the definition of $E_{1}(\eta)$ in \eqref{eq-def-E1}, one can see that $E_{1}(0)=\lambda^{\frac{4k+3}{4}\sigma_{3}}b_{0}^{\frac{3}{4}\sigma_{3}}$. These, together with \eqref{eq:asym-Bessel-parametrix-near-0}, imply that 
\begin{align}\label{eq-dF/dtau-part-3}
&\left.\frac{1}{2\pi i\lambda b_{0}}\left(\Phi^{(\mathrm{Bes})}(z)^{-1}E_{1}(\eta)^{-1}R(\eta)^{-1}R'(\eta)E_{1}(\eta)\Phi^{(\mathrm{Bes})}(z)\right)_{21}\right|_{\eta\to 0}
\nonumber 
\\
&=\left.\frac{1}{2\pi i\lambda b_{0}}\left(\Phi^{(\mathrm{Bes})}(z)^{-1}\left(\begin{matrix}
\star&\star\\-\frac{3b_{1}}{8}\left(1+\mathcal{O}\left(\lambda^{-k-3\delta}\right)\right)&\star
\end{matrix}\right)\Phi^{(\mathrm{Bes})}(z)\right)_{21}\right|_{\eta\to 0}
\nonumber 
\\
&=\frac{3b_{1}}{16b_{0}\lambda}\left(1+\mathcal{O}\left(\lambda^{-k-3\delta}\right)\right)
,\qquad \lambda \to +\infty.
\end{align}
Finally, recall that $x=y\lambda^{2k+1}$, we fix $y\neq 0$ and take $\delta=\frac{1}{6}$, then a combination of \eqref{eq-dF/dtau-three-parts}--\eqref{eq-dF/dtau-part-3} and the definitions of $b_{0}, b_{1}$ in \eqref{eq-def-b0} and  \eqref{eq-def-b0-bi} yields \eqref{eq-dF/dtau-case-I} as $|x|\to+\infty$, uniformly for $s$ such that $\alpha_{k}s^{2k+1}+x\in\left[-M|x|,-C|x|^{\frac{k+1/6}{2k+1}}\right]$.

\section{Asymptotic analysis in the transition region} \label{Sec:Tran}

In this section, we study the large $\lambda$ behavior of $X(\lambda\eta)$  for $\alpha_{k}r^{2k+1} + y \in [-C_{1}\lambda^{-k-1+\delta}, C_{1}\lambda^{-k-1+\delta}]$ with $\delta \in (0, \frac{1}{4})$. As discussed at the end of Section \ref{sec:diffident}, this corresponds to the transition region. Our goal is to establish the validity of Lemma \ref{lem-dF/dtau-case-II}. It is crucial to emphasize that $b_{0} = -\alpha_{k}r^{2k+1} - y$ is small and may vanish in this context. Therefore, the function $f_{1}(\eta)$ defined in \eqref{eq-def-f1(eta)} is no longer a valid conformal mapping in a small neighborhood of $\eta = 0$. Consequently, the local Bessel parametrix used in Section \ref{Sec:Alg} becomes inappropriate and will be replaced by the Painlev\'e XXXIV ($\mathrm{P_{34}}$) parametrix $\Phi(\zeta;x)$ or its variant $\tilde{\Phi}(\zeta;x)$ with $x$ being a parameter; see Appendix \ref{sec-appendix-p34-parametrix} for the definitions. 

As shown in Appendix \ref{sec:uniform-asym-scaled-P34},  the uniform asymptotic behavior of $\Phi(\zeta;x)$ differs depending on whether $\zeta \to \infty$ and $x$ approaches negative infinity, as compared to the case when $\zeta \to \infty$ and $x$ approaches positive infinity. These qualitatively different asymptotic behaviors lead us to split the analysis into two cases, that is,
$\alpha_{k}r^{2k+1}+y \in [-C_{1}\lambda^{-k-1+\delta}, 0]$ and 
$\alpha_{k}r^{2k+1}+y \in [0,C_{1}\lambda^{-k-1+\delta}]$.

\subsection{Riemann-Hilbert analysis when $\alpha_{k}r^{2k+1}+y\in[-C_{1}\lambda^{-k-1+\delta}, 0]$}
\label{sec:RH-analysis-case-II-}

\paragraph{First transformation} 
As previously noted in Remark \ref{rem-g1-wide-applied}, when $\alpha_{k}r^{2k+1}+y\in[-C_{1}\lambda^{-k-1+\delta}, 0]$, the function $g_{1}$ in \eqref{eq-def-g1-asym} remains suitable for normalizing the RH problem for $X$. We thus define
\begin{equation}\label{eq-def-P(eta)-case-II}
P(\eta)=\left(\begin{matrix}
1&0\\-d_{1}\lambda^{2k+2}&1
\end{matrix}\right)X(\lambda\eta)\exp\{\lambda^{\frac{4k+3}{2}}g_{1}(\eta)\sigma_{3}\},
\end{equation}
where $d_1$ is defined in \eqref{eq-def-d1}. It is then straightforward to check that $P$ satisfies the following RH problem. 

\subsubsection*{RH problem for $P$}
\begin{enumerate}
\item [(a)] $P(\eta)$ is analytic in $\mathbb{C}\setminus\Sigma^{P}$, where $\Sigma^{P}=\Sigma_{1}^{P}\cup\Sigma_{2}^{P}\cup\Sigma_{3}^{P}$ is illustrated in Figure \ref{fig-P}.
\item [(b)] On $\Sigma^{P}$, the limiting values $P_{\pm}(\eta)$ exist and satisfy the jump condition
\begin{equation}
P_{+}(\eta)=P_{-}(\eta)\begin{cases}\left(\begin{matrix}
1&0\\ e^{2\lambda^{\frac{4k+3}{2}}g_{1}(\eta)}&1
\end{matrix}\right),\quad &\eta\in\Sigma_{1}^{P}\cup\Sigma_{3}^{P},\\
\left(\begin{matrix}
0&1\\ -1&0
\end{matrix}\right),\quad &\eta\in\Sigma_{2}^{P}.
\end{cases}
\end{equation}
\item [(c)] As $\eta\to r$, we have $P(\eta)=\mathcal{O}(\log{(\eta-r)})$.
\item [(d)] As $\eta\to\infty$, we have
\begin{equation}
P(\eta)=\left(I+\mathcal{O}(\eta^{-1})\right)(\lambda\eta)^{-\frac{1}{4}\sigma_{3}}N,
\end{equation}
where $N$ is defined in \eqref{eq-def-N}.

\end{enumerate}

\begin{figure}
    \centering
    
\begin{tikzpicture}[x=0.75pt,y=0.75pt,yscale=-1,xscale=1]

\draw    (271.5,32) -- (371.5,132) ;
\draw [shift={(325.04,85.54)}, rotate = 225] [fill={rgb, 255:red, 0; green, 0; blue, 0 }  ][line width=0.08]  [draw opacity=0] (8.93,-4.29) -- (0,0) -- (8.93,4.29) -- cycle    ;
\draw    (272.5,220) -- (371.5,132) ;
\draw [shift={(325.74,172.68)}, rotate = 138.37] [fill={rgb, 255:red, 0; green, 0; blue, 0 }  ][line width=0.08]  [draw opacity=0] (8.93,-4.29) -- (0,0) -- (8.93,4.29) -- cycle    ;
\draw  [fill={rgb, 255:red, 0; green, 0; blue, 0 }  ,fill opacity=1 ] (371.67,131.83) .. controls (371.67,131.19) and (371.15,130.66) .. (370.5,130.66) .. controls (369.85,130.66) and (369.33,131.19) .. (369.33,131.83) .. controls (369.33,132.48) and (369.85,133) .. (370.5,133) .. controls (371.15,133) and (371.67,132.48) .. (371.67,131.83) -- cycle ;
\draw    (226.33,130.83) -- (369.33,131.83) ;
\draw [shift={(302.83,131.37)}, rotate = 180.4] [fill={rgb, 255:red, 0; green, 0; blue, 0 }  ][line width=0.08]  [draw opacity=0] (8.93,-4.29) -- (0,0) -- (8.93,4.29) -- cycle    ;

\draw (373.5,135) node [anchor=north west][inner sep=0.75pt]   [align=left] {$r$};
\draw (244,19) node [anchor=north west][inner sep=0.75pt]   [align=left] {$\displaystyle {\textstyle \Sigma_{1}^{P}}$};
\draw (199,125) node [anchor=north west][inner sep=0.75pt]   [align=left] {$\displaystyle {\textstyle \Sigma_{2}^{P}}$};
\draw (248,218) node [anchor=north west][inner sep=0.75pt]   [align=left] {$\displaystyle {\textstyle \Sigma_{3}^{P}}$};

\end{tikzpicture}
    \caption{The jump contour $\Sigma^P$ of RH problem for $P$.}
    \label{fig-P}
\end{figure}
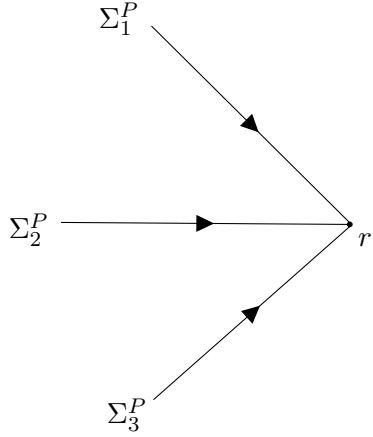

According to Lemma \ref{lem-g1(eta)}, there exist two positive constants $M_{2}$ and $C_{2}$ such that for $\eta\in \Sigma_{1}^{P}\cup\Sigma_{3}^{P}$,
\begin{equation}
\label{eq-real-g1(eta)-case-II}
\lambda^{\frac{4k+3}{2}}\re g_{1}(\eta)\leq -M_{2}\lambda^{\frac{4k+3}{2}}|\eta-r|^{\frac{3}{2}}\leq -C_{2}\lambda^{\frac{k}{2}},
\end{equation}
uniformly for $|\eta-r|\geq \lambda^{-k-1}$ and $\lambda$ large enough. As a consequence, the jump matrix of $P$ tends to the identity matrix exponentially fast on $\Sigma_{1}^{P}\cup\Sigma_{3}^{P}$ as $\lambda \to +\infty$. 

\paragraph{Global parametrix} 
By ignoring the jump matrix on $\Sigma_{1}^{P}\cup\Sigma_{3}^{P}$, 
the global parametrix is actually a shift of that in Section \ref{sec:RH-analysis-case-I}.

\subsubsection*{RH problem for $P^{(\infty)}$}
\begin{enumerate}
    \item [(a)] $P^{(\infty)}(\eta)$ is analytic in $\mathbb{C}\setminus (-\infty,r]$.
    \item [(b)] $P^{(\infty)}$ satisfies the jump condition
    \begin{equation}
        P^{(\infty)}_{+}(\eta)=P^{(\infty)}_{-}(\eta)\left(\begin{matrix}
            0&1\\ -1&0
        \end{matrix}\right), \qquad \eta\in(-\infty,r).
    \end{equation}
    \item [(c)] As $\eta\to\infty$, we have 
    \begin{equation}
        P^{(\infty)}(\eta)=(I+\mathcal{O}(\eta^{-1}))(\lambda\eta)^{-\frac{1}{4}}N.
    \end{equation}
\end{enumerate}
A solution of the above RH problem is given by
\begin{equation}\label{eq-def-Pinfty(eta)}
P^{(\infty)}(\eta)=(\lambda(\eta-r))^{-\frac{1}{4}\sigma_{3}}N.
\end{equation}

\paragraph{Local parametrix} 
The non-uniform convergence of the jump matrix of $P$ on $\Sigma^{P}_{1}\cup \Sigma_{3}^{P}$ to the identity matrix suggests the following local parametrix in a small neighborhood $ U(r;\rho_2)$ of $\eta=r$. Particularly, we later will take $\rho_2=\rho_2(\lambda)$ in a precise way such that $\rho_2(\lambda)\to 0$ as $\lambda \to +\infty$, so the disc is shrinking as $\lambda$ increases. 

\paragraph{RH problem for $P^{(r)}$}
\begin{enumerate}
\item [(a)] $P^{(r)}(\eta)$ is analytic in $\Sigma^{P} \setminus U(r;\rho_2)$.
\item [(b)] On $\Sigma^{P} \cap U(r;\rho_2) $, the limiting values $P^{(r)}_{\pm}(\eta)$ exist and satisfy the jump condition
\begin{equation}\label{eq-jump-P0}
P^{(r)}_{+}(\eta)=P^{(r)}_{-}(\eta)\begin{cases}\left(\begin{matrix}
1&0\\ e^{2\lambda^{\frac{4k+3}{2}}g_{1}(\eta)}&1
\end{matrix}\right),\quad &\eta\in U(r;\rho_2)\cap(\Sigma_{1}^{P}\cup\Sigma_{3}^{P}),\\
\left(\begin{matrix}
0&1\\ -1&0
\end{matrix}\right),\quad &\eta\in U(r;\rho_2)\cap \Sigma_{2}^{P}.
\end{cases}
\end{equation}
\item [(c)] As $\eta\to r$, we have $P^{(r)}(\eta)=\mathcal{O}(\log{(\eta-r)})$.
\item [(d)] As $\lambda \to\infty$, $P^{(r)}$ satisfies the matching condition 
\begin{equation}\label{eq-matching-P0-Pinfty}
    P^{(r)}(\eta)= (I+o(1))P^{(\infty)}(\eta), \qquad \eta \in \partial U(r;\rho_2). 
\end{equation}

\end{enumerate}

To solve the above RH problem, we first introduce a conformal mapping in the disk $U(r;\rho_{2})$ defined by
\begin{equation}
\label{eq-def-f2}
f_{2}(\eta)=\left[\frac{3}{2}\left(\sum\limits_{i=1}^{2k+1}(-1)^{i+1}b_{i}(\eta-r)^{i-1}\right)\right]^{\frac{2}{3}}(\eta-r),
\end{equation}
where $b_i$, $i=1,\ldots,2k+1$, are given in \eqref{eq-def-b0-bi}. Note that $f_2(\eta)$ is independent of $b_0$, it thus serves as a conformal mapping uniformly for $b_0 = -\alpha_{k}r^{2k+1}-y\in[0, C_{1}\lambda^{-k-1+\delta}]$. Then, we define 
\begin{equation}\label{eq-def-P0(eta)}
P^{(r)}(\eta)=
E_{2}(\eta)\Phi(\lambda^{\frac{4k+3}{3}}f_{2}(\eta);\lambda^{\frac{4k+3}{3}}\kappa_{0})e^{\lambda^{\frac{4k+3}{2}}g_{1}(\eta)\sigma_{3}},\quad \eta\in U(r; \rho_{2}),
\end{equation}
where $\kappa_0$ is a constant independent of $\eta$ to be determined later, 
\begin{equation}
\label{eq-def-E2(eta)}
E_{2}(\eta):=(\lambda(\eta-r))^{-\frac{1}{4}\sigma_{3}}(\lambda^{\frac{4k+3}{3}}f_{2}(\eta))^{\frac{1}{4}\sigma_{3}}e^{-\frac{\pi i}{4}\sigma_{3}}
\end{equation}
is  analytic  in $U(r;\rho_{2})$ and $\Phi(\zeta;x)$ is the $\mathrm{P_{34}}$ parametrix given in Appendix \ref{sec-appendix-p34-parametrix}. 

In \eqref{eq-def-P0(eta)}, the function $\Phi$ depends on $b_0$ through the constant $\kappa_0$.
To determine it, we set 
\begin{equation}
\label{eq-def-kappa(eta)}
\kappa(\eta)=-b_{0}(\eta-r)^{\frac{1}{2}}f_{2}(\eta)^{-\frac{1}{2}},
\end{equation}
so that
\begin{equation}\label{eq-relation-g1-f2}
g_{1}(\eta)=\frac{2}{3}\left(f_{2}(\eta)\right)^{\frac{3}{2}}+\kappa(\eta)\left(f_{2}(\eta)\right)^{\frac{1}{2}}.
\end{equation}
The constant $\kappa_0$ is chosen to be the approximation of $\kappa(r)$ for large $\lambda$. More precisely, note that if $\alpha_{k}r^{2k+1}+y\in [-C_{1}\lambda^{-k-1+\delta}, 0]$ and $|r|+|y|\geq \delta_{0}>0$ (see \eqref{eq-r-y-not-both-vanish}), we have
\begin{equation}\label{eq-asym-r}
  r=-\left(\frac{y}{\alpha_{k}}\right)^{\frac{1}{2k+1}}\left(1+\mathcal{O}(\lambda^{-k-1+\delta})\right), \qquad   \lambda\to+\infty.
\end{equation}
This, together with the definition of $b_{1}$ in \eqref{eq-def-b0-bi}, implies that 
\begin{equation}
\label{eq-asym-kappa(0)}
\kappa(r)=-b_{0}\left[\frac{3}{2}b_{1}\right]^{-\frac{1}{3}}=\kappa_{0}\left(1+\mathcal{O}(\lambda^{-k-1+\delta})\right), \qquad   \lambda\to+\infty,
\end{equation}
where 
\begin{equation}\label{eq-def-kappa0}
\kappa_{0}:=\kappa_{0}(r;y)=-b_{0}\left[(2k+1)\alpha_{k}\left(\frac{y}{\alpha_{k}}\right)^{\frac{2k}{2k+1}}\right]^{-\frac{1}{3}}\leq 0.
\end{equation}
For later use, we derive the approximation of $f_2'(r)$ for large $\lambda$. From the definition of $b_{0}$ in \eqref{eq-def-b0}, we see that $\frac{\partial b_{0}}{\partial r}=-(2k+1)\alpha_{k}r^{2k}=-\frac{3}{2}b_{1}$. This fact, together with \eqref{eq-asym-r} and \eqref{eq-def-kappa0}, yields
\begin{align}
\label{eq-f2'0}
f_{2}'(r)&=\left(\frac{3}{2}b_{1}\right)^{\frac{2}{3}}=\left(\frac{3}{2}b_{1}\right)\left(\frac{3}{2}b_{1}\right)^{-\frac{1}{3}} \nonumber \\
&=\frac{(2k+1)\alpha_{k}r^{2k}}{\left[(2k+1)\alpha_{k}\left(\frac{y}{\alpha_{k}}\right)^{\frac{2k}{2k+1}}\right]^{\frac{1}{3}}}\left(1+\mathcal{O}(\lambda^{-k-1+\delta})\right)
= \frac{\partial \kappa_{0}}{\partial r}\left(1+\mathcal{O}(\lambda^{-k-1+\delta})\right),
\end{align}  
as $\lambda\to+\infty$ uniformly for $\alpha_{k}r^{2k+1}+y\in[-C_1\lambda^{-k-1+\delta},0]$. 

With the aid of $\kappa_0$ in \eqref{eq-def-kappa0}, we now choose the radius $\rho_2$ of the disk to be 
\begin{equation}\label{eq-def-rho2}
    \rho_{2}=
    \frac{1}{\lambda^{\frac{4k+3}{2}}|\kappa_{0}|+\lambda^{\frac{2k}{3}+\frac{1}{2}}},
    \end{equation}
and prove the following lemma. 
\begin{lemma} \label{lem-para-pr}
The function $P^{(r)}$ defined in \eqref{eq-def-P0(eta)} solves the RH problem for $P^{(r)}$.
\end{lemma}

\begin{proof}
Recalling that $E_{2}(\eta)$ is an analytic prefactor in $U(r;\rho_{2})$, and using the jump condition of $\Phi$ in \eqref{eq-jump-Phi}, it is straightforward to verify that $P^{(r)}$ satisfies the jump condition in \eqref{eq-jump-P0}. The next task is to verify the matching condition \eqref{eq-matching-P0-Pinfty}. From the large $\zeta$ behavior of $\Phi$ in \eqref{eq-Phi-asym-infty}, it is clear that the exponential factor cannot be completely canceled by the choice of $\kappa_0$ in \eqref{eq-def-kappa0}; see the relation between $g_1$ and $f_2$ in \eqref{eq-relation-g1-f2}. To ensure that the remaining exponential quantity is controllable, we need to choose a shrinking radius $\rho_2$ in \eqref{eq-def-rho2}.

When $\alpha_{k}r^{2k+1}+y=-b_{0}\in [-C_{1}\lambda^{-k-1+\delta}, 0]$ with $\delta\in(0,\frac{1}{4})$, it is readily seen from \eqref{eq-def-kappa0} that 
\begin{equation}
    \kappa_{0}=\mathcal{O}(\lambda^{-k-1+\delta}), \qquad \lambda\to+\infty.
\end{equation}
This implies that $\rho_{2}\gg \lambda^{-k-1}$ when $\lambda$ is large enough. Due to \eqref{eq-real-g1(eta)-case-II}, one can see that the jump of $P$ on $(\Sigma_{1}^{P}\cup\Sigma_{3}^{P})\setminus U(r;\rho_{2})$ indeed tends to the identity matrix exponentially fast.

To verify the matching condition \eqref{eq-matching-P0-Pinfty}, we first obtain from \eqref{eq-def-Pinfty(eta)}, \eqref{eq-def-P0(eta)} and \eqref{eq-def-E2(eta)} that
\begin{multline}\label{eq-P0-Pinfty-1-explicit-representation}
 P^{(r)}(\eta)P^{(\infty)}(\eta)^{-1}\\
=\left(\frac{\lambda^{\frac{4k+3}{3}}f_{2}(\eta)}{\lambda(\eta-r) e^{\pi i}}\right)^{\frac{1}{4}\sigma_{3}}\Phi(\lambda^{\frac{4k+3}{3}}f_{2}(\eta);\lambda^{\frac{4k+3}{3}}\kappa_{0})e^{\lambda^{\frac{4k+3}{2}}g_{1}(\eta)\sigma_{3}}N^{-1}(\lambda(\eta-r))^{\frac{1}{4}\sigma_{3}}.
\end{multline}
As $\zeta \to \infty$, $\Phi(\zeta;x)$ exhibits different asymptotic behaviors for  bounded $x$  or $x\to-\infty$; see \eqref{eq-Phi-asym-infty} and \eqref{eq-asym-PII-large-xi-negative}. We thus split the discussions into two cases in what follows, namely, $\lambda^{\frac{4k+3}{3}}\kappa_{0}$ is bounded and  $\lambda^{\frac{4k+3}{3}}\kappa_{0}\to-\infty$.


When $\lambda^{\frac{4k+3}{3}}\kappa_{0}$ is bounded, {\it i.e.} $b_{0}=\mathcal{O}(\lambda^{-\frac{4k+3}{3}})$ as $\lambda\to+\infty$, according to \eqref{eq-def-rho2}, we find that $\rho_{2}$ behaves like $C(\lambda)\lambda^{-\frac{2k}{3}-\frac{1}{2}}$ as $\lambda\to+\infty$, where $C(\lambda)$ is bounded and non-zero when $\lambda$ is large enough.  A combination of \eqref{eq-def-kappa(eta)}, \eqref{eq-asym-kappa(0)} and \eqref{eq-def-kappa0} gives us 
\begin{equation}\label{eq-approx-kappa-eta-kappa0}
\kappa(\eta)-\kappa_{0}=\mathcal{O}(\lambda^{-2k-\frac{3}{2}}), \qquad \lambda\to+\infty,
\end{equation} 
uniformly for $\eta\in\partial U(r;\rho_{2})$. From the above formula and \eqref{eq-relation-g1-f2}, we have
\begin{align}\label{eq-exp-error-I}
&\exp\left\{\lambda^{\frac{4k+3}{2}}g_{1}(\eta)-\frac{2}{3}\left(\lambda^{\frac{4k+3}{3}}f_{2}(\eta)\right)^{\frac{3}{2}}-\lambda^{\frac{4k+3}{3}}\kappa_{0}\left(\lambda^{\frac{4k+3}{3}}f_{2}(\eta)\right)^{\frac{1}{2}}\right\} 
\nonumber 
\\
&=\exp\left\{\lambda^{\frac{4k+3}{2}}(\kappa(\eta)-\kappa_{0})\left(f_{2}(\eta)\right)^{\frac{1}{2}}\right\} = 1+\mathcal{O}(\rho_{2}^{\frac{1}{2}}), \qquad \lambda\to+\infty,
\end{align}
uniformly for all $\eta\in\partial U(r;\rho_{2})$. This implies that 
\begin{equation}\label{eq-approx-NeN-1}
   N\exp\left\{\left\{\lambda^{\frac{4k+3}{2}}(\kappa(\eta)-\kappa_{0})\left(f_{2}(\eta)\right)^{\frac{1}{2}}\right\}\sigma_{3}\right\}N^{-1}=I+ \left(\begin{matrix}
    \mathcal{O}(\rho_{2}) & \mathcal{O}(\rho_{2}^{\frac{1}{2}})\\
    \mathcal{O}(\rho_{2}^{\frac{1}{2}}) & \mathcal{O}(\rho_{2})
\end{matrix}\right)
\end{equation}
as $\lambda\to+\infty$ uniformly for all  $\eta\in\partial U(r;\rho_{2})$ and $\lambda^{\frac{4k+3}{3}}|\kappa_{0}|=\mathcal{O}(1)$.
Therefore, in view of \eqref{eq-Phi-asym-infty}, we obtain from \eqref{eq-P0-Pinfty-1-explicit-representation} and the above approximation that
\begin{align}
&P^{(r)}(\eta)P^{(\infty)}(\eta)^{-1} \nonumber \\
&=\left(\lambda(\eta-r)\right)^{-\frac{\sigma_{3}}{4}}\left(I+\left(\begin{matrix}
    \mathcal{O}(\lambda^{-\frac{4k+3}{3}}f_{2}(\eta)^{-1}) & \mathcal{O}(\lambda^{-\frac{4k+3}{6}} f_{2}(\eta)^{-\frac{1}{2}}) \\
    \mathcal{O}(\lambda^{-\frac{4k+3}{2}} f_{2}(\eta)^{-\frac{3}{2}}) &\mathcal{O}(\lambda^{-\frac{4k+3}{3}}f_{2}(\eta)^{-1})
\end{matrix}\right)
\right)\left(\lambda(\eta-r)\right)^{\frac{\sigma_{3}}{4}}\nonumber \\
&\quad \times(\lambda(\eta-r))^{-\frac{1}{4}\sigma_{3}}  N\exp\left\{\left\{\lambda^{\frac{4k+3}{2}}(\kappa(\eta)-\kappa_{0})\left(f_{2}(\eta)\right)^{\frac{1}{2}}\right\}\sigma_{3}\right\}N^{-1}(\lambda(\eta-r))^{\frac{1}{4}\sigma_{3}} \nonumber \\
&=(\lambda\rho_{2})^{-\frac{1}{4}\sigma_{3}}\left(I+\left(\begin{matrix}
    \mathcal{O}(\rho_{2}) & \mathcal{O}(\rho_{2}^{\frac{1}{2}})\\
    \mathcal{O}(\rho_{2}^{\frac{1}{2}}) & \mathcal{O}(\rho_{2})
\end{matrix}\right)\right)(\lambda\rho_{2})^{\frac{1}{4}\sigma_{3}},
 \qquad \lambda\to+\infty,   \label{eq-asym-Pr-Pinf-bounded}
\end{align}
uniformly for all $\eta\in\partial U(r; \rho_{2})$ and $\lambda^{\frac{4k+3}{3}}|\kappa_{0}|=\mathcal{O}(1)$.

When $-\lambda^{\frac{4k+3}{3}}\kappa_{0}\gg 1$ as $\lambda\to+\infty$, 
it is obvious that $\lambda^{\frac{4k+3}{2}}|\kappa_{0}| \gg \lambda^{\frac{2k}{3}+\frac{1}{2}}$. Then, it follows from  \eqref{eq-def-rho2} that 
\begin{equation}\label{eq-asym-kappa0rho2}
    |\kappa_{0}|\rho_{2}\sim \lambda^{-\frac{4k+3}{2}}, \qquad \lambda\to+\infty.
\end{equation}
Thus, \eqref{eq-exp-error-I} and \eqref{eq-approx-NeN-1} still hold uniformly for all $\eta\in\partial U(r;\rho_{2})$. Using \eqref{eq-asym-PII-large-xi-negative}, a similar computation as in \eqref{eq-asym-Pr-Pinf-bounded} gives us
\begin{align}
&P^{(r)}(\eta)P^{(\infty)}(\eta)^{-1} \nonumber\\
& =\left(\lambda(\eta-r)\right)^{-\frac{\sigma_{3}}{4}}\left(I+\left(\begin{matrix}
    \mathcal{O}(\lambda^{-(4k+3)}\kappa_{0}^{-2}f_{2}(\eta)^{-1}) &\mathcal{O}(\lambda^{-\frac{4k+3}{2}}\kappa_{0}^{-1}f_{2}(\eta)^{-\frac{1}{2}})\\
    \mathcal{O}(\lambda^{-\frac{3(4k+3)}{2}}\kappa_{0}^{-3}f_{2}(\eta)^{-\frac{3}{2}}) &\mathcal{O}(\lambda^{-(4k+3)}\kappa_{0}^{-2}f_{2}(\eta)^{-1})
\end{matrix}\right)
\right)\left(\lambda(\eta-r)\right)^{\frac{\sigma_{3}}{4}}  \nonumber\\
&\qquad \times (\lambda(\eta-r))^{-\frac{1}{4}\sigma_{3}}N\exp\left\{\left\{\lambda^{\frac{4k+3}{2}}(\kappa(\eta)-\kappa_{0})\left(f_{2}(\eta)\right)^{\frac{1}{2}}\right\}\sigma_{3}\right\}N^{-1}(\lambda(\eta-r))^{\frac{1}{4}\sigma_{3}}  \nonumber\\
&=(\lambda\rho_{2})^{-\frac{1}{4}\sigma_{3}}\left(I+\left(\begin{matrix}
    \mathcal{O}(\rho_{2}) & \mathcal{O}(\rho_{2}^{\frac{1}{2}})\\
    \mathcal{O}(\rho_{2}^{\frac{1}{2}}) & \mathcal{O}(\rho_{2})
\end{matrix}\right)\right)(\lambda\rho_{2})^{\frac{1}{4}\sigma_{3}}, \qquad \lambda\to+\infty, \label{eq-asym-Pr-Pinf-unbounded}
\end{align}
uniformly for all $\eta\in\partial U(r; \rho_{2})$ and $-\lambda^{\frac{4k+3}{3}}\kappa_{0}\gg 1$. 

Combining \eqref{eq-asym-Pr-Pinf-bounded} and \eqref{eq-asym-Pr-Pinf-unbounded}, we conclude that 
\begin{equation}\label{eq-asym-Pr-Pinf}
 P^{(r)}(\eta)P^{(\infty)}(\eta)^{-1}=
 I+(\lambda\rho_{2})^{-\frac{1}{4}\sigma_{3}}\left(\begin{matrix}
    \mathcal{O}(\rho_{2}) & \mathcal{O}(\rho_{2}^{\frac{1}{2}})\\
    \mathcal{O}(\rho_{2}^{\frac{1}{2}}) & \mathcal{O}(\rho_{2})
\end{matrix}\right)(\lambda\rho_{2})^{\frac{1}{4}\sigma_{3}}
\end{equation}
as $\lambda\to+\infty$ uniformly for all $\eta\in\partial U(r;\rho_{2})$ and $\alpha_{k}r^{2k+1}+y\in[-C_{1}\lambda^{-k-1+\delta},0]$. This completes the proof of the lemma.
    \end{proof}

\paragraph{Final transformation}
With the aid of the global and local parametrices built in \eqref{eq-def-Pinfty(eta)} and \eqref{eq-def-P0(eta)}, the final transformation is defined by
\begin{equation}\label{eq-R(eta)-by-P}
R(\eta)=\begin{cases}
(\lambda\rho_{2})^{\frac{1}{4}\sigma_{3}}P(\eta)P^{(r)}(\eta)^{-1}(\lambda\rho_{2})^{-\frac{1}{4}\sigma_{3}},&\eta\in U(r;\rho_{2}),\\
(\lambda\rho_{2})^{\frac{1}{4}\sigma_{3}}P(\eta)P^{(\infty)}(\eta)^{-1}(\lambda\rho_{2})^{-\frac{1}{4}\sigma_{3}},&\eta\in \mathbb{C}\setminus U(r;\rho_{2}).
\end{cases}
\end{equation}
It is straightforward to check that $R$ satisfies the following RH problem.

\begin{figure}
\centering

\tikzset{every picture/.style={line width=0.75pt}} 

\begin{tikzpicture}[x=0.75pt,y=0.75pt,yscale=-1,xscale=1]

\draw    (178,82) -- (235.53,136.16) ;
\draw [shift={(209.53,111.69)}, rotate = 223.27] [fill={rgb, 255:red, 0; green, 0; blue, 0 }  ][line width=0.08]  [draw opacity=0] (7.14,-3.43) -- (0,0) -- (7.14,3.43) -- cycle    ;
\draw    (236.53,172.16) -- (179.28,225.09) ;
\draw [shift={(211.8,195.03)}, rotate = 137.24] [fill={rgb, 255:red, 0; green, 0; blue, 0 }  ][line width=0.08]  [draw opacity=0] (7.14,-3.43) -- (0,0) -- (7.14,3.43) -- cycle    ;
\draw   (228.14,153.55) .. controls (228.14,139.74) and (239.34,128.55) .. (253.14,128.55) .. controls (266.95,128.55) and (278.14,139.74) .. (278.14,153.55) .. controls (278.14,167.35) and (266.95,178.55) .. (253.14,178.55) .. controls (239.34,178.55) and (228.14,167.35) .. (228.14,153.55) -- cycle ;
\draw  [fill={rgb, 255:red, 0; green, 0; blue, 0 }  ,fill opacity=1 ] (252.14,153.55) .. controls (252.14,152.99) and (252.59,152.55) .. (253.14,152.55) .. controls (253.69,152.55) and (254.14,152.99) .. (254.14,153.55) .. controls (254.14,154.1) and (253.69,154.55) .. (253.14,154.55) .. controls (252.59,154.55) and (252.14,154.1) .. (252.14,153.55) -- cycle ;
\draw    (250.14,128.55) -- (258.14,128.55) ;
\draw [shift={(257.94,128.55)}, rotate = 180] [fill={rgb, 255:red, 0; green, 0; blue, 0 }  ][line width=0.08]  [draw opacity=0] (7.14,-3.43) -- (0,0) -- (7.14,3.43) -- cycle    ;
\draw (250,155) node [anchor=north west][inner sep=0.75pt]   [align=left] {{\fontfamily{ptm}\selectfont \textit{r}}};


\end{tikzpicture}
\caption{The jump contour $\Sigma^{R}$ of the RH problem for $R$ in Section \ref{sec:RH-analysis-case-II-}.}
\label{fig-Sigma-R-case-II}
\end{figure}
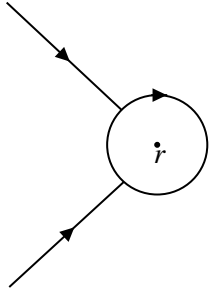

\paragraph{RH problem for $R$}
\begin{enumerate}
\item [(a)] $R(\eta)$ is analytic in $\mathbb{C}\setminus \Sigma^{R}$, where $\Sigma^{R}=\left(\Sigma_{1}^{P}\cup \Sigma_{3}^{P}\cup \partial U(r;\rho_{2})\right)\setminus U(r;\rho_{2})$ is shown in Figure \ref{fig-Sigma-R-case-II}.
\item [(b)] On $\Sigma^{R}$, the limiting values $R_{\pm}(\eta)$ exist and satisfy the jump condition 
\begin{equation}\label{eq-jump-R-case-II-}
    R_{+}(\eta)=R_{-}(\eta)v_{R}(\eta),
\end{equation}
where
\begin{eqnarray}\label{eq-approx-v-R-case-II-}
v_{R}(\eta)=
\begin{cases}
I+\mathcal{O}(\lambda^{\frac{2k}{3}+\frac{1}{2}}e^{-C_{2}\lambda^{\frac{k}{2}}}),  &\eta\in\Sigma^{R}\setminus \partial U(r;\rho_{2}),\\
I+\left(\begin{matrix}
    \mathcal{O}(\rho_{2}) & \mathcal{O}(\rho_{2}^{\frac{1}{2}})\\
    \mathcal{O}(\rho_{2}^{\frac{1}{2}}) & \mathcal{O}(\rho_{2})
\end{matrix}\right), &\eta\in\partial U(r;\rho_{2})
\end{cases}
\end{eqnarray}
as $\lambda\to+\infty$.
\item [(c)] As $\eta\to\infty$, we have $R(\eta)=I+\mathcal{O}(\eta^{-1})$.
\end{enumerate}

From \eqref{eq-approx-v-R-case-II-}, it is readily seen that that $v_{R}(\eta)$ tends to identity matrix exponentially fast as $\lambda\to+\infty$ uniformly for all $\eta\in\Sigma^{R}\setminus \partial U(r;\rho_{2})$. When $\eta\in\partial U(r;\rho_{2})$, we see that $v_{R}(\eta)=I+\mathcal{O}(\rho_{2}^{\frac{1}{2}})$.  Following the standard argument of the small norm RH problem (cf. \cite{Deift1999,DX93}), we conclude that 
\begin{equation}\label{eq-appprox-R1-standard}
R(\eta)=I+\mathcal{O}\left(\rho_{2}^{\frac{1}{2}}\right) 
\end{equation}
as $\lambda\to+\infty$ uniformly for all $\eta\in U(r;\rho_{2})$ and $\alpha_{k}r^{2k+1}+y\in[-C_{1}\lambda^{-k-1+\delta},0]$. Moreover, we have the following estimates. 

\begin{lemma}\label{lem-R(r)-R'(r)-case-II-}
    As $\lambda\to+\infty$, we have
    \begin{equation}\label{eq-R(r)-R'(r)-case-II-}
    R(r)=I+\left(\begin{matrix}
    \mathcal{O}\left(\rho_{2}\right) & \mathcal{O}\left(\rho_{2}^{\frac{1}{2}}\right)\\
    \mathcal{O}\left(\rho_{2}^{\frac{1}{2}}\right)& \mathcal{O}\left(\rho_{2}\right)
\end{matrix}\right), \qquad R'(r)=  \begin{pmatrix}
    \mathcal{O}(1) & \mathcal{O}\left(\rho_{2}^{-\frac{1}{2}}\right)\\
     \mathcal{O}\left(\rho_{2}^{-\frac{1}{2}}\right) & \mathcal{O}(1)
\end{pmatrix},  
    \end{equation}
uniformly for $\alpha_{k}r^{2k+1}+y\in[-C_{1}\lambda^{-k-1+\delta},0]$, where $\rho_{2}$ is defined in \eqref{eq-def-rho2} and $\delta\in(0,\frac{1}{4})$ is an arbitrarily fixed constant.
\end{lemma}
\begin{proof}
From \eqref{eq-jump-R-case-II-}, we have
\begin{equation}
    R_{+}(\eta)=R_{-}(\eta)+(v_{R}(\eta)-I)+(R_{-}(\eta)-I)(v_{R}(\eta)-I), \quad \eta\in\partial U(r;\rho_{2}).
\end{equation}
By Plemelj's formula, it follows that
\begin{equation}\label{eq-R1-integral-representation}
    R(\eta)=I+\frac{1}{2\pi i}\oint_{\partial U(r;\rho_{2})}\frac{(v_{R}(\xi)-I)+(R_{-}(\xi)-I)(v_{R}(\xi)-I)}{\xi-\eta}d\xi.
\end{equation}
For $\eta\in\partial U(r;\rho_{2})$, we obtain from \eqref{eq-approx-v-R-case-II-} and \eqref{eq-appprox-R1-standard} that 
\begin{equation}\label{eq-approx-v1+R1v1}
    (v_{R}(\eta)-I)+(R_{-}(\eta)-I)(v_{R}(\eta)-I)=\left(\begin{matrix}
    \mathcal{O}\left(\rho_{2}\right) & \mathcal{O}\left(\rho_{2}^{\frac{1}{2}}\right)\\
    \mathcal{O}\left(\rho_{2}^{\frac{1}{2}}\right)& \mathcal{O}\left(\rho_{2}\right)
\end{matrix}\right), \qquad \lambda\to +\infty.
\end{equation}
This, together with \eqref{eq-R1-integral-representation}, implies the estimate of $R(r)$ in \eqref{eq-R(r)-R'(r)-case-II-}. Making use of this fact and the Cauchy integral theorem, we further obtain that
\begin{align}\label{eq-R'(eta)-approx}
    R'(r)&=\frac{1}{2\pi i}\oint_{\partial U(r;\rho_{2})}\frac{(v_{R}(\xi)-I)+(R_{-}(\xi)-I)(v_{R}(\xi)-I)}{(\xi-r)^2}d\xi
     \nonumber \\ 
    &=\left(\begin{matrix}
    \mathcal{O}(1) & \mathcal{O}\left(\rho_{2}^{-\frac{1}{2}}\right)\\
     \mathcal{O}\left(\rho_{2}^{-\frac{1}{2}}\right) & \mathcal{O}(1)
\end{matrix}\right), \qquad \lambda\to +\infty,
\end{align}
uniformly for $\alpha_{k}r^{2k+1}+y\in[-C_{1}\lambda^{-k-1+\delta},0]$.  This finishes the proof of Lemma \ref{lem-R(r)-R'(r)-case-II-}.
\end{proof}

\subsection{Riemann-Hilbert analysis when $\alpha_{k}r^{2k+1}+y\in[0,C_{1}\lambda^{-k-1+\delta}]$}
\label{sec:RH-analysis-case-II+}
\paragraph{Construction of the $g$-function}

Since the conclusion in Lemma \ref{lem-g1(eta)} is not applicable in this case, we have to introduce a new $g$-function in the form
\begin{equation}\label{def:g_2}
g_{2}(\eta)=(\eta-r_{0})^{\frac{3}{2}}p_{2}(\eta),
 \qquad \eta \in \mathbb{C} \setminus 
(-\infty,r_0],
\end{equation}
where $p_{2}(\eta)$ is a polynomial of degree $2k$ and $r_{0}$ is a constant to be determined. 
It is required that $g_2$ satisfies the matching condition
\begin{equation}\label{eq-matching-g2-theta}
  \theta(\eta;y)-g_{2}(\eta)=d_{2}\eta^{-\frac{1}{2}}+\mathcal{O}(\eta^{-\frac{3}{2}}), \qquad \eta \to \infty,
\end{equation}
where $\theta(\eta;y)$ is defined in \eqref{def:thetaeta} and $d_{2}$ is a constant independent of $\eta$ whose explicit value is unimportant. A direct calculation yields that 
\begin{equation}\label{eq-def-r0}
 r_{0}=-\left(\frac{y}{\alpha_{k}}\right)^{\frac{1}{2k+1}}
\end{equation}
and 
\begin{align}\label{eq-def-p2(eta)}
p_{2}(\eta)=\sum\limits_{j=0}^{2k}\frac{4}{4k+3}\frac{\Gamma\left(j+\frac{3}{2}\right)}{\Gamma\left(j+1\right)\Gamma\left(\frac{3}{2}\right)}r_{0}^{j}\eta^{2k-j}.
\end{align}
This implies that 
\begin{equation}
g_{2}(\eta)=(\eta-r_{0})^{\frac{3}{2}}\sum\limits_{j=0}^{2k}\frac{4}{4k+3}\frac{\Gamma\left(j+\frac{3}{2}\right)}{\Gamma\left(j+1\right)\Gamma\left(\frac{3}{2}\right)}r_{0}^{j}\eta^{2k-j}, \qquad \eta \in  \mathbb{C} \setminus 
(-\infty,r_0].
\end{equation}

\paragraph{First transformation} Utilizing the above $g_{2}$, we can normalize $X$ as
\begin{equation}\label{def-P}
\tilde{P}(\eta)= \left(\begin{matrix}
1&0\\-d_{2}\lambda^{2k+2}&1
\end{matrix}
\right)X(\lambda\eta)\begin{cases}
\begin{pmatrix}
1 & 0\\
1 & 1
\end{pmatrix}e^{\lambda^{\frac{4k+3}{2}}g_2(\eta)\sigma_3}, \quad & \eta \in \textrm{I},\\
\begin{pmatrix}
1 & 0\\
-1 & 1
\end{pmatrix}e^{\lambda^{\frac{4k+3}{2}}g_2(\eta)\sigma_3}, \quad & \eta \in \textrm{II},\\
e^{\lambda^{\frac{4k+3}{2}}g_2(\eta)\sigma_3}, \quad & \textrm{elsewhere},
\end{cases}
\end{equation}
where regions I and II are illustrated in Figure \ref{fig:tildeP} and the constant $d_{2}$ is given in \eqref{eq-matching-g2-theta}. It is evident that $\tilde{P}$ satisfies the following RH problem.

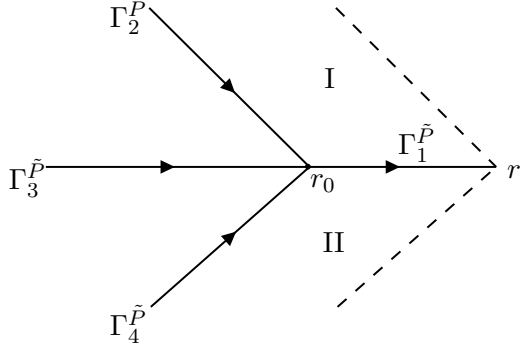
\begin{figure}[t]
\begin{center}
\tikzset{every picture/.style={line width=0.75pt}} 
\begin{tikzpicture}[x=0.75pt,y=0.75pt,yscale=-0.8,xscale=0.8]

\draw    (72,149) -- (223,149) -- (353,149) ;
\draw [shift={(152.5,149)}, rotate = 180] [fill={rgb, 255:red, 0; green, 0; blue, 0 }  ][line width=0.08]  [draw opacity=0] (8.93,-4.29) -- (0,0) -- (8.93,4.29) -- cycle    ;
\draw [shift={(293,149)}, rotate = 180] [fill={rgb, 255:red, 0; green, 0; blue, 0 }  ][line width=0.08]  [draw opacity=0] (8.93,-4.29) -- (0,0) -- (8.93,4.29) -- cycle    ;
\draw    (136.5,49) -- (236.5,149) ;
\draw [shift={(190.04,102.54)}, rotate = 225] [fill={rgb, 255:red, 0; green, 0; blue, 0 }  ][line width=0.08]  [draw opacity=0] (8.93,-4.29) -- (0,0) -- (8.93,4.29) -- cycle    ;
\draw    (137.5,237) -- (236.5,149) ;
\draw [shift={(190.74,189.68)}, rotate = 138.37] [fill={rgb, 255:red, 0; green, 0; blue, 0 }  ][line width=0.08]  [draw opacity=0] (8.93,-4.29) -- (0,0) -- (8.93,4.29) -- cycle    ;
\draw  [fill={rgb, 255:red, 0; green, 0; blue, 0 }  ,fill opacity=1 ] (236.67,148.83) .. controls (236.67,148.19) and (236.15,147.66) .. (235.5,147.66) .. controls (234.85,147.66) and (234.33,148.19) .. (234.33,148.83) .. controls (234.33,149.48) and (234.85,150) .. (235.5,150) .. controls (236.15,150) and (236.67,149.48) .. (236.67,148.83) -- cycle ;
\draw  [dash pattern={on 4.5pt off 4.5pt}]  (253,49) -- (353,149) ;
\draw  [dash pattern={on 4.5pt off 4.5pt}]  (254,237) -- (353,149) ;

\draw (235,152) node [anchor=north west][inner sep=0.75pt]   [align=left] {$r_0$};
\draw (358,146) node [anchor=north west][inner sep=0.75pt]   [align=left] {$r$};
\draw (110,40) node [anchor=north west][inner sep=0.75pt]   [align=left] {$\Gamma_{2}^{\tilde{P}}$};
\draw (47,142) node [anchor=north west][inner sep=0.75pt]   [align=left] {$\Gamma_{3}^{\tilde{P}}$};
\draw (110,234) node [anchor=north west][inner sep=0.75pt]   [align=left] {$\Gamma_{4}^{\tilde{P}}$};
\draw (244,85) node [anchor=north west][inner sep=0.75pt]   [align=left] {I};
\draw (243,188) node [anchor=north west][inner sep=0.75pt]   [align=left] {II};
\draw (290,120) node [anchor=north west][inner sep=0.75pt]   [align=left] {$\Gamma_{1}^{\tilde{P}}$};

\end{tikzpicture}
\caption{Regions I and II and the jump contours of the RH problem for $\tilde{P}$.}
   \label{fig:tildeP}
\end{center}
\end{figure}

\paragraph{RH problem for $\tilde{P}$}
\begin{enumerate}
\item[(a)] $\tilde{P}(\eta)$ is analytic in $\mathbb{C} \setminus \Gamma^{\tilde{P}}$, where $\Gamma^{\tilde{P}}:=\Gamma_{1}^{\tilde{P}}\cup\Gamma_{2}^{\tilde{P}}\cup\Gamma_{3}^{\tilde{P}}\cup\Gamma_{4}^{\tilde{P}}$ is shown in Figure \ref{fig:tildeP}.
\item[(b)] On $\Gamma^{\tilde{P}}$, the limiting values $\tilde{P}_{\pm}(\eta)$ exist and satisfy the jump condition
\begin{equation}
\tilde{P}_+(\eta)=\tilde{P}_+(\eta)\begin{cases}
\begin{pmatrix}
1 & e^{-2\lambda^{\frac{4k+3}{2}} g_2(\eta)}\\
0 & 1
\end{pmatrix}, \quad & \eta \in \Gamma^{\tilde{P}}_{1},\\
\begin{pmatrix}
1 & 0\\
e^{2\lambda^{\frac{4k+3}{2}} g_2(\eta)} & 1
\end{pmatrix}, \quad & \eta \in \Gamma^{\tilde{P}}_{2} \cup \Gamma_{4}^{\tilde{P}},\\
\begin{pmatrix}
0 & 1\\
-1 & 0
\end{pmatrix}, \quad & \eta \in \Gamma_{3}^{\tilde{P}}.
\end{cases}
\end{equation}
\item [(c)] As $\eta\to r$, we have $\tilde{P}(\eta)=\mathcal{O}(\log{(\eta-r)})$.
\item[(d)] As $\eta \to \infty$, we have 
\begin{equation}
\tilde{P}(\eta)=\left(I + \Boh(\eta^{-1})\right)(\lambda\eta)^{-\frac 14 \sigma_3}N,
\end{equation}
where $N$ is defined in \eqref{eq-def-N}.
\end{enumerate}

Note that 
\begin{equation}\label{def:g_2y}
    g_{2}(|y|^{\frac{1}{2k+1}}\eta)=|y|^{\frac{4k+3}{4k+2}}g(\eta), \qquad r_{0}=|y|^{\frac{1}{2k+1}}z_{0},
\end{equation}
where $g$ is defined in \eqref{eq:gfunction-appendix} and $z_{0}$ is specified in \eqref{eq:zzero}. These facts, together with \cite[Proposition 2.5 and Corollary 2.6]{Claeys-2012}, yield that $\Re g_2(\eta) < 0$ for $\eta \in \Gamma^{\tilde{P}}_2 \cup \Gamma^{\tilde{P}}_4$ and $\Re g_2(\eta) > 0$ for $\eta \in \Gamma^{\tilde{P}}_1$, provided that $\eta\neq r_{0}$. Moreover, we can show that there exist two fixed constants $M_{3},\tilde{\rho}_{2}>0$ such that
\begin{align}
\Re g_{2}(\eta)&\leq
-M_{3}, \quad \eta\in\Gamma_{2}^{\tilde{P}}\cup\Gamma_{4}^{\tilde{P}}, \label{g2-prop1}\\
\Re g_{2}(\eta)&\geq
M_{3},\qquad \eta\in\Gamma_{1}^{\tilde{P}}, \label{g2-prop2}
\end{align}
provided that $|\eta-r_{0}|\geq \tilde{\rho}_{2}$. This implies that, as $\lambda \to \infty$, the jump matrices of $\tilde{P}$ converge exponentially fast to the identity matrix, uniformly for $\eta\in(\Gamma^{\tilde{P}}_{1}\cup\Gamma^{\tilde{P}}_{2}\cup \Gamma^{\tilde{P}}_{4})\setminus U(r_{0};\tilde{\rho}_{2})$, where $U(r_{0}; \tilde{\rho}_{2})$ is a fixed neighborhood of $\eta=r_{0}$.

\paragraph{Global parametrix}
In view of the fact that $\alpha_{k}r^{2k+1}+y\in[0,C_{1}\lambda^{-k-1+\delta}]$, it follows from \eqref{eq-def-r0} that $r-r_{0}=\mathcal{O}(\lambda^{-k-1+\delta})$ as $\lambda\to+\infty$. This indicates that $r$ lies within the open disk $U(r_{0};\tilde{\rho}_{2})$ for large $\lambda$. Hence, we proceed to define the global parametrix in this case by
\begin{equation}\label{def-tildeP-infty}
\tilde{P}^{(\infty)}(\eta)=(\lambda(\eta-r_{0}))^{-\frac 14 \sigma_3} N, 
\end{equation}
which solves the following RH problem.

\paragraph{RH problem for $\tilde{P}^{(\infty)}$}
\begin{enumerate}
    \item [(a)] $\tilde{P}^{(\infty)}(\eta)$ is analytic in $\mathbb{C}\setminus(-\infty,r_{0}]$.
    \item [(b)] $\tilde{P}^{(\infty)}$ satisfies the jump condition
    \begin{equation}
        \tilde{P}^{(\infty)}_{+}(\eta)=\tilde{P}^{(\infty)}_{-}(\eta)\left(\begin{matrix}
            0&1\\-1&0
        \end{matrix}\right), \qquad \eta\in(-\infty,r_{0}).
    \end{equation}
    \item [(c)] As $\eta\to\infty$, we have 
    \begin{equation}
       \tilde{P}^{(\infty)}(\eta)=(I+\mathcal{O}(\eta^{-1}))(\lambda\eta)^{-\frac{1}{4} \sigma_3}N. 
    \end{equation}
\end{enumerate}

\paragraph{Local parametrix}
Since $\tilde P$ and $\tilde P^{(\infty)}$ are not uniformly close to each other in $U(r_{0}; \tilde{\rho}_2)$, it is necessary to construct a local parametrix that satisfies the following RH problem.

\paragraph{RH problem for $\tilde{P}^{(r_{0})}$}
\begin{enumerate}
\item[(a)] $\tilde{P}^{(r_{0})}(\eta)$ is analytic in $U(r_{0};\tilde{\rho}_2) \setminus \Gamma^{\tilde{P}}$.
\item[(b)] On  $U(r_{0};\tilde{\rho}_2) \cap \Gamma^{\tilde{P}}$, the limiting values $\tilde{P}_{\pm}^{(r_{0})}$ exist and satisfy the jump condition
\begin{equation}
\tilde{P}_{+}^{(r_{0})}(\eta)=\tilde{P}_{+}^{(r_{0})}(\eta)\begin{cases}
\begin{pmatrix}
1 & e^{-2\lambda^{\frac{4k+3}{2}} g_2(\eta)}\\
0 & 1
\end{pmatrix}, \quad & \eta \in \Gamma^{\tilde{P}}_{1}\cap U(r_{0};\tilde{\rho}_2),\\
\begin{pmatrix}
1 & 0\\
e^{2\lambda^{\frac{4k+3}{2}} g_2(\eta)} & 1
\end{pmatrix}, \quad & \eta \in (\Gamma^{\tilde{P}}_{2} \cup \Gamma_{4}^{\tilde{P}})\cap U(r_{0};\tilde{\rho}_2),\\
\begin{pmatrix}
0 & 1\\
-1 & 0
\end{pmatrix}, \quad & \eta \in \Gamma_{3}^{\tilde{P}} \cap  U(r_{0};\tilde{\rho}_{2}).
\end{cases}
\end{equation}
\item[(c)] As $\lambda \to\infty$, $\tilde P^{(r_{0})}$ satisfies the matching condition 
\begin{equation}
\label{eq-matching-Pr0-Pinfty}
    \tilde{P}^{(r_{0})}(\eta)= (I+o(1))\tilde{P}^{(\infty)}(\eta), \qquad \eta \in \partial U(r_{0};\tilde{\rho}_{2}),
\end{equation}
where $\tilde{P}^{(\infty)}$ is given in \eqref{def-tildeP-infty}.
\end{enumerate}

To solve this RH problem, we introduce the function
\begin{equation}\label{eq-def-tilde-f2}
f_{3}(\eta;r_{0})=\left(\frac{3}{2} g_2(\eta)\right)^{\frac 23},
\end{equation}
where $g_2(\eta)$ is defined in \eqref{def:g_2}. Clearly, $f_{3}(\eta;r_{0})$ is analytic in $U(r_{0};\tilde{\rho}_{2})$. Indeed, we have 
\begin{equation}\label{eq-def-f2(eta)}
f_{3}(\eta;r_{0})=\left(\frac{3p_{2}(r_{0})}{2}\right)^{\frac{2}{3}}(\eta-r_{0}) + \Boh((\eta-r_{0})^2), \qquad  \eta \to r_{0},
\end{equation}
where $p_{2}(r_{0})>0$ by \eqref{eq-def-p2(eta)}. 
Thus, $f_{3}(\eta;r_{0})$ is a conformal mapping for $\eta\in U(r_{0},\tilde{\rho}_{2})$. 
We next define
\begin{equation}\label{def:tildeP-cs}
\tilde{P}^{(r_{0})}(\eta)=\tilde{E}_2(\eta) \tilde{\Phi}(\lambda^{\frac{4k+3}{3}} f_{3}(\eta;r_{0}); \lambda^{\frac{4k+3}{3}} f_{3}(r;r_{0}))\exp\left\{\lambda^{\frac{4k+3}{2}}g_2(\eta)\sigma_3\right\},
\end{equation}
where
\begin{equation}\label{def:tildeE2}
\tilde{E}_2(\eta):=(\lambda(\eta-r_{0}))^{-\frac{1}{4} \sigma_3}(\lambda^{\frac{4k+3}{3}} f_{3}(\eta;r_{0}))^{\frac{1}{4} \sigma_3}e^{-\frac{\pi i}{4} \sigma_3}
\end{equation}
is an analytic prefactor in $U(r_{0};\tilde{\rho_{2}})$ and $\tilde{\Phi}(\zeta;x)$ is a variant of the  $\mathrm{P_{34}}$ 
parametrix defined in \eqref{eq-def-tilde-Phi}. From the RH problem  for $\tilde{\Phi}$ given in Appendix \ref{sec-appendix-p34-parametrix}, one can easily check that that $\tilde{P}^{(r_{0})}$ in \eqref{def:tildeP-cs} satisfies items (a) and (b) in the RH problem for $\tilde{P}^{(r_{0})}$. To show the matching condition \eqref{eq-matching-Pr0-Pinfty}, we see from \eqref{eq-tilde-Phi-asym-infty}, \eqref{eq-tilde-Phi-asym-uniform}, \eqref{def-tildeP-infty} and \eqref{def:tildeP-cs} that, as $\lambda \to +\infty$,
\begin{align}
&\tilde{P}^{(r_{0})}(\eta)\tilde{P}^{(\infty)}(\eta)^{-1} 
\nonumber \\
&=\left(\frac{\lambda(\eta-r_{0})}{\lambda^{\frac{4k+3}{3}} f_{3}(\eta;r_{0})}\right)^{-\frac{1}{4}\sigma_{3}}\left(I + \Boh(\lambda^{-\frac{4k+3}{3}}f_{3}(\eta;r_{0})^{-1})\right)\left(\frac{\lambda(\eta-r_{0})}{\lambda^{\frac{4k+3}{3}} f_{3}(\eta;r_{0})}\right)^{\frac{1}{4}\sigma_{3}} \nonumber \\
&=I+\Boh(\lambda^{-\frac{2k}{3}-1}), \qquad \eta\in\partial U(r;\tilde{\rho}_{2}),
\end{align}
uniformly for $\alpha_{k}r^{2k+1}+y\in[0,C_{1}\lambda^{-k-1+\delta}]$.

\paragraph{Final transformation} 
The final transformation is defined by
\begin{equation}\label{def-R-case-II+}
R(\eta)=\begin{cases}
\tilde{P}(\eta)\tilde{P}^{(r_{0})}(\eta)^{-1}, \quad & \eta \in U(r_{0}; \tilde{\rho}_{2}),\\
\tilde{P}(\eta)\tilde{P}^{(\infty)}(\eta)^{-1}, \quad & \eta \in \mathbb{C} \setminus U(r_{0}; \tilde{\rho}_{2}).
\end{cases}
\end{equation}
From the global and local parametrices built in \eqref{def-tildeP-infty} and \eqref{def:tildeP-cs} it is easy to check that $R$ satisfies the following RH problem.
\begin{figure}
\centering

\tikzset{every picture/.style={line width=0.75pt}} 

\begin{tikzpicture}[x=0.75pt,y=0.75pt,yscale=-1,xscale=1]

\draw    (198,79) -- (255.53,133.16) ;
\draw [shift={(229.53,108.69)}, rotate = 223.27] [fill={rgb, 255:red, 0; green, 0; blue, 0 }  ][line width=0.08]  [draw opacity=0] (7.14,-3.43) -- (0,0) -- (7.14,3.43) -- cycle    ;
\draw    (256.53,169.16) -- (199.28,222.09) ;
\draw [shift={(231.8,192.03)}, rotate = 137.24] [fill={rgb, 255:red, 0; green, 0; blue, 0 }  ][line width=0.08]  [draw opacity=0] (7.14,-3.43) -- (0,0) -- (7.14,3.43) -- cycle    ;
\draw   (248.14,150.55) .. controls (248.14,136.74) and (259.34,125.55) .. (273.14,125.55) .. controls (286.95,125.55) and (298.14,136.74) .. (298.14,150.55) .. controls (298.14,164.35) and (286.95,175.55) .. (273.14,175.55) .. controls (259.34,175.55) and (248.14,164.35) .. (248.14,150.55) -- cycle ;
\draw  [fill={rgb, 255:red, 0; green, 0; blue, 0 }  ,fill opacity=1 ] (272.14,150.55) .. controls (272.14,149.99) and (272.59,149.55) .. (273.14,149.55) .. controls (273.69,149.55) and (274.14,149.99) .. (274.14,150.55) .. controls (274.14,151.1) and (273.69,151.55) .. (273.14,151.55) .. controls (272.59,151.55) and (272.14,151.1) .. (272.14,150.55) -- cycle ;
\draw    (270.14,125.55) -- (278.14,125.55) ;
\draw [shift={(277.94,125.55)}, rotate = 180] [fill={rgb, 255:red, 0; green, 0; blue, 0 }  ][line width=0.08]  [draw opacity=0] (7.14,-3.43) -- (0,0) -- (7.14,3.43) -- cycle    ;
\draw  [fill={rgb, 255:red, 0; green, 0; blue, 0 }  ,fill opacity=1 ] (285.14,150.55) .. controls (285.14,149.99) and (285.59,149.55) .. (286.14,149.55) .. controls (286.69,149.55) and (287.14,149.99) .. (287.14,150.55) .. controls (287.14,151.1) and (286.69,151.55) .. (286.14,151.55) .. controls (285.59,151.55) and (285.14,151.1) .. (285.14,150.55) -- cycle ;

\draw (283,153) node [anchor=north west][inner sep=0.75pt]   [align=left] {$r$};
\draw (268,153) node [anchor=north west][inner sep=0.75pt]   [align=left] {$r_0$};

\end{tikzpicture}
\caption{The jump contour of the RH problem for $R$ in Section \ref{sec:RH-analysis-case-II+}}
\label{fig-Gamma-R}
\end{figure}
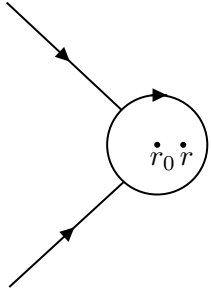

\paragraph{RH problem for $R$}
\begin{itemize}
\item[(a)] $R(\eta)$ is defined and analytic in $\mathbb{C} \setminus \Gamma^R$, where $
\Gamma^R:=  \left(\Gamma^{\tilde{P}}_{2}\cup \Gamma^{\tilde{P}
}_{4} \cup \partial U(r_{0}; \tilde{\rho}_{2})\right)\setminus U(r_{0}; \tilde{\rho}_{2})$ is shown in Figure \ref{fig-Gamma-R}.
\item[(b)] On $\Gamma^{R}$, the limiting values $R_{\pm}(\eta)$ exist and satisfy the jump condition
\begin{align}
R_+(\eta)=R_-(\eta) v_R(\eta), 
\end{align}
where
\begin{align}\label{jump:R1-cs}
v_R(\eta) = \begin{cases}
I+\mathcal{O}(\lambda^{-\frac{2k}{3}-1}), & \qquad \eta \in \partial U(r_{0}; \tilde{\rho}_{2}),
\\
I+\mathcal{O}(\lambda^{-\frac{1}{2}}e^{-M_{3}\lambda^{\frac{4k+3}{2}}}), & \qquad \eta \in \Gamma_R \setminus \partial U(r_{0}; \tilde{\rho}_{2}),
\end{cases}
\end{align}
as $\lambda\to+\infty$.
\item[(c)] As $\eta \to \infty$, we have $R(\eta) = I + \Boh(\eta^{-1}).$
\end{itemize}

By the standard argument of the small norm RH problem, we conclude that, as $\lambda\to+\infty$,
\begin{align}\label{eq-approximation-R}
R(\eta)=I+\Boh(\lambda^{-\frac{2k}{3}-1}), \qquad 
R'(r)=\Boh(\lambda^{-\frac{2k}{3}-1}),
\end{align}
uniformly for $\alpha_{k}r^{2k+1}+y\in[0,C_{1}\lambda^{-k-1+\delta}]$.

\subsection{Proof of Lemma \ref{lem-dF/dtau-case-II}}
\label{sec:case-II-proof}
We split the proof into two parts, corresponding to the analysis performed in Sections \ref{sec:RH-analysis-case-II-} and \ref{sec:RH-analysis-case-II+}, respectively. 

\subsubsection*{The case when  $\alpha_{k}r^{2k+1}+y\in[-C_{1}\lambda^{-k-1+\delta},0]$}

This part is an outcome of the analysis carried out in Section \ref{sec:RH-analysis-case-II-}. From Lemma \ref{lem-differential-indentity-2} and the definition of $P$ in \eqref{eq-def-P(eta)-case-II}, we obtain
\begin{equation}
\frac{\partial F}{\partial s}(s;x)=\left.\frac{1}{2\pi i\lambda}(P(\eta)^{-1}P'(\eta))_{21}\right|_{\eta\to r}.
\end{equation}
By inverting the transformation in \eqref{eq-R(eta)-by-P}, it follows from the above equation that
\begin{multline}
    \frac{\partial F}{\partial s}(s;x)
    \\ =\left.\frac{1}{2\pi i\lambda}\left[(P^{(r)}(\eta)^{-1}P^{(r)'}(\eta))_{21}+(P^{(r)}(\eta)^{-1} (\lambda \rho_{2})^{-\frac{1}{4}\sigma_{3}} R(\eta)^{-1}R'(\eta) (\lambda \rho_{2})^{\frac{1}{4}\sigma_{3}} P^{(r)}(\eta))_{21}\right]\right|_{\eta\to r}.
\end{multline}
Substituting $P^{(r)}$ in \eqref{eq-def-P0(eta)} into the above formula, we obtain from the local behavior of $\Phi$ near the origin given in \eqref{eq-appendix-Phi-zeto-to-0} and \eqref{eq-Phi-0} that
\begin{equation}\label{eq-approx-dF/dtau-three-parts-1}
\begin{split}
\frac{\partial F}{\partial s}(s;x)=\Theta_{1}(r)+\Theta_{2}(r)+\Theta_{3}(r),
\end{split}
\end{equation}
where
\begin{align}\label{eq-approx-dF/dtau-three-parts}
\Theta_{1}(r)&=\frac{\lambda^{\frac{4k+3}{3}}}{2\pi i\lambda}f_{2}'(r)\left(\Phi_{1}(\lambda^{\frac{4k+3}{3}}\kappa_{0})\right)_{21},\\
\Theta_{2}(r)&=\frac{1}{2\pi i\lambda}\left(\Phi_{0}(\lambda^{\frac{4k+3}{3}}\kappa_{0})^{-1}\Xi_{1}(r)\Phi_{0}(\lambda^{\frac{4k+3}{3}}\kappa_{0})\right)_{21},\\
\Theta_{3}(r)&=\frac{1}{2\pi i\lambda}\left(\Phi_{0}(\lambda^{\frac{4k+3}{3}}\kappa_{0})^{-1}\Xi_{2}(r)\Phi_{0}(\lambda^{\frac{4k+3}{3}}\kappa_{0})\right)_{21},
\end{align}
with $\Xi_{1}(\eta)=E_{2}(\eta)^{-1}E_{2}'(\eta)$, $\Xi_{2}(\eta)=E_{2}(\eta)^{-1}(\lambda \rho_{2})^{-\frac{1}{4}\sigma_{3}}R(\eta)^{-1}R'(\eta)(\lambda \rho_{2})^{\frac{1}{4}\sigma_{3}}E_{2}(\eta)$ and $\kappa_0$ given in \eqref{eq-def-kappa0}. 

We next estimate $\Theta_i(r)$, $i=1,2,3$, in \eqref{eq-approx-dF/dtau-three-parts-1}, respectively. In view of the definition of $E_{2}$ in \eqref{eq-def-E2(eta)}, it is readily seen that $E_{2}(r)$ is diagonal and
\begin{equation}
\Xi_{1}(r)=\mathcal{O}(1), \qquad \lambda \to+\infty.
\end{equation}
A further appeal to the estimate of $R(r)$ and $R'(r)$ in \eqref{eq-R(r)-R'(r)-case-II-} gives us  
\begin{equation}
\Xi_{2}(r)=\mathcal{O}(\lambda^{\frac{2k}{3}+\frac{1}{2}}), \qquad \lambda \to+\infty.
\end{equation}
It is then readily seen from the above two estimates and large $\lambda$ behavior of $\Phi_{0}(-\lambda)$ in \eqref{eq-approx-Phi0(-xi)} that as $\lambda \to +\infty$,
\begin{equation}\label{eq-approx-Phi20r}
\Theta_{2}(r)=
\mathcal{O}(\lambda^{\frac{k}{3}-1+\delta}) \quad \textrm{and} \quad \Theta_{3}(r)=
\mathcal{O}(\lambda^{k-\frac{1}{2}+\delta}),
\end{equation}
uniformly for $\alpha_{k}r^{2k+1}+y\in[-C_{1}\lambda^{-k-1+\delta},0]$. To estimate $\Theta_1$, we see from \eqref{eq-approx-dF/dtau-three-parts} and \eqref{eq-f2'0} that 
\begin{align}\label{eq-approx-Phi10r}
\Theta_{1}(r)&=\frac{\lambda^{\frac{4k+3}{3}}}{2\pi i\lambda}\frac{\partial\kappa_{0}}{\partial r}\left(\Phi_{1}(\lambda^{\frac{4k+3}{3}}\kappa_{0})\right)_{21}\left(1+\mathcal{O}(\lambda^{-k-1+\delta})\right) \nonumber \\
&=\frac{1}{2\pi i}\frac{\partial \chi(s;x)}{\partial s}\left(\Phi_{1}(\chi(s;x)\right)_{21}\left(1+\mathcal{O}(\lambda^{-k-1+\delta})\right), \qquad \lambda \to +\infty,
\end{align}
where $\chi(s;x)$ is defined in \eqref{eq-def-chi(tau)} with $s=r\lambda$ and $x=y\lambda^{2k+1}$. 

Finally, substituting \eqref{eq-approx-Phi20r} and \eqref{eq-approx-Phi10r} into \eqref{eq-approx-dF/dtau-three-parts-1}, we obtain the  desired result in \eqref{eq-dF/dtau-case-II} by making use of the fact in \eqref{eq-Airy-kernel-PII-Hamiltonian} and by setting $\delta=1/6$ and $y\neq 0$ fixed. 

\subsubsection*{The case when  $\alpha_{k}r^{2k+1}+y\in[0, C_{1}\lambda^{-k-1+\delta}]$} 
This part is an outcome of the analysis carried out in
Section \ref{sec:RH-analysis-case-II+}. Similar to the derivation of \eqref{eq-approx-dF/dtau-three-parts-1}, one can show in this case that 
\begin{eqnarray}
    \frac{\partial F}{\partial s}(s;x)&=&\left.\frac{1}{2 \pi i\lambda}\left(e^{\lambda^{\frac{4k+3}{2}}g_2(\eta)\sigma_3}\tilde{P}(\eta)^{-1}\tilde{P}(\eta)'e^{-\lambda^{\frac{4k+3}{2}}g_2(\eta)\sigma_3}\right)_{21}\right|_{\eta \to r}
    \label{eq-tilde-approx-dF/dtau-three-parts} \\
&=&\tilde\Theta_{1}(r)+\tilde\Theta_{2}(r)+\tilde\Theta_{3}(r), \nonumber
\end{eqnarray}
where
\begin{align}
\tilde\Theta_{1}(r)&=\frac{\lambda^{\frac{4k+3}{3}}}{2\pi i\lambda}
f_3'(r;r_0)\left(\Phi_{1}\left(\lambda^{\frac{4k+3}{3}}f_{3}(r;r_{0})\right)\right)_{21},\\
\tilde\Theta_{2}(r)&=\frac{1}{2\pi i\lambda}\left(\tilde{\Phi}_{0}\left(\lambda^{\frac{4k+3}{3}}f_{3}(r;r_{0}))\right)^{-1}\tilde{\Xi}_{1}(r)\tilde{\Phi}_{0}\left(\lambda^{\frac{4k+3}{3}}f_{3}(r;r_{0})\right)\right)_{21},\\
\tilde\Theta_{3}(r)&=\frac{1}{2\pi i\lambda}\left(\tilde{\Phi}_{0}\left(\lambda^{\frac{4k+3}{3}}f_{3}(r;r_{0}))\right)^{-1}\tilde{\Xi}_{2}(r)\tilde{\Phi}_{0}\left(\lambda^{\frac{4k+3}{3}}f_{3}(r;r_{0})\right)\right)_{21}.
\end{align}
Here, $\Phi_1$, $\tilde{\Phi}_0$, $f_3$ are given in \eqref{eq-tilde-Phi-0}, \eqref{eq-tilde-Phi0(x)} and \eqref{eq-def-tilde-f2}, respectively, $\tilde\Xi_{1}(\eta)=\tilde{E}_{2}(\eta)^{-1}\tilde{E}_{2}'(\eta)$ and $\tilde\Xi_{2}(\eta)=\tilde{E}_{2}(\eta)^{-1}R(\eta)^{-1}R'(\eta)\tilde{E}_{2}(\eta)$.

In view of the definition of $\tilde{E}_{2}$ in \eqref{def:tildeE2} and the estimates of $R(r)$ and $R'(r)$ in \eqref{eq-approximation-R}, it follows that
\begin{equation}
\tilde{\Xi}_{1}(r)=\mathcal{O}(1),\qquad \tilde{\Xi}_{2}(r)=\mathcal{O}(\lambda^{-1}), \qquad \lambda\to+\infty.
\end{equation}
This, together with large $\lambda$ asymptotics of $\tilde{\Phi}_{0}(\lambda)$ in \eqref{eq-asym-tilde-Phi0} and the fact that $f_{3}(r;r_{0})=\mathcal{O}(\lambda^{-k-1+\delta})$, implies that
\begin{equation}\label{eq-approx-tilde-Theta-2-3}
\tilde{\Theta}_{2}(r)=\mathcal{O}\left(\lambda^{\frac{k}{6}+\frac{\delta}{2}-1}\right),\qquad \tilde{\Theta}_{3}(r)=\mathcal{O}\left(\lambda^{\frac{k}{6}+\frac{\delta
}{2}-2}\right)
\end{equation}
as $\lambda \to +\infty$, uniformly for $\alpha_{k}r^{2k+1}+y\in[0,C_{1}\lambda^{-k-1+\delta}]$.

From the definition of $f_{3}$ in \eqref{eq-def-tilde-f2}, it is readily seen that 
\begin{equation}\label{eq:changvariable}
    \lambda^{\frac{4k+3}{3}}f_{3}(r;r_{0})=f_{3}(s;s_{0}) \quad \textrm{and} \quad \lambda^{\frac{4k}{3}} f_{3}'(r;r_{0})= f_{3}'(s;s_{0}),
\end{equation}
where $s=r\lambda$ and $s_{0}=r_{0}\lambda$ with $r_{0}$ and $s_{0}$ given in \eqref{eq-def-r0} and \eqref{eq-def-s0} respectively. Inserting \eqref{eq-approx-tilde-Theta-2-3} into \eqref{eq-tilde-approx-dF/dtau-three-parts}, we thus obtain \eqref{eq-dF/dtau-case-II} by \eqref{eq:changvariable}, \eqref{eq-Airy-kernel-PII-Hamiltonian}, and by setting $\delta=1/6$ and $y\neq 0$ fixed.

\section{Asymptotic analysis in the exponential decay region} \label{Sec:Exp}
In this section, we establish the large $\lambda$  asymptotics of $X(\lambda\eta)$ for $\alpha_{k}r^{2k+1}+y\in[C_{1}\lambda^{-k-1+\delta},+\infty)$ with $\delta\in(0,\frac{1}{4})$, which finally leads to the proof of Lemma \ref{lem-dF/dtau-case-III}. As discussed at the end of Section \ref{sec:diffident}, this case corresponds to the exponential decay region. It comes out the first transformation and the global parametrix are exactly the same as those given in Section \ref{sec:RH-analysis-case-II+}. The difference lies in the construction of local parametrix. Due to our assumption on $r$ and $y$, we have $|r-r_0|\geq C_{1}\lambda^{-k-1+\delta}$, which means that $r$ tends to $r_{0}$ slower than the case in Section \ref{sec:RH-analysis-case-II+}. Hence, we construct local parametrices around $\eta=r_0$ and $\eta=r$, respectively. In particular, the one near $\eta=r_0$ involves the Airy parametrix. 

\subsection{Riemann-Hilbert analysis when $\alpha_{k}r^{2k+1}+y\in[C_{1}\lambda^{-k-1+\delta},+\infty)$}
\label{sec:RH-analysis-case-III}

\paragraph{First transformation and global parametrix} 
We define $Q(\eta)=\tilde{P}(\eta)$, where $\tilde{P}$ is given in \eqref{def-P}. It then follows that $Q$ satisfies the following RH problem.

\paragraph{RH problem for $Q$}
\begin{enumerate}
\item[(a)] $Q(\eta)$ is analytic in $\mathbb{C} \setminus \Gamma^{Q}$, where $\Gamma^{Q}=\Gamma_{1}^{Q}\cup \Gamma_{2}^{Q}\cup \Gamma_{3}^{Q}\cup \Gamma_{4}^{Q}$ is shown in Fig \ref{fig:Q}.
\item[(b)] On $\Gamma_{Q}$, the limiting values $Q_{\pm}(\eta)$ exist and satisfy $Q_+(\eta)=Q_+(\eta)v_{Q}(\eta)$, where
\begin{eqnarray}\label{def:vQ}
v_{Q}(\eta):=
\begin{cases}
\begin{pmatrix}
1 & e^{-2\lambda^{\frac{4k+3}{2}} g_2(\eta)}\\
0 & 1
\end{pmatrix}, \quad & \eta \in \Gamma_1^{Q},\\
\begin{pmatrix}
1 & 0\\
e^{2\lambda^{\frac{4k+3}{2}} g_2(\eta)} & 1
\end{pmatrix}, \quad & \eta \in \Gamma_2^{Q} \cup \Gamma_4^{Q},\\
\begin{pmatrix}
0 & 1\\
-1 & 0
\end{pmatrix}, \quad & \eta \in \Gamma_3^{Q},
\end{cases}
\end{eqnarray}
with $g_2$ given in \eqref{def:g_2}. 
\item [(c)] As $\eta\to r$, we have $Q(\eta)=\mathcal{O}(\log{(\eta-r)})$.
\item[(d)] As $\eta \to \infty$, we have
\begin{equation}\label{eq:AsyQ}
Q(\eta)=\left(I + \Boh(\eta^{-1})\right)(\lambda\eta)^{-\frac 14 \sigma_3}N,
\end{equation}
where $N$ is defined in \eqref{eq-def-N}.

\end{enumerate}

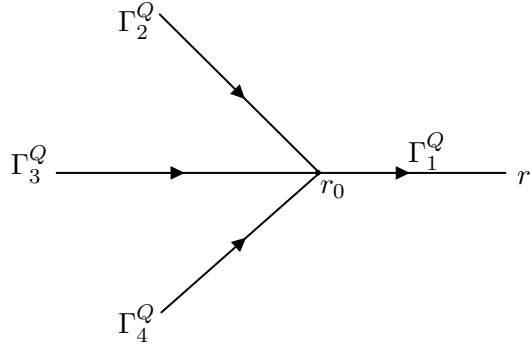
\begin{figure}[t]
\begin{center}
\tikzset{every picture/.style={line width=0.75pt}} 
\begin{tikzpicture}[x=0.75pt,y=0.75pt,yscale=-0.8,xscale=0.8]

\draw    (72,149) -- (223,149) -- (353,149) ;
\draw [shift={(152.5,149)}, rotate = 180] [fill={rgb, 255:red, 0; green, 0; blue, 0 }  ][line width=0.08]  [draw opacity=0] (8.93,-4.29) -- (0,0) -- (8.93,4.29) -- cycle    ;
\draw [shift={(293,149)}, rotate = 180] [fill={rgb, 255:red, 0; green, 0; blue, 0 }  ][line width=0.08]  [draw opacity=0] (8.93,-4.29) -- (0,0) -- (8.93,4.29) -- cycle    ;
\draw    (136.5,49) -- (236.5,149) ;
\draw [shift={(190.04,102.54)}, rotate = 225] [fill={rgb, 255:red, 0; green, 0; blue, 0 }  ][line width=0.08]  [draw opacity=0] (8.93,-4.29) -- (0,0) -- (8.93,4.29) -- cycle    ;
\draw    (137.5,237) -- (236.5,149) ;
\draw [shift={(190.74,189.68)}, rotate = 138.37] [fill={rgb, 255:red, 0; green, 0; blue, 0 }  ][line width=0.08]  [draw opacity=0] (8.93,-4.29) -- (0,0) -- (8.93,4.29) -- cycle    ;
\draw  [fill={rgb, 255:red, 0; green, 0; blue, 0 }  ,fill opacity=1 ] (236.67,148.83) .. controls (236.67,148.19) and (236.15,147.66) .. (235.5,147.66) .. controls (234.85,147.66) and (234.33,148.19) .. (234.33,148.83) .. controls (234.33,149.48) and (234.85,150) .. (235.5,150) .. controls (236.15,150) and (236.67,149.48) .. (236.67,148.83) -- cycle ;

\draw (235,152) node [anchor=north west][inner sep=0.75pt]   [align=left] {$r_0$};
\draw (358,146) node [anchor=north west][inner sep=0.75pt]   [align=left] {$r$};
\draw (109,40) node [anchor=north west][inner sep=0.75pt]   [align=left] {$\Gamma_{2}^{Q}$};
\draw (42,130) node [anchor=north west][inner sep=0.75pt]   [align=left] {$\Gamma_{3}^{Q}$};
\draw (109,230) node [anchor=north west][inner sep=0.75pt]   [align=left] {$\Gamma_{4}^{Q}$};
\draw (290,122) node [anchor=north west][inner sep=0.75pt]   [align=left] {$\Gamma_{1}^{Q}$};

\end{tikzpicture}
\caption{The jump contours of the RH problem for $Q$, where $r_{0}$ is defined in \eqref{eq-def-r0}.}
   \label{fig:Q}
\end{center}
\end{figure}

As mentioned in Section \ref{sec:RH-analysis-case-II+}, we know that $\Re g_2(\eta) < 0$ when $\eta \in \Gamma^{Q}_2 \cup \Gamma^{Q}_4$ and $\Re g_2(\eta) > 0$ when $\eta \in \Gamma^{Q}_1$. Moreover, from \eqref{eq-def-p2(eta)}, we see that $p_{2}(r_{0})>0$ is fixed, then it follows from \eqref{def:g_2} that there exist a constant $M_{4}>0$ such that
\begin{align}
\lambda^{\frac{4k+3}{2}}\Re g_{2}(\eta)&\leq-M_{4}\lambda^{\frac{k}{2}+\frac{3}{2}\delta},\quad \eta \in \Gamma_{2}^{Q} \cup \Gamma_4^{Q}, \label{g2-prop1-1} \\
\lambda^{\frac{4k+3}{2}}\Re g_{2}(\eta)&\geq M_{4}\lambda^{\frac{k}{2}+\frac{3}{2}\delta},\qquad \eta \in \Gamma_{1}^{Q} \label{g2-prop1-2}
\end{align}
provided that $|\eta-r_{0}|\geq\rho_{3}$, where $\rho_{3} = \Boh(\lambda^{-k-1+\delta})$ is a shrinking radius given in \eqref{eq-def-rho3}. Although the above approximation is slightly weaker than those in \eqref{g2-prop1} and \eqref{g2-prop2}, it is sufficient to guarantee 
$v_Q(\eta)\to I$ exponentially fast uniformly for $\eta\in\Gamma_{1}^{Q}\cup\Gamma_{2}^{Q}\cup\Gamma_{4}^{Q}$ and $|\eta-r_{0}|\geq\rho_{3}$. We thus have the global parametrix 
\begin{equation}\label{def-Q-infty}
Q^{(\infty)}(\eta)=(\lambda(\eta-r_{0}))^{-\frac 14 \sigma_3} N, \qquad \eta\in\mathbb{C}\setminus (-\infty, r_0],
\end{equation}
which satisfies the asymptotics \eqref{eq:AsyQ} and the jump condition
\begin{equation}
        Q^{(\infty)}_{+}(\eta)=Q^{(\infty)}_{-}(\eta)\left(\begin{matrix}
            0&1\\-1&0
        \end{matrix}\right), \qquad \eta\in(-\infty,r_{0}).
    \end{equation}

\paragraph{Local parametrices near $\eta=r_{0}$ and $\eta=r$}
Since $Q$ and $Q^{(\infty)}$ are not uniformly close to each other near $\eta=r_0$ and  $Q$ admits a logarithm singularity near $\eta=r$, we need to construct local parametrices around these two points. More precisely, let
\begin{equation}\label{eq-def-rho3}
   \rho_{3}=C_{3}\lambda^{-k-1+\delta},
\end{equation}
where  $C_{3}>0$ is a small positive constant, the local parametrices $Q^{(r_{0})}$ and $Q^{(r)}$ will be built in two shrinking disks $U(r_{0};\rho_{3})$ and $U(r;\rho_{3})$, respectively. Note that for large $\lambda$ and appropriately chosen $C_3$,  $U(r_{0};\rho_{3})\cap U(r;\rho_{3})= \emptyset$. 

The local parametrix in $U(r_{0}; \rho_3)$ reads as follows.
\paragraph{RH problem for $Q^{(r_{0})}$}
\begin{enumerate}
\item[(a)] $Q^{(r_{0})}(\eta)$ is analytic in $U(r_{0};\rho_{3}) \setminus \Gamma^{Q}$.
\item[(b)] On $U(r_{0};\rho_{3}) \cap \Gamma^{Q}$, the limiting values $Q^{(r_{0})}_{\pm}(\eta)$ exist and satisfy the jump condition
\begin{equation}
Q_{+}^{(r_{0})}(\eta)=Q_{-}^{(r_{0})}(\eta)\begin{cases}
\begin{pmatrix}
1 & e^{-2\lambda^{\frac{4k+3}{2}} g_2(\eta)}\\
0 & 1
\end{pmatrix}, \quad & \eta \in \Gamma^{Q}_{1}\cap U(r_{0};\rho_{3}),\\
\begin{pmatrix}
1 & 0\\
e^{2\lambda^{\frac{4k+3}{2}} g_2(\eta)} & 1
\end{pmatrix}, \quad & \eta \in (\Gamma^{Q}_{2} \cup \Gamma_{4}^{Q})\cap U(r_{0};\rho_{3}),\\
\begin{pmatrix}
0 & 1\\
-1 & 0
\end{pmatrix}, \quad & \eta \in \Gamma_{3}^{Q} \cap  U(r_{0};\rho_{3}).
\end{cases}
\end{equation}
\item[(c)] As $\lambda \to\infty$, $Q^{(r_{0})}$ satisfies the matching condition 
\begin{equation}
\label{eq-matching-Qr0-Qinfty}
    Q^{(r_{0})}(\eta)= (I+o(1))Q^{(\infty)}(\eta), \qquad \eta \in \partial U(r_{0};\rho_{3}),
\end{equation}
where $Q^{(\infty)}$ is given in \eqref{def-Q-infty}. 
\end{enumerate}
The conformal mapping we used here is the same as that used in Section \ref{sec:RH-analysis-case-II+}, {\it i.e.}
\begin{equation}
f_{3}(\eta;r_{0})=\left(\frac{3}{2} g_2(\eta)\right)^{\frac 23},
\end{equation}
which behaves like
\begin{equation}\label{eq-def-f3(eta)}
f_{3}(\eta;r_{0})=\left(\frac{3p_{2}(r_{0})}{2}\right)^{\frac{2}{3}}(\eta-r_{0}) + \Boh((\eta-r_{0})^2), \qquad \textrm{as} \quad \eta \to r_{0},
\end{equation}
with $p_{2}(r_{0})>0$. The local parametrix $Q^{(r_{0})}(\eta)$ is given by
\begin{equation}\label{def:Q-cs}
Q^{(r_{0})}(\eta)=E_3(\eta) \Phi^{(\Ai)}(\lambda^{\frac{4k+3}{3}} f_{3}(\eta;r_{0}))\exp\left\{\lambda^{\frac{4k+3}{2}}g_2(\eta)\sigma_3\right\},
\end{equation}
where
\begin{equation}
E_3(\eta)=(\lambda(\eta-r_{0}))^{-\frac{1}{4} \sigma_3}e^{-\frac{\pi i}{4} \sigma_3}(\lambda^{\frac{4k+3}{3}} f_{3}(\eta;r_{0}))^{\frac{1}{4} \sigma_3}
\end{equation}
and $\Phi^{(\Ai)}$ is the classical Airy parametrix (cf. \cite{DKMVZ99}). Recall that $\Phi^{(\Ai)}$ satisfies the jump condition
\begin{equation}\label{eq:Airyjump}
\Phi^{(\Ai)}_+(z)=\Phi^{(\Ai)}_-(z)\begin{cases}
\begin{pmatrix}
0 & 1\\
-1 & 0
\end{pmatrix}, \quad & z<0,\\
\begin{pmatrix}
1 & 1\\
0 & 1
\end{pmatrix}, \quad & z>0,\\
\begin{pmatrix}
1 & 0\\
1 & 1
\end{pmatrix}, \quad & z\in e^{\frac{2\pi i}{3}}(0,+\infty), \\
\begin{pmatrix}
1 & 0\\
1 & 1
\end{pmatrix}, \quad & z\in e^{\frac{-2\pi i}{3}}(0,+\infty),
\end{cases}
\end{equation}
and the asymptotic behavior
\begin{align}\label{infty:Ai}
\Phi^{({\Ai})}(z) = \frac{1}{\sqrt{2}} z^{-\frac14 \sigma_3}\begin{pmatrix}
1 & i\\
i & 1
\end{pmatrix}\left(I + \Boh(z^{-\frac 32})\right)e^{-\frac 23 z^{3/2} \sigma_3}, \qquad z\to \infty,
\end{align}
one can check that $Q^{(r_{0})}$ in \eqref{def:Q-cs} indeed solves the RH problem for $Q^{(r_{0})}$ by noting that the prefactor $E_3(\eta)$ is analytic in $U(r_0,\rho_3)$. 

Similarly, we need to solve the following local parametrix in $U(r;\rho_3)$.
\paragraph{RH problem for $Q^{(r)}$}
\begin{enumerate}
    \item [(a)] $Q^{(r)}(\eta)$ is analytic in $U(r;\rho_{3}) \setminus  [r_0,r]$.
    \item [(b)]  On $U(r;\rho_{3})$, the limiting values $Q^{(r)}_{\pm}(\eta)$ exist and satisfy the jump condition 
    \begin{eqnarray}
        Q_{+}^{(r)}(\eta)=Q_{-}^{(r)}(\eta)\left(\begin{matrix}
            1&e^{-2\lambda^{\frac{4k+3}{2}}g_{2}(\eta)}\\0&1
        \end{matrix}\right). 
    \end{eqnarray}
    \item [(c)] As $\eta\to r$, we have $Q^{(r)}(\eta)=\mathcal{O}(\log{(\eta-r)})$.
    \item [(d)] As $\lambda\to+\infty$, $Q^{(r)}$ satisfies the matching condition 
    \begin{eqnarray}
      Q^{(r)}(\eta)=(I+o(1))Q^{(\infty)}(\eta), \qquad \eta \in \partial U(r;\rho_{3}).
    \end{eqnarray}
\end{enumerate}

It is straightforward to check that 
\begin{equation}\label{eq-def-Qr(eta)}
    Q^{(r)}(\eta)=(\lambda(\eta-r_{0}))^{-\frac{\sigma_{3}}{4}}N\left(\begin{matrix}
        1&\mathcal{C}_{1}(\eta,\lambda)\\0&1
    \end{matrix}\right)
\end{equation}
solves the RH problem for $Q^{(r)}$, where 
\begin{equation}   \mathcal{C}_{1}\left(\eta,\lambda\right)=\frac{1}{2\pi i}\int_{\frac{r+r_{0}}{2}}^{r}\frac{e^{-2\lambda^{\frac{4k+3}{2}}g_{2}(\xi)}}{\xi-\eta}d\xi.
\end{equation}

\paragraph{Final transformation} As usual, the final transformation is defined by
\begin{equation}\label{def-R-case-III}
R(\eta)=\begin{cases}
Q(\eta)Q^{(r_{0})}(\eta)^{-1}, \quad & \eta \in U(r_{0}; \rho_3),\\
Q(\eta)Q^{(r)}(\eta)^{-1}, \quad & \eta \in U(r; \rho_3),\\
Q(\eta)Q^{(\infty)}(\eta)^{-1}, \quad & \eta \in \mathbb{C} \setminus \left(U(r_{0}; \rho_3)\cup U(r; \rho_3) \right).
\end{cases}
\end{equation}
Then it is readily seen that $R$ satisfies the following RH problem.
\paragraph{RH problem for $R$}
\begin{itemize}
\item[(a)] $R(\eta)$ is defined and analytic in $\mathbb{C} \setminus \Gamma^R$, where
$\Gamma^R$ is shown in Figure \ref{fig:R-case-III}.
\item[(b)] On $\Sigma^{R}$, the limiting values $R_{\pm}(\eta)$ exist and satisfy the jump condition
\begin{align}
R_+(\eta)=R_-(\eta) v_R(\eta), \qquad \eta\in \Gamma^R,
\end{align}
where
\begin{align}\label{jump:R-cs}
v_R(\eta) = \begin{cases}
Q^{(r_{0})}(\eta)Q^{(\infty)}(\eta)^{-1}, & \qquad \eta \in \partial U(r_{0}; \rho_3),\\
Q^{(r)}(\eta)Q^{(\infty)}(\eta)^{-1}, & \qquad \eta \in \partial U(r; \rho_3),\\
Q^{(\infty)}(\eta)v_{Q}(\eta)Q^{(\infty)}(\eta)^{-1}, & \qquad \eta \in \Gamma^R \setminus (\partial U(r_{0}; \rho_3)\cup \partial U(r; \rho_3)),
\end{cases}
\end{align}
with $v_Q$ given in \eqref{def:vQ}. 
\item[(c)] As $\eta \to \infty$, we have
\begin{align}
R(\eta) = I + \Boh(\eta^{-1}).
\end{align}
\end{itemize}

When $\eta \in \Gamma^R \setminus \partial U(r_{0}; \rho_3)$, we have from \eqref{g2-prop1-1} and \eqref{g2-prop1-2} that
\begin{equation}
v_R (\eta)=I + \Boh(\lambda^{\frac{k-\delta}{2}}e^{-M_{4}\lambda^{\frac{k}{2}+\frac{3}{2}\delta}}), \qquad \lambda \to +\infty,
\end{equation}
for some constant $M_{4}>0$. For $\eta \in \partial U(r_{0}; \rho_3)$, substituting the large $z$ asymptotic expansion of $\Phi^{({\Ai})}(z)$ \eqref{infty:Ai} into \eqref{def:Q-cs} and \eqref{jump:R-cs}, we have
\begin{align}
v_R (\eta)&=Q^{(r_{0})}Q^{(\infty)}(\eta)^{-1} \nonumber \\
&=(\lambda(\eta-r_{0}))^{-\frac{1}{4}\sigma_{3}}\left(I + \Boh(\lambda^{-\frac{4k+3}{2}}f_{3}(\eta;r_{0})^{-\frac{3}{2}})\right)(\lambda(\eta-r_{0}))^{-\frac{1}{4}\sigma_{3}}
\nonumber \\
&=I+\Boh(\lambda^{-2k-2}(\eta-r_{0})^{-2}) =I+\Boh(\lambda^{-2\delta}), \qquad \lambda\to+\infty.
\end{align}
By a standard argument of the small norm RH problem, we conclude that, as $\lambda\to +\infty$,
\begin{align}\label{eq:Rexp-S6}
R(\eta)=I+\Boh(\lambda^{-2\delta}), \quad R'(\eta)=\Boh(\lambda^{k+1-3\delta}),\quad \eta\in \mathbb{C}\setminus \partial U(r_{0};\rho_{3}).
\end{align}

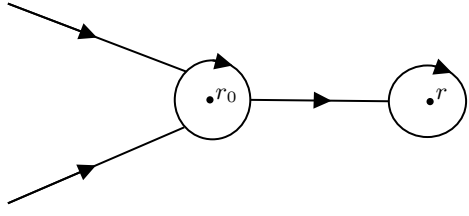
\begin{figure}
\centering

\tikzset{every picture/.style={line width=0.75pt}} 

\begin{tikzpicture}[x=0.75pt,y=0.75pt,yscale=-1,xscale=1]

\draw   (135.83,165.43) .. controls (135.83,154.46) and (144.27,145.57) .. (154.69,145.57) .. controls (165.1,145.57) and (173.54,154.46) .. (173.54,165.43) .. controls (173.54,176.4) and (165.1,185.29) .. (154.69,185.29) .. controls (144.27,185.29) and (135.83,176.4) .. (135.83,165.43) -- cycle ;
\draw    (52.54,117) -- (141.6,151.23) ;
\draw    (52.54,217.47) -- (140.6,179.23) ;
\draw   (242.69,166.14) .. controls (242.69,156.2) and (251.32,148.14) .. (261.97,148.14) .. controls (272.62,148.14) and (281.26,156.2) .. (281.26,166.14) .. controls (281.26,176.08) and (272.62,184.14) .. (261.97,184.14) .. controls (251.32,184.14) and (242.69,176.08) .. (242.69,166.14) -- cycle ;
\draw    (173.54,165.43) -- (242.69,166.14) ;
\draw  [fill={rgb, 255:red, 0; green, 0; blue, 0 }  ,fill opacity=1 ] (152.11,165.43) .. controls (152.11,164.7) and (152.68,164.11) .. (153.4,164.11) .. controls (154.11,164.11) and (154.69,164.7) .. (154.69,165.43) .. controls (154.69,166.16) and (154.11,166.75) .. (153.4,166.75) .. controls (152.68,166.75) and (152.11,166.16) .. (152.11,165.43) -- cycle ;
\draw  [fill={rgb, 255:red, 0; green, 0; blue, 0 }  ,fill opacity=1 ] (262,166.31) .. controls (262,165.58) and (262.58,164.99) .. (263.29,164.99) .. controls (264,164.99) and (264.58,165.58) .. (264.58,166.31) .. controls (264.58,167.03) and (264,167.62) .. (263.29,167.62) .. controls (262.58,167.62) and (262,167.03) .. (262,166.31) -- cycle ;
\draw    (52.54,117) -- (94.27,133.04) ;
\draw [shift={(97.07,134.12)}, rotate = 201.03] [fill={rgb, 255:red, 0; green, 0; blue, 0 }  ][line width=0.08]  [draw opacity=0] (8.93,-4.29) -- (0,0) -- (8.93,4.29) -- cycle    ;
\draw    (52.54,217.47) -- (93.82,199.55) ;
\draw [shift={(96.57,198.35)}, rotate = 156.53] [fill={rgb, 255:red, 0; green, 0; blue, 0 }  ][line width=0.08]  [draw opacity=0] (8.93,-4.29) -- (0,0) -- (8.93,4.29) -- cycle    ;
\draw    (154.69,145.57) -- (161.88,147.4) ;
\draw [shift={(164.79,148.14)}, rotate = 194.28] [fill={rgb, 255:red, 0; green, 0; blue, 0 }  ][line width=0.08]  [draw opacity=0] (8.93,-4.29) -- (0,0) -- (8.93,4.29) -- cycle    ;
\draw    (205.77,165.62) -- (211.11,165.79) ;
\draw [shift={(214.11,165.88)}, rotate = 181.82] [fill={rgb, 255:red, 0; green, 0; blue, 0 }  ][line width=0.08]  [draw opacity=0] (8.93,-4.29) -- (0,0) -- (8.93,4.29) -- cycle    ;
\draw    (265.6,148.62) -- (271.45,150.75) ;
\draw [shift={(274.27,151.78)}, rotate = 200.07] [fill={rgb, 255:red, 0; green, 0; blue, 0 }  ][line width=0.08]  [draw opacity=0] (8.93,-4.29) -- (0,0) -- (8.93,4.29) -- cycle    ;

\draw (154.71,158.49) node [anchor=north west][inner sep=0.75pt]   [align=left] {{\footnotesize {\fontfamily{ptm}\selectfont $r_{0}$}}};
\draw (264.65,158.49) node [anchor=north west][inner sep=0.75pt]   [align=left] {{\fontfamily{ptm}\selectfont {\footnotesize $r$}}};

\end{tikzpicture}

\label{fig:R-case-III}
\caption{The jump contour of the RH problem for $R$ in Section \ref{sec:RH-analysis-case-III}.}
\end{figure}
\subsection{Proof of Lemma \ref{lem-dF/dtau-case-III}} 
\label{sec:case-III-proof}
Since $\zeta=\lambda\eta$, $s=r\lambda$, we find that $\zeta\to s$ yields $\eta\to r$. As mentioned at the beginning of this section, we take $Q(\eta) = \tilde{P}(\eta)$. It then follows from \eqref{eq-tilde-approx-dF/dtau-three-parts} that
\begin{equation}\label{eq-dF/ds-Q}
\begin{split}
\frac{\partial F}{\partial s}(s;x)=&\left.\frac{1}{2 \pi i\lambda}\left(e^{\lambda^{\frac{4k+3}{2}}g_2(\eta)\sigma_3}Q(\eta)^{-1}Q(\eta)'e^{-\lambda^{\frac{4k+3}{2}}g_2(\eta)\sigma_3}\right)_{21}\right|_{\eta \to r}\\
=&\left.\frac{1}{2 \pi i\lambda}e^{-2\lambda^{\frac{4k+3}{2}}g_{2}(\eta)}\tilde{Q}_{21}(\eta)\right|_{\eta\to r},
\end{split}
\end{equation}
where 
\begin{equation*}
\tilde{Q}(\eta)=
Q^{(r)}(\eta)^{-1}R (\eta)^{-1}R (\eta)'Q^{(r)}(\eta)+Q^{(r)}(\eta)^{-1}Q^{(r)}(\eta)'.
\end{equation*}
Making use of the explicit expression of $Q^{(r)}(\eta)$ in \eqref{eq-def-Qr(eta)} and the approximations of $R(\eta),R'(\eta)$ in \eqref{eq:Rexp-S6}, one can easily find that $\tilde{Q}_{21}(r)$ exhibits only algebraic growth as $\lambda\to+\infty$ and $g_{2}(r)>0$. Hence, we derive from \eqref{eq-dF/ds-Q} that
\begin{equation}
\frac{\partial F}{\partial s}(s; x)=\Boh\left(e^{-\lambda^{\frac{4k+3}{2}}g_{2}(r)}\right)
\end{equation}
as $\lambda\to+\infty$. Observing that $x=y\lambda^{2k+1}, s=r\lambda$ and the definitions of $s_{0}, r_{0}$ in \eqref{eq-def-s0}, \eqref{eq-def-r0} respectively, one can easy to find that $s_{0}=\lambda r_{0}$. Moreover, by fixing $y\neq 0$, we can further obtain 
\begin{equation}\label{dF/ds-exponential-case-III}
\frac{\partial F}{\partial s}(s; x)=\Boh\left(e^{-My^{\frac{4k+3}{4k+2}}\lambda^{\frac{4k+3}{2}}(r-r_{0})^{\frac{3}{2}}}\right)
\end{equation}
as $\lambda\to+\infty$ for some fixed constant $M>0$. This implies \eqref{eq-dF/ds-case-III} immediately. With this, the proof of Lemma \ref{lem-dF/dtau-case-III} is concluded.

\appendix

\section{The Painlev\'{e} XXXIV ($\mathrm{P_{34}}$)  parametrix}
\label{sec-appendix-p34-parametrix}

The $\mathrm{P_{34}}$ parametrix used in the present work is defined by the following RH problem.
\subsubsection*{RH problem for $\Phi$}
\begin{enumerate}
\item [(a)] $\Phi(\zeta;x)$ is analytic for $\zeta\in\mathbb{C}\setminus\Gamma$, where $x\in \mathbb{R}$ and $ \Gamma:=\cup_{j=1}^{3}\Gamma_{j}$ is shown in Figure \ref{fig:Phi}. 
\item [(b)] On $\Gamma$, the limiting values $\Phi_{\pm}(\zeta;x)$ exist and satisfy the  jump condition
\begin{equation}\label{eq-jump-Phi}
\Phi_{+}(\zeta;x)=\Phi_{-}(\zeta;x)\begin{cases}
\begin{pmatrix}
1&0\\1&1
\end{pmatrix},& \zeta\in \Gamma_{1}\cup\Gamma_{3},
\\
\begin{pmatrix}
0&1\\-1&0
\end{pmatrix},&\zeta\in \Gamma_2.
\end{cases}
\end{equation} 
\item [(c)] As $\zeta\to\infty$, we have 
\begin{equation}\label{eq-Phi-asym-infty}
\Phi(\zeta;x)=\left(I+\frac{\Phi_{-1}}{\zeta}+\mathcal{O}(\zeta^{-2})\right)\zeta^{-\frac{1}{4}\sigma_{3}}\frac{I+i\sigma_{1}}{\sqrt{2}}e^{-\left(\frac{2}{3}\zeta^{\frac{3}{2}}+x\zeta^{\frac{1}{2}}\right)\sigma_{3}},
\end{equation}
where $\Phi_{-1}$ is independent of $\zeta$. 
\item [(d)] As $\zeta\to 0$, we have
\begin{equation}\label{eq-appendix-Phi-zeto-to-0}
\Phi(\zeta;x)=\Phi^{(0)}(\zeta;x)\left(\begin{matrix}
1&\frac{\log{\zeta}}{2\pi i}\\0&1
\end{matrix}\right)\begin{cases}
I,&\zeta\in \mathrm{I},\\
\left(\begin{matrix}
1&0\\-1&1
\end{matrix}\right),&\zeta\in \mathrm{II},\\
\left(\begin{matrix}
1&0\\1&1
\end{matrix}\right),&\zeta\in \mathrm{III},
\end{cases}
\end{equation}
where regions I--III are shown in Figure \ref{fig:Phi} and $\Phi^{(0)}$ is analytic near the origin satisfying
\begin{equation}\label{eq-Phi-0}
\Phi^{(0)}(\zeta;x)=\Phi_{0}(x)(I+\Phi_{1}(x)\zeta+\mathcal{O}(\zeta^2)), \qquad \zeta \to 0. 
\end{equation}
\end{enumerate}

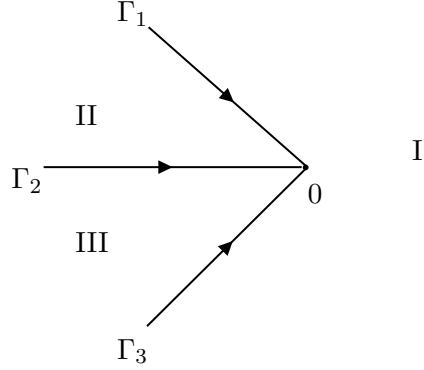
\begin{figure}[t]
\begin{center}
\tikzset{every picture/.style={line width=0.75pt}} 
\begin{tikzpicture}[x=0.75pt,y=0.75pt,scale=0.8]

\draw    (72,149) -- (233,149);
\draw [shift={(152.5,149)}, rotate = 180] [fill={rgb, 255:red, 0; green, 0; blue, 0 }  ][line width=0.08]  [draw opacity=0] (8.93,-4.29) -- (0,0) -- (8.93,4.29) -- cycle    ;
\draw    (136.5,49) -- (236.5,149) ;
\draw [shift={(190.04,102.54)}, rotate = 225] [fill={rgb, 255:red, 0; green, 0; blue, 0 }  ][line width=0.08]  [draw opacity=0] (8.93,-4.29) -- (0,0) -- (8.93,4.29) -- cycle    ;
\draw    (137.5,237) -- (236.5,149) ;
\draw [shift={(190.74,189.68)}, rotate = 138.37] [fill={rgb, 255:red, 0; green, 0; blue, 0 }  ][line width=0.08]  [draw opacity=0] (8.93,-4.29) -- (0,0) -- (8.93,4.29) -- cycle    ;
\draw  [fill={rgb, 255:red, 0; green, 0; blue, 0 }  ,fill opacity=1 ] (236.67,148.83) .. controls (236.67,148.19) and (236.15,147.66) .. (235.5,147.66) .. controls (234.85,147.66) and (234.33,148.19) .. (234.33,148.83) .. controls (234.33,149.48) and (234.85,150) .. (235.5,150) .. controls (236.15,150) and (236.67,149.48) .. (236.67,148.83) -- cycle ;

\draw (235,140) node [anchor=north west][inner sep=0.75pt]   [align=left] {$0$};
\draw (116,43) node [anchor=north west][inner sep=0.75pt]   [align=left] {$\Gamma_3$};
\draw (50,150) node [anchor=north west][inner sep=0.75pt]   [align=left] {$\Gamma_2$};
\draw (116,255) node [anchor=north west][inner sep=0.75pt]   [align=left] {$\Gamma_1$};

\draw (90,190) node [anchor=north west][inner sep=0.75pt]   [align=left] {II};
\draw (90,110) node [anchor=north west][inner sep=0.75pt]   [align=left] {III};
\draw (300,168) node [anchor=north west][inner sep=0.75pt]   [align=left] {I};

\end{tikzpicture}
\caption{The jump contour $\Gamma$ of the RH problem for $\Phi$.}
   \label{fig:Phi}
\end{center}
\end{figure}

The above RH problem is actually a special case of the general $\mathrm{P_{34}}$ parametrix which depends on two parameters. By \cite{ikj2008,ikj2009,XZ2011}, $\Phi$ exists uniquely. Moreover, with $\Phi_{-1}$ given in  \eqref{eq-Phi-asym-infty},  
$$
u(x):=-\frac{x}{2}-i(\Phi_{-1})_{12}'(x)
$$
satisfies the $\mathrm{P_{34}}$ equation
$$
u''(x)=4u(x)^2+2su(x)+\frac{u'(x)^2}{2u(x)}
$$
and is pole-free on the real axis.

It is worthwhile noting that $\Phi$ is also closely related to the Airy determinant. Indeed, by \cite[Equation (2.17)]{CIK10}, it follows that (up to a shift) 
\begin{equation}
\label{eq-differential-identity-P34} 
\frac{d}{dx}\log\det(I-\mathbb{K}_{x}^{\mathrm{Ai}})=\frac{1}{2\pi i}\lim\limits_{\zeta\to 0}\left(\Phi(\zeta;x)^{-1}\frac{\partial}{\partial \zeta}\Phi(\zeta;x)\right)_{21}=\frac{1}{2\pi i}(\Phi_{1}(x))_{21},
\end{equation}
where the second equality follows from \eqref{eq-Phi-0}. On the other hand, it is well-known that 
\begin{equation}\label{eq-Airy-kernel-PII-Hamiltonian}
\log\det(I-\mathbb{K}_{x}^{\mathrm{Ai}})=-\int_{x}^{+\infty}H_{\mathrm{P_{II}}}(\xi)d\xi,
\end{equation}
where $H_{\mathrm{P_{II}}}(x)$ is the Hamiltonian (also known as the Jimbo-Miwa-Okamoto $\sigma$ form) associated with the Hastings-McLeod solution $q$ of Painlev\'{e} II equation \eqref{def:PII}; cf. \cite{FW2001,TW94}. A combination of the above two formulas gives us
\begin{equation}\label{eq-tracy-widom-formula}
\begin{split}
\log\det(I-\mathbb{K}_{x}^{\mathrm{Ai}})=-\int_{x}^{+\infty}\frac{1}{2\pi i}(\Phi_{1}(\xi))_{21}d\xi=-\int_{x}^{+\infty}H_{\mathrm{P_{II}}}(\xi)d\xi.
\end{split}
\end{equation}

We conclude this section by introducing a variant of the $\mathrm{P_{34}}$ parametrix, which will be used in Section  \ref{sec:RH-analysis-case-II+}. It is defined by
\begin{equation}\label{eq-def-tilde-Phi}
\tilde{\Phi}(\zeta;x)=\left(\begin{matrix}
1&0\\ix^2/4&1
\end{matrix}\right)\Phi(\zeta-x;x)
\begin{cases}
I,& \zeta\in \mathrm{I},\\
\left(\begin{matrix}
1&0\\1&1
\end{matrix}\right), & \zeta\in \mathrm{II'},\\
\left(\begin{matrix}
1&0\\-1&1
\end{matrix}\right), &  \zeta\in \mathrm{III'},
\end{cases}
\end{equation}
where regions I, $\mathrm{II'}$ and $\mathrm{III'}$ are illustrated in Figure \ref{fig:tilde{Phi}} and we choose $x>0$. It is then readily seen that $\tilde{\Phi}$ satisfies the following RH problem. 

\subsubsection*{RH problem for $\tilde{\Phi}$}
\begin{enumerate}
\item [(a)] $\tilde{\Phi}(\zeta;x)$ is analytic for $\zeta\in\mathbb{C}\setminus\tilde{\Sigma}$, where $\tilde{\Sigma}:=\cup_{j=1}^4 \tilde{\Gamma}_{j} $ is shown in Figure \ref{fig:tilde{Phi}}.
\item [(b)] On $\tilde{\Sigma}$, the limiting values $\tilde{\Phi}_{\pm}(\zeta;x)$ exist and satisfy the jump condition
\begin{equation}
\tilde{\Phi}_{+}(\zeta;x)=\tilde{\Phi}_{-}(\zeta;x)
\begin{cases}
\left(\begin{matrix}
1&1\\0&1
\end{matrix}\right), &\zeta\in\tilde{\Gamma}_{1},\\
\left(\begin{matrix}
1&0\\1&1
\end{matrix}\right),&\zeta\in\tilde{\Gamma}_{2}\cup \tilde{\Gamma}_{4},\\
\left(\begin{matrix}
0&1\\-1&0
\end{matrix}\right), &\zeta\in\tilde{\Gamma}_{3}.
\end{cases}
\end{equation}

\item [(c)] As $\zeta\to\infty$, we have
\begin{equation}\label{eq-tilde-Phi-asym-infty}
\tilde{\Phi}(\zeta;x)=\left(I+\mathcal{O}(\zeta^{-1})\right)\zeta^{-\frac{1}{4}\sigma_{3}}\frac{I+i\sigma_{1}}{\sqrt{2}}e^{-\frac{2}{3}\zeta^{\frac{3}{2}}\sigma_{3}},
\end{equation}
uniformly for any positive bounded $x$.

\item [(d)] As $\zeta\to x$, we have
\begin{equation}\label{eq-appendix-tilde-Phi-zeto-to-0}
\tilde{\Phi}(\zeta;x)=\tilde\Phi^{(0)}(\zeta;x)\left(\begin{matrix}
1&\frac{\log{(\zeta-x)}}{2\pi i}\\0&1
\end{matrix}\right),
\end{equation}
and
\begin{equation}\label{eq-tilde-Phi-0}
\tilde{\Phi}^{(0)}(\zeta;x)=\tilde{\Phi}_{0}(x)(I+\Phi_{1}(x)(\zeta-x)+\mathcal{O}((\zeta-x)^2)).
\end{equation}
Here,
\begin{equation}\label{eq-tilde-Phi0(x)}
\tilde{\Phi}_{0}(x)=\left(\begin{matrix}
1&0\\ ix^2/4 & 1
\end{matrix}\right)\Phi_{0}(x), 
\end{equation} 
and $\Phi_{i}$, $i=0,1$, are given in \eqref{eq-Phi-0}.

\end{enumerate}

\begin{figure}[t]
\begin{center}
\tikzset{every picture/.style={line width=0.75pt}} 
\begin{tikzpicture}[x=0.75pt,y=0.75pt,yscale=-0.8,xscale=0.8]

\draw    (72,149) -- (223,149) -- (353,149) ;
\draw [shift={(152.5,149)}, rotate = 180] [fill={rgb, 255:red, 0; green, 0; blue, 0 }  ][line width=0.08]  [draw opacity=0] (8.93,-4.29) -- (0,0) -- (8.93,4.29) -- cycle    ;
\draw [shift={(293,149)}, rotate = 180] [fill={rgb, 255:red, 0; green, 0; blue, 0 }  ][line width=0.08]  [draw opacity=0] (8.93,-4.29) -- (0,0) -- (8.93,4.29) -- cycle    ;
\draw    (136.5,49) -- (236.5,149) ;
\draw [shift={(190.04,102.54)}, rotate = 225] [fill={rgb, 255:red, 0; green, 0; blue, 0 }  ][line width=0.08]  [draw opacity=0] (8.93,-4.29) -- (0,0) -- (8.93,4.29) -- cycle    ;
\draw    (137.5,237) -- (236.5,149) ;
\draw [shift={(190.74,189.68)}, rotate = 138.37] [fill={rgb, 255:red, 0; green, 0; blue, 0 }  ][line width=0.08]  [draw opacity=0] (8.93,-4.29) -- (0,0) -- (8.93,4.29) -- cycle    ;
\draw  [fill={rgb, 255:red, 0; green, 0; blue, 0 }  ,fill opacity=1 ] (236.67,148.83) .. controls (236.67,148.19) and (236.15,147.66) .. (235.5,147.66) .. controls (234.85,147.66) and (234.33,148.19) .. (234.33,148.83) .. controls (234.33,149.48) and (234.85,150) .. (235.5,150) .. controls (236.15,150) and (236.67,149.48) .. (236.67,148.83) -- cycle ;
\draw  [dash pattern={on 4.5pt off 4.5pt}]  (253,49) -- (353,149) ;
\draw  [dash pattern={on 4.5pt off 4.5pt}]  (254,237) -- (353,149) ;

\draw (235,152) node [anchor=north west][inner sep=0.75pt]   [align=left] {$0$};
\draw (358,146) node [anchor=north west][inner sep=0.75pt]   [align=left] {$x$};
\draw (388,86) node [anchor=north west][inner sep=0.75pt]   [align=left] {$\mathrm{I}$};
\draw (116,43) node [anchor=north west][inner sep=0.75pt]   [align=left] {$\tilde\Gamma_2$};
\draw (53,142) node [anchor=north west][inner sep=0.75pt]   [align=left] {$\tilde\Gamma_3$};
\draw (116,234) node [anchor=north west][inner sep=0.75pt]   [align=left] {$\tilde\Gamma_4$};
\draw (244,85) node [anchor=north west][inner sep=0.75pt]   [align=left] {$\mathrm{II'}$};
\draw (243,188) node [anchor=north west][inner sep=0.75pt]   [align=left] {$\mathrm{III'}$};
\draw (104,85) node [anchor=north west][inner sep=0.75pt]   [align=left] {$\mathrm{II}$};
\draw (103,188) node [anchor=north west][inner sep=0.75pt]   [align=left] {$\mathrm{III}$};
\draw (295,126) node [anchor=north west][inner sep=0.75pt]   [align=left] {$\tilde\Gamma_1$};

\end{tikzpicture}
\caption{The jump contour $\tilde{\Sigma}$ of the RH problem for $\tilde{\Phi}$.}
   \label{fig:tilde{Phi}}
\end{center}
\end{figure}
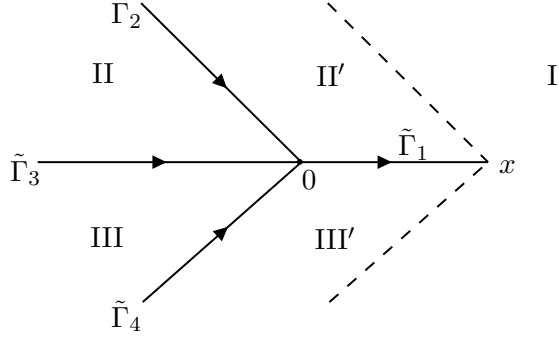

\section{Uniform asymptotics of the $\mathrm{P_{34}}$ parametrix}
\label{sec:uniform-asym-scaled-P34}

The main purpose of this section is to extend the asymptotics in \eqref{eq-Phi-asym-infty} and \eqref{eq-tilde-Phi-asym-infty} for bounded $x$ to the case where both $\zeta$ and $x$ are large. More precisely, as $\zeta \to \infty$, we will establish asymptotics of $\Phi(\zeta;x)$ and $\tilde \Phi(\zeta;x)$ for large negative and positive $x$, respectively. They are crucial in the Riemann-Hilbert analysis in Section \ref{sec:RH-analysis-case-II-} and \ref{sec:RH-analysis-case-II+}. The Riemann-Hilbert analysis for the $\mathrm{P_{34}}$ parametrix has already been documented in \cite{BCI-2016,ikj2009, XZ2011}, but the specific results we enumerated are not explicitly presented in the literature. To facilitate the reader’s understanding, we state the results in the form of lemmas and provide sketch of proofs.


\begin{lemma}\label{lem-appendix-1}
Let $\Phi$ be the $\mathrm{P_{34}}$ parametrix introduced in Appendix \ref{sec-appendix-p34-parametrix}. 
There exists a fixed constant $L>0$ such that, as $x\to-\infty$, 
\begin{equation}\label{eq-asym-PII-large-xi-negative}
\Phi(\zeta;x)=
\zeta^{-\frac{1}{4}\sigma_{3}}\left(I+\left(\begin{matrix}
    \mathcal{O}(x^{-2}\zeta^{-1}) & \mathcal{O}(x^{-1}\zeta^{-1/2})\\
    \mathcal{O}(x^{-3}\zeta^{-3/2}) & \mathcal{O}(x^{-2}\zeta^{-1})
\end{matrix}\right)\right)\frac{I+i\sigma_{1}}{\sqrt{2}}e^{-\left(\frac{2}{3}\zeta^{\frac{3}{2}}+x\zeta^{\frac{1}{2}}\right)\sigma_{3}},
\end{equation}
uniformly for all $|\zeta|\geq L|x|$.
\end{lemma}

\begin{proof}
Make the scaling $\zeta=|x|\eta$. Let $U(0;\rho)$ be a fixed small disk centered at $\eta=0$ with radius $\rho>0$. We define
\begin{equation}\label{eq-def-appendix-R}
R(\eta):=\begin{cases}
\Phi(|x|\eta;-|x|)e^{|x|^{\frac{3}{2}}\left(\frac{2}{3}\eta^{\frac{3}{2}}-\eta^{\frac{1}{2}}\right)\sigma_{3}}P^{(\infty)}(\eta)^{-1},&\eta \in \mathbb{C}\setminus U(0;\rho),\\
\Phi(|x|\eta;-|x|)e^{|x|^{\frac{3}{2}}\left(\frac{2}{3}\eta^{\frac{3}{2}}-\eta^{\frac{1}{2}}\right)\sigma_{3}}P^{(0)}(\eta)^{-1},&\eta \in U(0;\rho),
\end{cases}
\end{equation}
where 
\begin{equation}
\label{eq-appendix-P-infty}
P^{(\infty)}(\eta)=(|x|\eta)^{-\frac{1}{4}\sigma_{3}}\frac{I+i\sigma_{1}}{\sqrt{2}}
\end{equation}
and 
\begin{equation}
\label{eq-appendix-P-0}
P^{(0)}(\eta)=(|x|\eta)^{-\frac{1}{4}\sigma_{3}}(|x|^3 f(\eta))^{\frac{1}{4}\sigma_{3}}e^{\frac{\pi i}{4}\sigma_{3}}\Phi^{(\mathrm{Bes})}(|x|^3f(\eta))e^{|x|^{\frac{3}{2}}\left(\frac{2}{3}\eta^{\frac{3}{2}}-\eta^{\frac{1}{2}}\right)\sigma_{3}},
\end{equation}
In \eqref{eq-appendix-P-0}, 
\begin{equation}\label{def:f}
f(\eta):=(\eta^{\frac{1}{2}}-\frac{2}{3}\eta^{\frac{3}{2}})^2=\eta(1-\frac{2}{3}\eta)^{2}
\end{equation}
and $\Phi^{(\mathrm{Bes})}$ is the Bessel parametrix characterized by the jump condition  \eqref{eq:PhiBesjump} and asymptotic behaviors \eqref{eq-asym-bessel}--\eqref{eq:asym-Bessel-parametrix-near-0}. 

Since 
\begin{equation}
P^{(\infty)}_+(\eta)=P^{(\infty)}_-(\eta)\begin{pmatrix}
    0 & 1
    \\
    -1 & 0
\end{pmatrix}, \qquad \eta \in (-\infty, 0),
\end{equation}
it is readily seen from \eqref{eq-def-appendix-R} and RH problem for $\Phi$ that $R$ satisfies the following RH problem. 
\subsubsection*{RH problem for $R$}
\begin{enumerate}
\item [(a)] $R$ is analytic in $\mathbb{C}\setminus \Sigma^{R}$, where $\Sigma_{R}$ is illustrated in Figure \ref{fig-R-bessel-appendix-B1}. 
\item [(b)] On $\Sigma^{R}$, the limiting values $R_{\pm}(\eta)$ exist and satisfy the following jump condition $$R_{+}(\eta)=R_{-}(\eta)v_{R}(\eta),$$ where
\begin{equation}\label{eq:vRest}
v_{R}(\eta)=\begin{cases}
I+\mathcal{O}(e^{-Cx^{\frac{3}{2}}}), \quad &\eta\in \Sigma^{R}\setminus(\partial U(0;\rho)),\\
I+\mathcal{O}\left(\frac{1}{|x|}\right),\quad &\eta\in \partial U(0;\rho),
\end{cases}
\end{equation}
for some positive $C$ and large $|x|$.
\item [(c)] As $\eta\to\infty$, we have $R(\eta)=I+\mathcal{O}(\eta^{-1})$. 
\end{enumerate}
From the large $z$ expansion of $\Phi^{(\mathrm{Bes})}$ in $\eqref{eq-asym-bessel}$, we actually have 
\begin{equation}\label{eq:vRexp}
v_{R}(\eta)\sim I+\left(|x|\eta\right)^{-\frac{1}{4}\sigma_{3}}\left(\sum\limits_{j=1}^{\infty}\frac{J_{j}}{(|x|^{3}f(\eta))^{\frac{j}{2}}}\right)\left(|x|\eta\right)^{\frac{1}{4}\sigma_{3}}, \qquad \eta \in \partial U(0;\rho),
\end{equation}
as $|x|\to+\infty$, where $J_{2j}$ and $J_{2j+1}$ are diagonal and anti-diagonal constant matrices, respectively. 

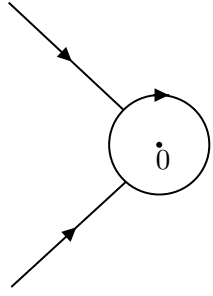
\begin{figure}

\centering

\tikzset{every picture/.style={line width=0.75pt}} 

\begin{tikzpicture}[x=0.75pt,y=0.75pt,yscale=-1,xscale=1]

\draw    (178,82) -- (235.53,136.16) ;
\draw [shift={(209.53,111.69)}, rotate = 223.27] [fill={rgb, 255:red, 0; green, 0; blue, 0 }  ][line width=0.08]  [draw opacity=0] (7.14,-3.43) -- (0,0) -- (7.14,3.43) -- cycle    ;
\draw    (236.53,172.16) -- (179.28,225.09) ;
\draw [shift={(211.8,195.03)}, rotate = 137.24] [fill={rgb, 255:red, 0; green, 0; blue, 0 }  ][line width=0.08]  [draw opacity=0] (7.14,-3.43) -- (0,0) -- (7.14,3.43) -- cycle    ;
\draw   (228.14,153.55) .. controls (228.14,139.74) and (239.34,128.55) .. (253.14,128.55) .. controls (266.95,128.55) and (278.14,139.74) .. (278.14,153.55) .. controls (278.14,167.35) and (266.95,178.55) .. (253.14,178.55) .. controls (239.34,178.55) and (228.14,167.35) .. (228.14,153.55) -- cycle ;
\draw  [fill={rgb, 255:red, 0; green, 0; blue, 0 }  ,fill opacity=1 ] (252.14,153.55) .. controls (252.14,152.99) and (252.59,152.55) .. (253.14,152.55) .. controls (253.69,152.55) and (254.14,152.99) .. (254.14,153.55) .. controls (254.14,154.1) and (253.69,154.55) .. (253.14,154.55) .. controls (252.59,154.55) and (252.14,154.1) .. (252.14,153.55) -- cycle ;
\draw    (250.14,128.55) -- (258.14,128.55) ;
\draw [shift={(257.94,128.55)}, rotate = 180] [fill={rgb, 255:red, 0; green, 0; blue, 0 }  ][line width=0.08]  [draw opacity=0] (7.14,-3.43) -- (0,0) -- (7.14,3.43) -- cycle    ;

\draw (250,155) node [anchor=north west][inner sep=0.75pt]   [align=left] {0};

\end{tikzpicture}
\caption{The jump contour of the RH problem for $R$ in the proof of Lemma \ref{lem-appendix-1}.} 
\label{fig-R-bessel-appendix-B1}
\end{figure}

By a standard argument of the small norm RH problem, it then follows from \eqref{eq:vRest} and \eqref{eq:vRexp} that
\begin{equation}\label{eq:Rexp}
    R(\eta)\sim|x|^{-\frac{1}{4}\sigma_{3}}\left(I+\sum\limits_{j=1}^{\infty}\frac{R_{j}(\eta)}{|x|^{\frac{3j}{2}}}\right)|x|^{\frac{1}{4}\sigma_{3}}
\end{equation}
as $|x|\to+\infty$, where $R_{j}(\eta)$ is analytic in $\mathbb{C}\setminus \partial U(0;\rho)$. In particular, the structures of $J_j$ in \eqref{eq:vRexp} imply that for $\eta\in \mathbb{C}\setminus \overline {U(0;\rho)}$, $R_i$, $i=1,2,3$, takes the following form:
\begin{equation}\label{R1-R2-eta-appendix-B}
    R_{1}(\eta)=\frac{\left(\begin{matrix}
        0& \star \\0&0
    \end{matrix}\right)}{\eta}, \qquad R_{2}(\eta)= \frac{\left(\begin{matrix}
       \star & 0\\0& \star
    \end{matrix}\right)}{\eta}
\end{equation}
and 
\begin{eqnarray}\label{R3-eta-appendix-B}
     R_{3}(\eta)= \frac{\left(\begin{matrix}
        0& \star \\ \star &0
    \end{matrix}\right)}{\eta}+\frac{\left(\begin{matrix}
        0& \star \\0&0
    \end{matrix}\right)}{\eta^2},
\end{eqnarray}
where $\star$ stands for a constant independent of $\eta$ and $x$.

By inverting the transformation in \eqref{eq-def-appendix-R} and noting that $\zeta=|x|\eta$, we obtain \eqref{eq-asym-PII-large-xi-negative} from the combination of \eqref{eq:Rexp}, \eqref{R1-R2-eta-appendix-B}, \eqref{R3-eta-appendix-B} and \eqref{eq-appendix-P-infty}. 
\end{proof}

\begin{lemma}\label{lem-appendix-2}
Let $\tilde{\Phi}$ be a variant of the $\mathrm{P_{34}}$ parametrix defined in \eqref{eq-def-tilde-Phi}. 
There exists a fixed constant $L>1$ such that,
as $x\to+\infty$,
\begin{equation}
\label{eq-tilde-Phi-asym-uniform}
\tilde{\Phi}(\zeta;x)=\left(I+\mathcal{O}\left(\frac{1}{\zeta}\right)\right)\zeta^{-\frac{1}{4}\sigma_{3}}\frac{I+i\sigma_{1}}{\sqrt{2}}e^{-\frac{2}{3}\zeta^{\frac{3}{2}}\sigma_{3}},
\end{equation}
uniformly for all $|\zeta|\geq L|x|$.
\end{lemma}
\begin{proof}
Let $U(0;\rho)$ and $U(1;\rho)$ be two fixed disks centered at $\eta=0$ and $\eta=1$ with radius $\rho\in(0,\frac{1}{2})$, respectively. Similar to \eqref{eq-def-appendix-R}, we define 
\begin{equation}\label{eq-appendix-transform-S-A}
R(\eta)=\begin{cases}
\tilde{\Phi}(x\eta;x)e^{\frac{2}{3}x^{\frac{3}{2}}\eta^{\frac{3}{2}}\sigma_{3}}S^{(\infty)}(\eta)^{-1}, &\eta\in\mathbb{C}\setminus \left(\overline{U(0;\rho)} \cup \overline{U(1;\rho)} \right), 
\\
\tilde{\Phi}(x\eta;x)e^{\frac{2}{3}x^{\frac{3}{2}}\eta^{\frac{3}{2}}\sigma_{3}}S^{(0)}(\eta)^{-1}, &\eta\in U(0;\rho),\\
\tilde{\Phi}(x\eta;x)e^{\frac{2}{3}x^{\frac{3}{2}}\eta^{\frac{3}{2}}\sigma_{3}}S^{(1)}(\eta)^{-1}, &\eta\in U(1;\rho),
\end{cases}
\end{equation}
where
\begin{align}
S^{(\infty)}(\eta)&=(x\eta)^{-\frac{1}{4}\sigma_{3}}\frac{I+i\sigma_{1}}{\sqrt{2}}, \label{eq-appendix-B2-Sinfty}
\\
S^{(0)}(\eta)&=\Phi^{(\mathrm{
Ai})}(x\eta)e^{\frac{2}{3}x^{\frac{3}{2}}\eta^{\frac{3}{2}}\sigma_{3}},
\label{eq-appendix-B2-S0}
\\
S^{(1)}(\eta)&=(x\eta)^{-\frac{1}{4}\sigma_{3}}\frac{I+i\sigma_{1}}{\sqrt{2}}\left(\begin{matrix}
1&\mathcal{C}_{2}(\eta;x)\\ 0&1
\end{matrix}\right).
\label{eq-appendix-B2-S1}
\end{align}
Here, $\Phi^{(\mathrm{Ai})}$ is the Airy parametrix characterized by the jump condition \eqref{eq:Airyjump} and the asymptotic behavior \eqref{infty:Ai}, and 
\begin{equation}\label{eq-def-mathcal-C2}
\mathcal{C}_{2}(\eta;x)=\frac{1}{2\pi i}\int_{\frac{1}{2}}^{1}\frac{e^{-\frac{4}{3}(x\xi)^{\frac{3}{2}}}}{\xi-\eta}d\xi, 
\qquad \eta\in\mathbb{C}\setminus \Big[\frac{1}{2},1 \Big].
\end{equation}
The functions $S^{(\infty)}$ and $S^{(i)}$, $i=0,1$, serve as global and local parametrices for 
$\tilde{\Phi}(x\eta;x)e^{\frac{2}{3}x^{\frac{3}{2}}\eta^{\frac{3}{2}}\sigma_{3}}$, respectively. 

It is then straightforward to check from RH problem for $\tilde{\Phi}$ that $R$ satisfies the following RH problem.
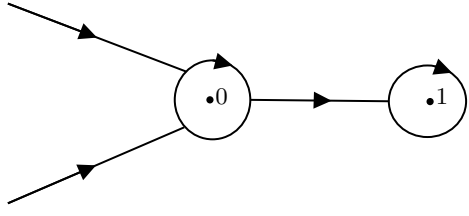
\begin{figure}
\centering

\tikzset{every picture/.style={line width=0.75pt}} 

\begin{tikzpicture}[x=0.75pt,y=0.75pt,yscale=-1,xscale=1]

\draw   (135.83,165.43) .. controls (135.83,154.46) and (144.27,145.57) .. (154.69,145.57) .. controls (165.1,145.57) and (173.54,154.46) .. (173.54,165.43) .. controls (173.54,176.4) and (165.1,185.29) .. (154.69,185.29) .. controls (144.27,185.29) and (135.83,176.4) .. (135.83,165.43) -- cycle ;
\draw    (52.54,117) -- (141.6,151.23) ;
\draw    (52.54,217.47) -- (140.6,179.23) ;
\draw   (242.69,166.14) .. controls (242.69,156.2) and (251.32,148.14) .. (261.97,148.14) .. controls (272.62,148.14) and (281.26,156.2) .. (281.26,166.14) .. controls (281.26,176.08) and (272.62,184.14) .. (261.97,184.14) .. controls (251.32,184.14) and (242.69,176.08) .. (242.69,166.14) -- cycle ;
\draw    (173.54,165.43) -- (242.69,166.14) ;
\draw  [fill={rgb, 255:red, 0; green, 0; blue, 0 }  ,fill opacity=1 ] (152.11,165.43) .. controls (152.11,164.7) and (152.68,164.11) .. (153.4,164.11) .. controls (154.11,164.11) and (154.69,164.7) .. (154.69,165.43) .. controls (154.69,166.16) and (154.11,166.75) .. (153.4,166.75) .. controls (152.68,166.75) and (152.11,166.16) .. (152.11,165.43) -- cycle ;
\draw  [fill={rgb, 255:red, 0; green, 0; blue, 0 }  ,fill opacity=1 ] (262,166.31) .. controls (262,165.58) and (262.58,164.99) .. (263.29,164.99) .. controls (264,164.99) and (264.58,165.58) .. (264.58,166.31) .. controls (264.58,167.03) and (264,167.62) .. (263.29,167.62) .. controls (262.58,167.62) and (262,167.03) .. (262,166.31) -- cycle ;
\draw    (52.54,117) -- (94.27,133.04) ;
\draw [shift={(97.07,134.12)}, rotate = 201.03] [fill={rgb, 255:red, 0; green, 0; blue, 0 }  ][line width=0.08]  [draw opacity=0] (8.93,-4.29) -- (0,0) -- (8.93,4.29) -- cycle    ;
\draw    (52.54,217.47) -- (93.82,199.55) ;
\draw [shift={(96.57,198.35)}, rotate = 156.53] [fill={rgb, 255:red, 0; green, 0; blue, 0 }  ][line width=0.08]  [draw opacity=0] (8.93,-4.29) -- (0,0) -- (8.93,4.29) -- cycle    ;
\draw    (154.69,145.57) -- (161.88,147.4) ;
\draw [shift={(164.79,148.14)}, rotate = 194.28] [fill={rgb, 255:red, 0; green, 0; blue, 0 }  ][line width=0.08]  [draw opacity=0] (8.93,-4.29) -- (0,0) -- (8.93,4.29) -- cycle    ;
\draw    (205.77,165.62) -- (211.11,165.79) ;
\draw [shift={(214.11,165.88)}, rotate = 181.82] [fill={rgb, 255:red, 0; green, 0; blue, 0 }  ][line width=0.08]  [draw opacity=0] (8.93,-4.29) -- (0,0) -- (8.93,4.29) -- cycle    ;
\draw    (265.6,148.62) -- (271.45,150.75) ;
\draw [shift={(274.27,151.78)}, rotate = 200.07] [fill={rgb, 255:red, 0; green, 0; blue, 0 }  ][line width=0.08]  [draw opacity=0] (8.93,-4.29) -- (0,0) -- (8.93,4.29) -- cycle    ;

\draw (154.71,158.49) node [anchor=north west][inner sep=0.75pt]   [align=left] {{\footnotesize {\fontfamily{ptm}\selectfont $0$}}};
\draw (264.65,158.49) node [anchor=north west][inner sep=0.75pt]   [align=left] {{\fontfamily{ptm}\selectfont {\footnotesize $1$}}};

\end{tikzpicture}

\caption{The jump contour of the RH problem for $R$ in the proof of Lemma \ref{lem-appendix-2}.}
\label{fig:R-appendixB-2}
\end{figure}

\subsubsection*{RH problem for $R$}
\begin{enumerate}
\item [(a)] $R(\eta)$ is analytic in $\mathbb{C}\setminus \Sigma^{R}$, where $\Sigma^{R}$ is shown in Figure \ref{fig:R-appendixB-2}.
\item [(b)] On $\Sigma^{R}$, the limiting values $R_{\pm}(\eta)$ exist and satisfy the jump condition $$R_{+}(\eta)=R_{-}(\eta)v_{R}(\eta),$$ where
\begin{equation}
v_{R}(\eta)=\begin{cases}
I+\mathcal{O}(e^{-Cx^{\frac{3}{2}}}), \quad &\eta\in \Sigma^{R}\setminus(\partial U(0;\rho)),\\
I+\mathcal{O}\left(\frac{1}{|x|}\right),\quad &\eta\in \partial U(0;\rho),
\end{cases}
\end{equation}
for some positive $C$ and large $|x|$.
\item [(c)] As $\eta\to\infty$, we have $R(\eta)=I+\mathcal{O}(\eta^{-1})$ 
\end{enumerate}
The small norm RH problem implies that there exists a positive constant $L>1$ such that $R(\eta)=I+\mathcal{O}\left(1/(x\eta)\right)$ as $x\to\infty$ uniformly for all $|\eta|\geq L$. By inverting the transformations \eqref{eq-appendix-transform-S-A} and noting that $\zeta=x\eta$, we obtain \eqref{eq-tilde-Phi-asym-uniform}.
\end{proof}

\begin{remark}
A combination of \eqref{eq-tilde-Phi-asym-uniform} and \eqref{eq-def-tilde-Phi} yields 
\begin{equation}
\label{eq-Phi-asym-uniform-positive}
\Phi(\zeta;x)=\left(\begin{matrix}
    1&0\\-ix^2/4 & 1
\end{matrix}\right)
\left(I+\mathcal{O}\left(\frac{1}{\zeta}\right)\right)(\zeta+x)^{-\frac{1}{4}\sigma_{3}}\frac{I+i\sigma_{1}}{\sqrt{2}}e^{-\frac{2}{3}(\zeta+x)^{\frac{3}{2}}\sigma_{3}}.
\end{equation}
as $x\to+\infty$ uniformly for all $|\zeta|\geq L|x|$. It is much different from the asymptotic behavior of $\Phi(\zeta;x)$ as $x\to-\infty$
given in \eqref{eq-asym-PII-large-xi-negative}. 
\end{remark}

From the proofs of Lemmas \ref{lem-appendix-1} and \ref{lem-appendix-2}, we also have the following asymptotics for $\Phi_{0}$ and $\tilde{\Phi}_{0}$ defined in \eqref{eq-Phi-0} and \eqref{eq-tilde-Phi-0} respectively. 

\begin{lemma}
As $x\to+\infty$, we have
\begin{equation}\label{eq-approx-Phi0(-xi)}
\Phi_{0}(-x)=\left(I+\mathcal{O}\left(\frac{1}{x}\right)\right)(x\pi)^{\frac{1}{2}\sigma_{3}} \left(\begin{matrix}
1&\varphi_{0}(x)\\0&1\end{matrix}\right),
\end{equation}
and 
\begin{equation}\label{eq-asym-tilde-Phi0}
\tilde{\Phi}_{0}(x)=\left(I+\mathcal{O}\left(\frac{1}{x}\right)\right)x^{-\frac{1}{4}\sigma_{3}}\frac{I+i\sigma_{1}}{\sqrt{2}}  e^{-\frac{2}{3}x^{\frac{3}{2}}\sigma_{3}}\left(\begin{matrix}
1&\tilde{\varphi}_{0}(x)\\0&1\end{matrix}\right),
\end{equation}
where $\varphi_{0}(x)$ and $\tilde{\varphi}_{0}(x)$ are two functions of $x$ whose explicit expressions are not needed.
\end{lemma}

\begin{proof}
To show \eqref{eq-approx-Phi0(-xi)}, we see from 
\eqref{eq-appendix-P-0} and \eqref{eq-def-appendix-R} that for $\eta\in U(0,\rho)$,
\begin{align}
\Phi(x\eta;-x)&=R(\eta)P^{(0)}(\eta)e^{-x^{\frac{3}{2}}\left(\frac{2}{3}\eta^{\frac{3}{2}}-\eta^{\frac{1}{2}}\right)\sigma_{3}} \nonumber \\
&=R(\eta)(x\eta)^{-\frac{1}{4}\sigma_{3}}(x^3 f(\eta))^{\frac{1}{4}\sigma_{3}}e^{\frac{\pi i}{4}\sigma_{3}}\Phi^{(\mathrm{Bes})}(x^3f(\eta)).
\end{align}
In view of the asymptotic behavior of $\Phi^{(\mathrm{Bes})}$ in \eqref{eq:asym-Bessel-parametrix-near-0} and \eqref{def:f}, one has  
\begin{equation}
\Phi^{(\mathrm{Bes})}(x^{3}f(\eta))\sim \left(\frac{\sqrt{\pi}}{e^{\frac{\pi i}{4}}}\right)^{\sigma_{3}}\left[\left(\begin{matrix}
1&\frac{\gamma_{\mathrm{E}}-\log{2}}{\pi i}\\0&1\end{matrix}\right)+\mathcal{O}\left(x^3f(\eta)\right)\right]\left(\begin{matrix}
1&\frac{\log{(x^3f(\eta))}}{2\pi i}\\0&1\end{matrix}\right), \quad \eta\to 0.
\end{equation}
This, together with the asymptotic behavior of $\Phi(\zeta;x)$ as $\zeta\to 0$ in \eqref{eq-appendix-Phi-zeto-to-0} and \eqref{eq:Rexp}, implies that as $x\to+\infty$,
\begin{equation}
\begin{split}
\Phi_{0}(-x)=&R(0)(x\pi)^{\frac{1}{2}\sigma_{3}}\left(\begin{matrix}
1&\frac{\gamma_{\mathrm{E}}+\log(x/2)}{\pi i}\\0&1\end{matrix}\right)=\left(I+\mathcal{O}\left(\frac{1}{x}\right)\right)(x\pi)^{\frac{1}{2}\sigma_{3}}\left(\begin{matrix}
1&\frac{\gamma_{\mathrm{E}}+\log(x/2)}{\pi i}\\0&1\end{matrix}\right),
\end{split}
\end{equation}
which is \eqref{eq-approx-Phi0(-xi)}. 

To show \eqref{eq-asym-tilde-Phi0}, we see from   \eqref{eq-appendix-B2-S1} and \eqref{eq-appendix-transform-S-A} that for $\eta\in U(1;\rho)$,
\begin{equation}
\tilde{\Phi}(x\eta;x)=R(\eta)(x\eta)^{-\frac{1}{4}\sigma_{3}}\frac{I+i\sigma_{1}}{\sqrt{2}}\left(\begin{matrix}
1&\mathcal{C}_{2}(\eta;x)\\0&1
\end{matrix}\right)e^{-\frac{2}{3}(x\eta)^{\frac{3}{2}}\sigma_{3}},
\end{equation}
where $\mathcal{C}_{2}(\eta;x)$ is defined in \eqref{eq-def-mathcal-C2}.
Comparing the above formula with \eqref{eq-appendix-tilde-Phi-zeto-to-0} and taking the limit $\eta\to 1$, we have
\begin{equation}\label{eq:tildephi0}
\tilde{\Phi}_{0}(x)=R(1)x^{-\frac{1}{4}\sigma_{3}}\frac{I+i\sigma_{1}}{\sqrt{2}} e^{-\frac{2}{3}x^{\frac{3}{2}}\sigma_{3}}\left(\begin{matrix}
1&\tilde{\varphi}_{0}(x)\\0&1\end{matrix}\right).
\end{equation}
In the above formula, we need the definition of $\mathcal{C}_{2}(\eta;x)$ in \eqref{eq-def-mathcal-C2} and its property $\mathcal{C}_{2}(\eta;x) = \frac{e^{- \frac{4}{3} x^{\frac{3}{2}}}}{2 \pi i} \log(\eta - 1) +\mathcal{O}(1)$ as $\eta \to 1$. This implies the limit 
\begin{equation}
    \tilde{\varphi}_{0}(x)=\lim\limits_{\eta\to 1}\left(e^{\frac{4}{3}(x\eta)^{\frac{3}{2}}}\mathcal{C}_{2}(\eta;x)-\frac{\log(x(\eta-1))}{2\pi i}\right)
\end{equation}
exists. Inserting the fact $R(1)=I+\mathcal{O}(x^{-1})$ as $x\to+\infty$ into \eqref{eq:tildephi0}, we arrive at \eqref{eq-asym-tilde-Phi0} immediately.
\end{proof}

\section{Asymptotics of the Hamiltonian $h$}
\label{PIasy}
Let $h(x)=h(x,t_1,\ldots,t_{2k-1})$ be the Hamiltonian associated with the special solution $\mathsf{q}$ of the Painlev\'{e} I hierarchy $\mathrm{P_{I}^{2k}}$. Following the asymptotics analysis carried out in \cite{Claeys-2012}, it is the aim of this section to prove \eqref{eq-h(x)-asym-expansio}, that is, if $t_{1}=t_{2}=\cdots=t_{2k-1}=0$,
\begin{equation}\label{eq-h(x)-asym-expansion}
h(x)=\frac{(2k+1)}{2(2k+2)}\alpha_k^{-\frac{1}{2k+1}} x^{\frac{2k+2}{2k+1}} +\frac{2k}{24(2k+1)x}+\mathcal{O}\left(x^{-\frac{6k+4}{2k+1}}\right),\qquad x\to\pm\infty,
\end{equation}
where $\alpha_k$ is defined in \eqref{eq-def-Ak}.


Suppose $t_1=\cdots=t_{2k-1}=0$, we begin the analysis with the following RH problem for $S$; see \cite[Section 2.3.3]{Claeys-2012}, where $m$ is replaced by $2k$ in the discussion below.
\paragraph{RH problem for $S$}
\begin{enumerate}
    \item [(a)] $S(\eta)$ is analytic in $\mathbb{C}\setminus\Gamma_{z_{0}}$, where $\Gamma_{z_{0}}:=\Gamma_{1}^{S}\cup\Gamma_{2}^{S}\cup\Gamma_{3}^{S}\cup\Gamma_{4}^{S}$ is shown in Figure \ref{fig:S-appendix-C}.
    \item [(b)] On $\Gamma_{z_0}$, the limiting values $S_{\pm}(\eta)$ exist and satisfy the jump condition 
\begin{equation}\label{eq:Sjump}
S_+(\eta)=S_-(\eta)\begin{cases}
\begin{pmatrix}
1 & e^{-2|x|^{\frac{4k+3}{4k+2}} g(\eta)}\\
0 & 1
\end{pmatrix}, \quad & \eta \in \Gamma^{S}_1,\\
\begin{pmatrix}
1 & 0\\
e^{2|x|^{\frac{4k+3}{4k+2}} g(\eta)} & 1
\end{pmatrix}, \quad & \eta \in \Gamma^{S}_2 \cup \Gamma^{S}_4,\\
\begin{pmatrix}
0 & 1\\
-1 & 0
\end{pmatrix}, \quad & \eta \in \Gamma^{S}_3.
\end{cases}
\end{equation}
\item [(c)] As $\eta\to\infty$, we have 
\begin{equation}
    S(\eta)=\left[I+\frac{B_{1}(x)}{\eta}+\mathcal{O}(\eta^{-2})\right]|x|^{-\frac{1}{8k+4}\sigma_{3}}\eta^{-\frac{1}{4}\sigma_{3}}N,
\end{equation}
where $N$ is given in \eqref{eq-def-N} and $B_1$ is independent of $\eta$. 
\end{enumerate}

\begin{figure}[t]
\begin{center}
\tikzset{every picture/.style={line width=0.75pt}} 
\begin{tikzpicture}[x=0.75pt,y=0.75pt,yscale=-0.8,xscale=0.8]

\draw    (72,149) -- (223,149) -- (353,149) ;
\draw [shift={(152.5,149)}, rotate = 180] [fill={rgb, 255:red, 0; green, 0; blue, 0 }  ][line width=0.08]  [draw opacity=0] (8.93,-4.29) -- (0,0) -- (8.93,4.29) -- cycle    ;
\draw [shift={(293,149)}, rotate = 180] [fill={rgb, 255:red, 0; green, 0; blue, 0 }  ][line width=0.08]  [draw opacity=0] (8.93,-4.29) -- (0,0) -- (8.93,4.29) -- cycle    ;
\draw    (136.5,49) -- (236.5,149) ;
\draw [shift={(190.04,102.54)}, rotate = 225] [fill={rgb, 255:red, 0; green, 0; blue, 0 }  ][line width=0.08]  [draw opacity=0] (8.93,-4.29) -- (0,0) -- (8.93,4.29) -- cycle    ;
\draw    (137.5,237) -- (236.5,149) ;
\draw [shift={(190.74,189.68)}, rotate = 138.37] [fill={rgb, 255:red, 0; green, 0; blue, 0 }  ][line width=0.08]  [draw opacity=0] (8.93,-4.29) -- (0,0) -- (8.93,4.29) -- cycle    ;
\draw  [fill={rgb, 255:red, 0; green, 0; blue, 0 }  ,fill opacity=1 ] (236.67,148.83) .. controls (236.67,148.19) and (236.15,147.66) .. (235.5,147.66) .. controls (234.85,147.66) and (234.33,148.19) .. (234.33,148.83) .. controls (234.33,149.48) and (234.85,150) .. (235.5,150) .. controls (236.15,150) and (236.67,149.48) .. (236.67,148.83) -- cycle ;

\draw (235,152) node [anchor=north west][inner sep=0.75pt]   [align=left] {$z_0$};
\draw (112,43) node [anchor=north west][inner sep=0.75pt]   [align=left] {$\Gamma^{S}_2$};
\draw (48,142) node [anchor=north west][inner sep=0.75pt]   [align=left] {$\Gamma^{S}_3$};
\draw (112,234) node [anchor=north west][inner sep=0.75pt]   [align=left] {$\Gamma^{S}_4$};
\draw (320,124) node [anchor=north west][inner sep=0.75pt]   [align=left] {$\Gamma^{S}_1$};

\end{tikzpicture}
\caption{The jump contour of the RH problem for $S$ in Appendix \ref{PIasy}}
   \label{fig:S-appendix-C}
\end{center}
\end{figure}
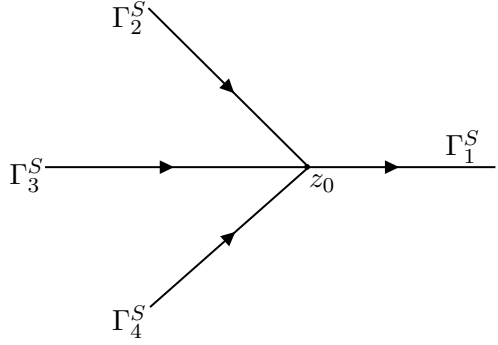

In \eqref{eq:Sjump}, the function $g$ is defined by 
\begin{equation}\label{eq:gfunction-appendix}
g(\eta)=(\eta-z_{0})^{\frac{3}{2}} p\left(\frac{\eta}{z_0}\right),
\end{equation} 
where 
\begin{equation}\label{eq:p}
 p(z)=\frac{4}{4k+3}z_0^{2k}\sum\limits_{j=0}^{2k}\frac{\Gamma\left(j+\frac{3}{2}\right)}{\Gamma\left(j+1\right)\Gamma\left(\frac{3}{2}\right)}z^{2k-j},
\end{equation}
and 
 \begin{equation}\label{eq:zzero}
z_{0}=-\sgn(x)\left[\frac{\Gamma\left(2k+2\right)\Gamma\left(\frac{3}{2}\right)}{2\Gamma\left(2k+\frac{3}{2}\right)}\right]^{\frac{1}{2k+1}}=-\sgn(x)\alpha_k^{-\frac{1}{2k+1}}.
\end{equation}   
On account of Equations (2.4), (2.35), (2.56) and (2.60) in  \cite{Claeys-2012}, it comes out that the Hamiltonian $h$ is related to the above RH problem through the formula
\begin{align}
  h(x)&=\frac{2k+1}{4(k+1)}\alpha_k^{-\frac{1}{2k+1}} |x|^{\frac{2k+2}{2k+1}}+|x|^{\frac{1}{2k+1}}(B_{1}(x))_{12}.  \label{eq:relation-h-B1-appendix}
\end{align}

To establish the large $|x|$ behavior of $B_1$, we need to construct global and local parametrices for the RH problem for $S$. The global parametrix reads
\begin{equation}\label{eq:Pinfty}
P^{(\infty)}(\eta)=(|x|^{\frac{1}{2k+1}}(\eta-z_{0}))^{-\frac 14 \sigma_3} N, \quad \eta\in\mathbb{C}\setminus(-\infty,z_{0}],
\end{equation}
which satisfies the jump condition
\begin{equation}
P^{(\infty)}_+(\eta)=P^{(\infty)}_-(\eta)\begin{pmatrix}
    0 & 1
    \\
    -1 & 0
\end{pmatrix}, \qquad \eta \in (-\infty, z_0). 
\end{equation}
For the local parametrix, we define 
\begin{equation}
  f(\eta)=\left(\frac{3}{2}g(\eta)\right)^{\frac{2}{3}}, 
\end{equation}
where $g$ is given in \eqref{eq:gfunction-appendix}. 
Note that, as $\eta \to z_0$, 
\begin{align}\label{eq:f(eta)-expand-z0}
 f(\eta)&=\left[\frac{3}{2}p(1)(\eta-z_{0})^{\frac{3}{2}}\left(1+\frac{p'(1)}{z_{0}p(1)}(\eta-z_{0})+\mathcal{O}((\eta-z_{0})^2)\right)\right]^{\frac{2}{3}} \nonumber \\
 &=\left(\frac{3}{2}\right)^{\frac{2}{3}}p(1)^{\frac{2}{3}}(\eta-z_{0})+\mathcal{O}((\eta-z_{0})^2),     
\end{align}
it is then readily seen that $f$ is conformal in a small fixed disk $U(z_{0};\rho)$ around $\eta=z_0$. We build the local parametrix by
\begin{equation}\label{eq:P0}
P^{(z_{0})}(\eta)=E(\eta) \Phi^{(\Ai)}(|x|^{\frac{4k+3}{3(2k+1)}} f(\eta))e^{|x|^{\frac{4k+3}{2(2k+1)}}g(\eta)\sigma_3}, \qquad \eta\in U(z_{0}; \rho),
\end{equation}
where
\begin{equation}
E(\eta):=(|x|^{\frac{1}{2k+1}}(\eta-z_{0}))^{-\frac{1}{4} \sigma_3}e^{-\frac{\pi i}{4} \sigma_3}(|x|^{\frac{4k+3}{3(2k+1)}} f(\eta))^{\frac{1}{4}\sigma_3},
\end{equation}
is an analytic prefactor in $U(z_{0}; \rho)$ and $\Phi^{(\Ai)}$ is the Airy parametrix characterized by the jump condition \eqref{eq:Airyjump} and the asymptotic behavior \eqref{infty:Ai}. 

As usual, by setting
\begin{equation}\label{def-R-appendix-C}
R(\eta)=\begin{cases}
S(\eta)P^{(z_{0})}(\eta)^{-1}, \quad & \eta \in U(z_{0}; \rho),\\
S(\eta)P^{(\infty)}(\eta)^{-1}, \quad & \eta \in \mathbb{C} \setminus \overline{U(z_{0}; \rho)},
\end{cases}
\end{equation}
it is straightforward to check that $R$ satisfies the following RH problem.
\paragraph{RH problem for $R$}
\begin{enumerate}
\item [(a)] $R(\eta)$ is analytic in $\mathbb{C}\setminus \Gamma^{R}$, where $\Gamma^{R}:=\Gamma_1^S \cup\Gamma_{2}^{S}\cup\Gamma_{4}^{S}\cup \partial U(z_0;\rho)\setminus U(z_0;\rho)$. 
\item [(b)] On $\Gamma^{R}$, the limiting values $R_{\pm}(\eta)$ exist and satisfy the jump condition $$R_{+}(\eta)=R_{-}(\eta)v_{R}(\eta),$$ 
where
\begin{eqnarray}
v_{R}(\eta)=
\begin{cases}
P^{(\infty)}(\eta)\left(\begin{matrix}
1 & e^{-2|x|^{\frac{4k+3}{4k+2}} g(\eta)}\\
0 & 1
\end{matrix}\right)P^{(\infty)}(\eta)^{-1}, &\eta\in\Gamma_{1}^{S} \setminus \partial U(z_0;\rho), \\
P^{(\infty)}(\eta)\left(\begin{matrix}
1&0\\e^{2|x|^{\frac{4k+3}{2}}g(\eta)}&1
\end{matrix}\right)P^{(\infty)}(\eta)^{-1},  &\eta\in\Gamma_{2}^{S}\cup\Gamma_{4}^{S} \setminus \partial U(z_0;\rho),\\
P^{(z_0)}(\eta)P^{(\infty)}(\eta)^{-1}, &\eta\in \partial U(z_0;\rho).
\end{cases}
\end{eqnarray}
\item [(c)] As $\eta\to\infty$, we have  
\begin{equation}\label{eq:Rexpand}
R (\eta)=I+\frac{R_1(x)}{\eta}+\Boh(\eta^{-2}),
\end{equation}
where
\begin{equation}\label{eq:relation-R1-B1-appendix}
    R_{1}(x)=B_{1}(x)-\frac{z_0}{4}\sigma_{3}.
\end{equation}
\end{enumerate}
As $x\to\pm\infty$,  $v_R(\eta)$ tends to the identity matrix exponentially fast for $\eta \in \Gamma^{R} \setminus \partial U(z_{0}; \rho)$.
For $\eta\in \partial U(z_{0}; \rho)$, we have from \eqref{eq:Pinfty}, \eqref{eq:P0} and \eqref{infty:Ai} that 
\begin{equation}\label{eq:vRexpand}
\begin{split}
v_R (\eta)
=|x|^{-\frac{1}{4(2k+1)}\sigma_3}\left(I +\frac{\Delta_1(\eta)}{|x|^{\frac{4k+3}{2(2k+1)}}}+\frac{\Delta_2(\eta)}{|x|^{\frac{4k+3}{2k+1}}}+ \Boh(|x|^{-\frac{3(4k+3)}{2(2k+1)}})\right)|x|^{\frac{1}{4(2k+1)}\sigma_3}, 
\end{split}
\end{equation}
as $x\to\pm\infty$, where  $\Delta_1$ is an off-diagonal matrix-valued function and $\Delta_2$ is a diagonal matrix-valued function with
\begin{equation}\label{eq:Delta}
(\Delta_1)_{12}=\frac{5}{48}\frac{1}{f(\eta)^{\frac{3}{2}}(\eta-z_0)^{\frac{1}{2}}}.
\end{equation}

By the standard argument for the small norm RH problem, it follows from \eqref{eq:vRexpand} that
\begin{equation}\label{eq:R12-appendix}
(R_1)_{12}=|x|^{-\frac{2(k+1)}{2k+1}}
\Res\left((\Delta_1)_{12}, z_0\right)+\Boh(|x|^{-\frac{6k+5}{2k+1}}).
\end{equation}
In view of \eqref{eq:f(eta)-expand-z0}, we have
\begin{equation}\label{eq:Res-Delta1-12}
\Res\left((\Delta_1)_{12}, z_0\right)=-\frac{5}{72z_0}\frac{p'(1)}{p(1)^2},
\end{equation}
where $p$ is defined in \eqref{eq:p}. To calculate $p(1)$ and $p'(1)$, we rewrite $g(\eta)$ in the form 
\begin{equation}
    g(\eta)=(\eta-z_{0})^{\frac{3}{2}}q(\eta)
\end{equation}
with 
\begin{equation}
q(\eta)=\sum\limits_{j=0}^{2k}q_{j}(\eta-z_{0})^{j}.
\end{equation}
It then follows that 
\begin{equation}\label{eq:relation-p-q}
    p(1)=q_{0}, \qquad p'(1)=z_{0}q_{1}.
\end{equation}
Note that (see \eqref{eq-matching-g2-theta} and \eqref{def:g_2})
\begin{equation} 
\label{eq:d-condition}
\theta(\eta;\sgn(x))-g(\eta)=\mathcal{O}(\eta^{-\frac{1}{2}}),\qquad \eta\to\infty,
\end{equation}
where $\theta(\zeta;x)$ is given in \eqref{def:thetaeta}.
By expanding the left hand side of the above formula in terms of powers of $\eta-z_0$,  
we obtain
\begin{equation}
    \begin{split}\label{eq:value-q0-q1}
q_{0}=\frac{2\Gamma\left(2k+\frac{3}{2}\right)}{\Gamma(2k+1)\Gamma\left(\frac{5}{2}\right)}z_{0}^{2k},\qquad
q_{1}=\frac{2\Gamma\left(2k+\frac{3}{2}\right)}{\Gamma(2k)\Gamma\left(\frac{7}{2}\right)}z_{0}^{2k-1}.  
    \end{split}
\end{equation}
This, together with \eqref{eq:relation-p-q} and  \eqref{eq:Res-Delta1-12}, implies that
\begin{equation}\label{eq:Delta1-12-explicit-value}
\Res((\Delta_{1})_{12},z_{0})=-\frac{5}{72}\frac{2k\cdot\Gamma(2k+1)\Gamma(\frac{5}{2})}{5\Gamma\left(2k+\frac{3}{2}\right)}z_{0}^{-2k-1}=\frac{\sgn(x)}{24}\frac{2k}{2k+1}.
\end{equation}
A combination of \eqref{eq:zzero}, \eqref{eq:relation-h-B1-appendix}, \eqref{eq:relation-R1-B1-appendix}, \eqref{eq:R12-appendix} and \eqref{eq:Delta1-12-explicit-value} finally leads to \eqref{eq-h(x)-asym-expansion}.

\section*{Acknowledgements}
Dan Dai was partially supported by grants from the Research Grants Council of the Hong Kong Special Administrative Region, China [Project No. CityU 11311622, CityU 11306723 and CityU 11301924]. Wen-Gao Long was partially supported by the National Natural Science Foundation of China [Grant No. 12401094], the Natural Science Foundation of Hunan Province [Grant No. 2024JJ5131] and the Outstanding Youth Fund of Hunan Provincial Department of Education [Grant No. 23B0454]. Shuai-Xia Xu was partially supported by the National Natural Science Foundation of China [Grant No. 12431008 and 12371257] and Guangdong Basic and Applied Basic Research Foundation [Grant No. 2022B1515020063]. Lu-Ming Yao was partially supported by the National Natural Science Foundation of China [Grant No. 12401316].  Lun Zhang was partially supported by National Natural Science Foundation of China [Grant No. 12271105] and ``Shuguang Program'' supported by Shanghai Education Development Foundation and Shanghai Municipal Education Commission.


\begin{thebibliography}{99}

\bibitem{BRD08}
J. Baik, R. Buckingham and J. DiFranco, Asymptotics of Tracy-Widom distributions and the
total integral of a Painlev\'{e} II function,  Comm. Math. Phys. 280 (2008), 463--497.


\bibitem{BK24}
F. Benjamin and I. Krasovsky,
Sine-kernel determinant on two large intervals,
Comm. Pure Appl. Math. 77 (2024), 1958--2029.


\bibitem{Bleher99}
P. Bleher and A. Its, Semiclassical asymptotics of orthogonal polynomials, Riemann-Hilbert problem, and
universality in the matrix model, Ann. Math. 150 (1999), 185--266.

\bibitem{BCI-2016}
A. Bogatskiy, T. Claeys and A. Its,
Hankel determinant and orthogonal polynomials for a
Gaussian weight with a discontinuity at the edge, Comm. Math. Phys. 347 (2016), 127--162.

\bibitem{BI14}
T. Bothner and A. Its, Asymptotics of a Fredholm determinant corresponding to the first
bulk critical universality class in random matrix models, Comm. Math. Phys. 328 (2014),
155--202.


\bibitem{BB91}
M.~J. Bowick and E. Br\'{e}zin, Universal scaling of the tail of the density of eigenvalues in random
matrix models, Phys. Lett. B 268 (1991), 21--28.


\bibitem{BMP90}
E. Br\'{e}zin, E. Marinar and G. Parisi, A nonperturbative ambiguity free solution of a
string model, Phys. Lett. B 242 (1990), 35--38.

\bibitem{BS24}
S.-S. Byun and S. Park, 
Large gap probabilities of complex and symplectic spherical ensembles with point charges, preprint arXiv:2405.00386.

\bibitem{CTG21}
M. Cafasso, T. Claeys and M. Girotti, 
Fredholm determinant solutions of the Painlev\'{e}
II hierarchy and gap probabilities of determinantal point processes, Int. Math. Res. Not. 2021 (2021), 2437--2478. 

\bibitem{Ch24}
C. Charlier, 
Large gap asymptotics on annuli in the random normal matrix model, 
Math. Ann. 388 (2024), 3529--3587. 


\bibitem{CLM21a}
C. Charlier, J. Lenells and J. Mauersberger,
Higher order large gap asymptotics at the hard edge for Muttalib-Borodin ensembles,
Comm. Math. Phys. 384 (2021), 829--907. 

\bibitem{CLM21b}
C. Charlier, J. Lenells and J. Mauersberger,
The multiplicative constant for the Meijer-$G$ kernel determinant,
Nonlinearity 34 (2021), 2837--2877.

\bibitem{Claeys-2006-thesis}
T. Claeys, Universality in critical random matrix ensembles and pole-free solutions of Painlev\'{e} equations, Thesis of doctor degree in Katholieke Universiteit Leuven, 2006.


\bibitem{Claeys-2012}
T. Claeys, Pole-free solutions of the first Painlev\'e hierarchy and non-generic critical behavior for the KdV equation,  Phys. D 241 (2012), 2226--2236.

\bibitem{CG09}
T. Claeys and T. Grava, Universality of the break-up profile for the KdV equation in the small dispersion
limit using the Riemann–Hilbert approach, Comm. Math. Phys. 286 (2009), 979--1009. 

\bibitem{CG12}
T. Claeys and T. Grava, 
The KdV hierarchy: universality and a Painlev\'e transcendent, Int. Math. Res. Not. 2012 (2012), 5063--5099.

\bibitem{CIK10}
T. Claeys, A. Its and I. Krasovsky, Higher-order analogues of the Tracy-Widom distribution and the Painlev\'e {II} hierarchy, Comm. Pure Appl. Math. 63 (2010), 362--412.

\bibitem{CV07}
T. Claeys and M. Vanlessen, Universality of a double scaling limit near singular edge
points in random matrix models, Comm. Math. Phys. 273 (2007), 499--532.

\bibitem{Claeys-Vanlessen-2007}
T. Claeys and M. Vanlessen, The existence of a real pole-free solution of the fourth order analogue of the Painlev\'{e} I equation, Nonlinearity 20 (2007), 1163--1184.

\bibitem{CDI22}
I. Corwin, P. Deift and A. Its, Harold Widom's work in random matrix theory,
Bull. Amer. Math. Soc. (N.S.) 59 (2022), 155--173.

\bibitem{Dai-Long}
D. Dai and W.-G. Long,
Asymptotics and total integrals of the $\mathrm{P_{I}^{2}}$ tritronqu\'{e}e solution and its Hamiltonian, SIAM J. Math. Anal. 56 (2024), 5350--5371.


\bibitem{DXZ18}
D. Dai, S.-X. Xu and L. Zhang,
Gap probability at the hard edge for random matrix ensembles with pole singularities in the potential,
SIAM J. Math. Anal. 50 (2018), 2233--2279. 

\bibitem{DXZ21}
D. Dai, S.-X. Xu and L. Zhang, 
Asymptotics of Fredholm determinant associated with the Pearcey kernel, Comm. Math. Phys. 382 (2021), 1769--1809. 

\bibitem{Deift1999}
P. Deift, Orthogonal Polynomials and Random Matrices: A Riemann-Hilbert Approach, Courant Lecture Notes, vol. 3, New York University, 1999.

\bibitem{DG07}
P. Deift and D. Gioev, Universality at the edge of the spectrum for unitary, orthogonal, and symplectic ensembles of random matrices,
Comm. Pure Appl. Math. 60 (2007), 867--910.

\bibitem{DIK13}
P. Deift, A. Its and I. Krasovsky, Toeplitz matrices and Toeplitz determinants under the impetus of the Ising model: some history and some recent results,
Comm. Pure Appl. Math. 66 (2013), 1360--1438.

\bibitem{DIK2008}
P. Deift, A. Its and I. Krasovsky, Asymptotics of the Airy-kernel determinant, Comm.
Math. Phys. 278 (2008), 643--678.

\bibitem{DIKZ07}
P. Deift, A. Its, I. Krasovsky and X. Zhou, The Widom-Dyson constant for the gap probability in random
matrix theory, J. Comput. Appl. Math. 202 (2007), 26--47.


\bibitem{DIZ97}
P. Deift, A. Its and X. Zhou, A Riemann-Hilbert approach to asymptotic problems arising in the
theory of random matrix models, and also in the theory of integrable statistical mechanics,
Ann. Math. 146 (1997), 149--235.

\bibitem{DKV11}
P. Deift, I. Krasovsky and J. Vasilevska, Asymptotics for a determinant with a confluent
hypergeometric kernel, Int. Math. Res. Notices 2011 (2011), 2117--2160.

\bibitem{DKM}
P. Deift, T. Kriecherbauer and K.T-R McLaughlin, New results on the equilibrium
measure for logarithmic potentials in the presence of an external field, J. Approx.
Theory 95 (1998), 388--475.


\bibitem{DKMVZ99}
P. Deift, T. Kriecherbauer,  K. T.-R. McLaughlin, S. Venakides and X. Zhou,  Uniform asymptotics for polynomials orthogonal with respect to varying exponential weights and applications
to universality questions in random matrix theory, Comm. Pure Appl. Math. 52 (1999), 1335--1425.


\bibitem{DX93}
 P.~ Deift and X.~Zhou, 
 A steepest descent method for oscillatory Riemann-Hilbert problems,
\newblock {Ann. Math.} 137 (1993), 295--368.

\bibitem{Dub06}
B. Dubrovin, On Hamiltonian perturbations of hyperbolic systems of conservation laws, II: Universality of critical behavior, Comm. Math. Phys. 267 (2006), 117--139.

\bibitem{Dub08}
B. Dubrovin, On Universality of Critical Behavior in Hamiltonian PDEs, Geometry, Topology, and Mathematical Physics, American Mathematical Society Translation Series 2, vol. 224, 
American Mathematical Society, Providence, (2008), 59--109.

\bibitem{Dub09}
B. Dubrovin, T. Grava and C. Klein,
On universality of critical behavior in the
focusing nonlinear Schr\"{o}dinger equation, elliptic umbilic catastrophe and the
tritronqu\'{e}e solution to the Painlev\'{e}-I equation, J. Nonlinear Sci. 19 (2009), 57--94.

\bibitem{Dyson76}
F. Dyson, Fredholm determinants and inverse scattering problems, Comm. Math. Phys. 47 (1976), 171--183.

\bibitem{Ehr06}
T. Ehrhardt, Dyson's constant in the asymptotics of the Fredholm determinant of the sine kernel, Comm. Math. Phys. 262 (2006), 317--341.

\bibitem{Ehr10}
T. Ehrhardt, The asymptotics of a Bessel-kernel determinant which arises in random matrix
theory, Adv. Math. 225 (2010), 3088--3133.

\bibitem{FW2001}
P.~J. Forrester and N.~S. Witte, Application of the $\tau$-function theory of Painlev\'{e} equations to random matrices: PIV, PII
and the GUE, Commun. Math. Phys. 219 (2001), 357--398.


\bibitem{Gordoa}
R. Gordoa and A. Pickering, 
Nonisospectral scattering problems: a key to integrable
hierarchies, J. Math. Phys. 40 (1999), 5749--5786.

\bibitem{Grava-Kapaev-Klein-2015}
T. Grava, A. Kapaev and C. Klein, On the tritronqu\'{e}e solutions of $\mathrm{P^{2}_{I}}$,  Constr. Approx. 41 (2015), 425--466.


\bibitem{HM}
S.~P. Hastings and J.~B. McLeod, A boundary value problem associated with the
second Painlev\'{e} transcendent and the Korteweg-de Vries equation,
Arch. Rational Mech. Anal. 73 (1980), 31--51.

\bibitem{ikj2008}
A. R. Its, A. B. J. Kuijlaars  and J. \"{O}stensson, Critical edge behavior in unitary random matrix ensembles and the thirty fourth Painlev\'{e}
transcendent, Int. Math. Res. Not.  2008 (2008), rnn017, 67pp.


\bibitem{ikj2009}
A. R. Its, A. B. J. Kuijlaars  and J. \"{O}stensson,  Asymptotics for a special solution of the thirty fourth Painlev\'e equation, Nonlinearity  22 (2009), 1523--1558.


\bibitem{ILP18}
A. Its, O. Lisovyy and A. Prokhorov, Monodromy dependence and connection formulae for isomonodromic tau functions, Duke Math. J. 167 (2018), 1347--1432.

\bibitem{IP16}
A. Its and A. Prokhorov, 
Connection problem for the tau-function of the sine-Gordon reduction of Painlev\'{e}-III equation via the Riemann-Hilbert approach, 
Int. Math. Res. Not. 2016 (2016), 6856--6883.

\bibitem{IP18}
A.~R. Its and A. Prokhorov, 
On some Hamiltonian properties of the isomonodromic tau functions, 
Rev. Math. Phys. 30 (2018), 
1840008, 38 pp.

\bibitem{Kap95}
A.~A. Kapaev, Weakly nonlinear solutions of equation $\mathrm{P_I^2}$, J. Math. Sci. 73 (1995), 468--481.

\bibitem{Kra09}
I. Krasovsky, Large gap asymptotics for random matrices, XVth International Congress
on Mathematical Physics, New Trends in Mathematical Physics, Springer, 2009, 413--419.

\bibitem{Kra04}
I. Krasovsky, Gap probability in the spectrum of random matrices and asymptotics
of polynomials orthogonal on an arc of the unit circle, Int. Math. Res. Not. 2004 (2004),
1249--1272.


\bibitem{KM24}
I. Krasovsky and T.~H. Maroudas, 
Airy-kernel determinant on two large intervals, 
Adv. Math. 440 (2024), Paper No. 109505, 79 pp.

\bibitem{Kud97}
N.~A.~Kudryashov,
The first and second Painlev\'{e} equations of higher order and some relations between them, 
Phys. Lett. A 224 (1997), 353--360.



\bibitem{Kui00}
A.~B.~J. Kuijlaars and  K. T.-R. McLaughlin, Generic behavior of the density of states in random
matrix theory and equilibrium problems in the presence of real analytic external fields,  Comm. Pure Appl. Math. 53 (2000), 736--785.

\bibitem{KMVV}
A.~B.~J. Kuijlaars, K.~T-R.~McLaughlin, W.~ Van Assche and M. Vanlessen, The Riemann-Hilbert approach to strong asymptotics for orthogonal polynomials on $[-1, 1]$, Adv. Math. 188 (2004), 337--398.

\bibitem{KV02}
A.~B.~J. Kuijlaars and M. Vanlessen,  Universality for eigenvalue correlations from the modified Jacobi unitary ensemble, Int. Math. Res. Not. 2002 (2002), 1575--1600.

 


\bibitem{Meh}
M. L. Mehta, Random Matrices, 2nd edn, Academic Press, Boston, 1991.

\bibitem{Mugan}
U. Mugan and F. Jrad, Painlev\'{e} test and the first Painlev\'{e} hierarchy, 
J. Phys. A: Math. Gen. 32 (1999), 7933--7952.


\bibitem{Pastur97}
L. Pastur and M. Shcherbina, Universality of the local eigenvalue statistics for a class of unitary invariant
random matrix ensembles, J. Stat. Phys. 86 (1997), 109--147.


\bibitem{Saff}
E. B. Saff and V. Totik, Logarithmic Potentials with External Fields, Springer-Verlag, New York, 1997.

\bibitem{Shim04}
R. Shimomura, A certain expression for the first Painlev\'{e} hierarchy, Proc. Japan Acad. Ser. A 80   (2004), 105--109.


%
\bibitem{TW94}
C.~A. Tracy and H. Widom,
Level-spacing distributions and the Airy kernel,
Comm. Math. Phys. 159 (1994), 151--174.

\bibitem{Vanlessen-2007}
M. Vanlessen, Strong asymptotics of Laguerre-type orthogonal polynomials and applications
in random matrix theory, Constr. Approx.  25 (2007), 125--175.

\bibitem{Widom}
H. Widom, The strong Szeg\H{o} limit theorem for circular arcs, Indiana Univ. Math. J. 21 (1971), 277--283.

\bibitem{XD}
S.-X. Xu and D. Dai,
Tracy-Widom distributions in critical unitary random
matrix ensembles and the coupled Painlev\'{e} II system,
Comm. Math. Phys. 365 (2019), 515--567. 


\bibitem{XZ2011}
S.-X. Xu and Y. Q. Zhao,
Painlev\'e XXXIV asymptotics of orthogonal polynomials for the Gaussian weight with a jump at the edge, Stud. Appl. Math. 127 (2011), 67--105.

\bibitem{YZ24a}
L.-M. Yao and L. Zhang,
Asymptotics of the hard edge Pearcey determinant,
SIAM J. Math. Anal. 56 (2024), 137--172.

%
%

\end{thebibliography}
\end{document}